\documentclass[runningheads]{llncs}
\usepackage[T1]{fontenc}
\usepackage{graphicx}

\usepackage{tabularx}
\usepackage{multirow}
\usepackage{thmtools}
\usepackage{amsmath}
\usepackage{amssymb}
\usepackage{hyperref}
\usepackage{cleveref}
\usepackage{enumerate}
\usepackage{comment}

\usepackage{adjustbox}

\newcommand{\remove}[1]{}

\newcommand{\defproblem}[3]{
  \vspace{3mm}
\noindent\fbox{
  \begin{minipage}{.95\textwidth}
  \begin{tabular*}{\textwidth}{@{\extracolsep{\fill}}lr} #1  \\ \end{tabular*}
  {\bf{Input:}} #2  \\
  {\bf{Goal:}} #3
  \end{minipage}
  }
  \vspace{2mm}
  }

\newcommand{\cA}{{\mathcal A}}
\newcommand{\BB}{{\mathcal B}}

\newcommand{\FF}{\ensuremath{\mathcal{F}}\xspace}

\newcommand{\OO}{\mathcal{O}}
\newcommand{\PP}{\ensuremath{\mathcal{P}}\xspace}

\newcommand{\RR}{\ensuremath{\mathcal{R}}\xspace}

\newcommand{\nn}{{\mathbb N}}

\newcommand{\rmf}{\textsf{sig}}
\newcommand{\rmg}{\textsf{exts}}

\newcommand{\nka}{${\sf NP \subseteq coNP/poly}$}

\newcommand{\sol}{\sf Sol}

\newcommand{\kpath}{\textsc{Long-Path}}

\newcommand{\enumkpath}{\sc Enum Long-Path}
\newcommand{\enumkpathAll}{\sc Enum Long-Path At Least}
\newcommand{\enumkcycle}{\sc Enum Long-Cycle}
\newcommand{\enumkcycleAll}{\sc Enum Long-Cycle At Least}

\newcommand{\dcut}{\sc $d$-Cut}

\newcommand{\vc}{{\sf vc}} 
\newcommand{\diss}{{\sf diss}} 
\newcommand{\pos}{{\sf pos}} 
\newcommand{\order}{{\sf order}} 
\newcommand{\orderr}{{\sf order1}} 
\newcommand{\orderrr}{{\sf order2}} 
\newcommand{\pOne}{{\sf p1}} 
\newcommand{\pTwo}{{\sf p2}} 
\newcommand{\coc}{{\sf coc}} 
\newcommand{\cvd}{{\sf dtc}} 

\newcommand{\mcp}{{\sf maxcp}}
\newcommand{\vi}{{\sf vi}} 
\newcommand{\ccc}{{\sf cc}}

\newcommand{\sv}[1]{}

\newcommand{\Oh}{\mathcal{O}}
\newcommand{\pdpsk}{pd-ps kernel}

\usepackage{tikz}
\usetikzlibrary{decorations.shapes,decorations.pathreplacing,decorations.pathmorphing}
\usetikzlibrary{arrows,matrix,shapes}
\usetikzlibrary{positioning}
\tikzset{
        stars/.style={star,inner sep=2pt}
      }

\usepackage{amsthm}
\newtheoremstyle{freethm}% <name>
 {3pt}% <Space above>
 {3pt}% <Space below>
 {\itshape}% <Body font>
 {}% <Indent amount>
 {\bfseries}% <Theorem head font>
 {}% <Punctuation after theorem head>
 {5pt}% <Space after theorem headi>
 {\thmname{#1}\thmnumber{ #2}\thmnote{ (#3)}.} % Head spec
 
\theoremstyle{freethm}

\newtheorem{observation}{Observation}
\Crefname{observation}{Observation}{Observations}
\Crefname{figure}{Fig.}{Figs.}

\newif\iflong
\longtrue
     
\usepackage[disable]{todonotes}
\usepackage[ruled,vlined,linesnumbered]{algorithm2e}

\begin{document}
\title{Polynomial-Size Enumeration Kernelizations for Long Path Enumeration }
\titlerunning{Polynomial-Size Enumeration Kernelizations for Long Path Enumeration}
% If the paper title is too long for the running head, you can set
% an abbreviated paper title here
%
\author{Christian Komusiewicz\inst{1}\orcidID{0000-0003-0829-7032} \and 
Diptapriyo Majumdar \inst{2}\thanks{Supported by the Science and Engineering Research Board (SERB) grant SRG/2023/001592.}\orcidID{0000-0003-2677-4648} \and
Frank Sommer\inst{3}\thanks{Supported by the Alexander von Humboldt Foundation.}\orcidID{0000-0003-4034-525X}}
\authorrunning{Christian Komusiewicz, and Diptapriyo Majumdar, and Frank Sommer}
% First names are abbreviated in the running head.
% If there are more than two authors, 'et al.' is used.
%
\institute{University of Jena, Germany
\email{c.komusiewicz@uni-jena.de} \and
Indraprastha Institute of Information Technology Delhi, New Delhi, India \email{diptapriyo@iiitd.ac.in} \and
TU Wien, Austria  \email{fsommer@ac.tuwien.ac.at}}
\maketitle              % typeset the header of the contribution
\begin{abstract}

%%enum-k-path-abstract
Enumeration kernelization for parameterized enumeration problems was defined by Creignou et al. [Theory Comput. Syst. 2017] and was later refined by Golovach et al. [J. Comput. Syst. Sci. 2022] to polynomial-delay enumeration kernelization.
We consider  {\enumkpath}, the enumeration variant of the {\kpath} problem, from the perspective of enumeration kernelization.
Formally, given an undirected graph $G$ and an integer $k$, the objective of {\enumkpath} is to enumerate all paths of $G$ having exactly $k$ vertices.
We consider the structural parameters vertex cover number, dissociation number, and distance to clique and provide polynomial-delay enumeration kernels of polynomial size for each of these parameters.

%%% Local Variables: 
%%% mode: latex
%%% TeX-master: "Enum-k-Path-Main.tex"
%%% End: 

%\todo[inline]{add keywords}
\end{abstract}

\section{Introduction}
\label{sec:intro}

Kernelization is one of the most important contributions of parameterized complexity to the toolbox for handling hard computational problems~\cite{CyganFKLMPPS15,DowneyF13,FominLSZ19}. The idea in kernelization is to develop polynomial-time algorithms that shrink any instance~$x$ of a parameterized problem~$L$ with parameterization~$\kappa$ to an equivalent instance~$x'$, the size of which is upper-bounded by~$g(\kappa(x))$ for some function~$g$. 
The crucial innovation of kernelization is that the size bound~$g$ allows for a theory of the effectiveness of polynomial-time preprocessing, a goal that seems virtually unreachable in a nonparameterized setting. Accordingly, the aim is to obtain a kernelization with~$g$ growing as slowly as possible.

Originally limited to decision problems, several extensions of kernelization  adapt the notion to optimization~\cite{LPRS17}, counting~\cite{JansenS23,LokshtanovMSbhZ24}, and enumeration problems~\cite{CreignouMMSV17,Damaschke06,GolovachKKL22,gomes2024matchingmulticutalgorithmscomplexity}.
For enumeration problems, the focus of this work, a first attempt at a suitable kernelization definition, was to ask that~$x'$ should not only be equivalent to~$x$ but furthermore that~$x'$ contains \emph{all} solutions of~$x$~\cite{Damaschke06}. This definition has several drawbacks: First,~it is restricted to subset-minimization problems. Second, the parameter is necessarily at least as large as the solution size. Finally, one may hope for kernels only in the case that the total number of solutions is bounded by a function of~$\kappa(x)$. These drawbacks prompted Creignou et al.~\cite{CreignouMMSV17} to give a different definition of enumeration kernelization that asks for two algorithms: The first algorithm is the kernelization; it shrinks the input instance in polynomial time to the \emph{kernel}, an instance whose size is bounded in the parameter. The second algorithm is the \emph{solution-lifting} algorithm which allows to enumerate all solutions of the input instance from the solutions of the kernel. The crucial restriction for the solution-lifting algorithm is that it should be an FPT-delay algorithm, that is, the time spent between outputting consecutive solutions should be bounded by~$f(\kappa(x))\cdot |x|^{\OO(1)}$ for some function~$f$. 
This definition gives the desired property that an enumeration problem has an FPT-delay algorithm if and only if it has an enumeration kernel~\cite{CreignouMMSV17}. 
As discussed by Golovach et al.~\cite{GolovachKKL22}, however, any enumeration problem with an FPT-delay algorithm automatically admits an enumeration kernel of \emph{constant size}. 
The crux of kernelization---providing a measure~$g$ of the effectiveness of data reduction---is undermined by this fact. 
This led Golovach et al.~\cite{GolovachKKL22} to define enumeration kernels whose solution-lifting algorithms are only allowed to have \emph{polynomial} delay. 
This particular running time bound has several desirable consequences~\cite{GolovachKKL22}: 
First, a problem has an enumeration kernel of constant size if and only if it admits a polynomial-delay enumeration algorithm. 
This now directly corresponds to kernelization of decision problems, where a problem admits a kernelization of constant size if and only if it is polynomial-time solvable. Second, a parameterized enumeration problem admits an FPT-delay algorithm if and only if it admits a polynomial-delay enumeration kernel. Hence, the stricter delay requirement for the solution-lifting algorithm still suffices to capture all problems with FPT-delay. The size of this automatically implied kernel is, however, exponential in the parameter. This now makes it interesting to ask for polynomial-delay enumeration kernelization of {\em polynomial size} (henceforth \emph{\pdpsk}). A further informal argument in favor of polynomial-delay enumeration kernels is that the limitation to polynomial delay often results in conceptually simple solution-lifting algorithms. 
Thus, solutions that are not in the kernel are similar to the kernel solutions and the kernel essentially captures all interesting solution types.

Summarizing, a \pdpsk{} provides {(i)} a compact representation of all types of solutions to the input instance, and {(ii)} a guarantee that for every solution of the output instance, the collection of solutions that are not ``contained'' in the kernel can be enumerated easily, that is, in polynomial delay. A further advantage of polynomial-delay enumeration kernels, which we discuss in more detail in \Cref{sec:prelims}, is that they provide a means to bound the total running time that an enumeration algorithm spends on outputs without polynomial delay. In particular, one can observe that the smaller the kernel, the better this running time bound will be. 

By the above, enumeration kernels with polynomial-delay solution-lifting algorithms are a desirable type of algorithms for hard parameterized enumeration problems and obtaining a small kernel size is an interesting goal in the design of these algorithms. A priori it is unclear, however,  whether there are many problems for which we can find such kernelizations: The only currently known \pdpsk{}s are given for enumeration variants of the \textsc{Matching Cut} and {\dcut} problem~\cite{GolovachKKL22,MajumdarR23}, two graph problems that have received some interest but are certainly not among the most famous ones.
Thus, it is open whether there are further examples of \pdpsk{}s, ideally for well-studied hard problems. In this work, we thus consider the following well-known problem as a candidate for application of the framework.

\defproblem{{\enumkpath}}{An undirected graph $G = (V, E)$ and an integer $k$.}{Enumerate all $k$-paths, that is, all paths having exactly $k$~distinct vertices in $G$.}

The corresponding decision version of the {\enumkpath} problem is the NP-complete {\kpath} problem which asks if an undirected graph $G$ has a $k$-path.  
 {\kpath} is well-studied from the perspective of parameterized complexity \cite{BjorklundHKK17,FellowsLMMRS09,FominLPS16,FominLPS17,GoyalMPZ15} and kernelization \cite{BodlaenderDFH09,BodlaenderJK13,Jansen17,JansenPW19}.
In particular, {\kpath} admits a polynomial kernel when parameterized by the vertex cover number, the distance to cluster, and the max-leaf number of the input graph~\cite{BodlaenderJK13}. 
In contrast, {\kpath} admits no polynomial kernel when parameterized by the solution size~$k$~\cite{BodlaenderDFH09} or by the (vertex deletion) distance to outerplanar graphs~\cite{BodlaenderJK13}.

\paragraph{Our Contributions.} The two above-mentioned kernelization hardness results for \kpath{} directly exclude \pdpsk{}s for parameterization by solution size~$k$ or by distance to outerplanar graphs. We thus consider {\enumkpath} parameterized by two parameters that are lower-bounded by the distance to outerplanar graphs: the vertex cover number (${\vc}$) and the dissociation number (${\diss}$) of the input graph.
We also consider the distance to clique (${\cvd}$) of the input graph which is incomparable to ${\vc}$ and $\diss$.
\iflong
We refer to \Cref{fig:hierarchy-parameters} for an illustration of the parameter relations and our results.
\begin{figure}[t]
\centering
	\includegraphics[scale=0.3]{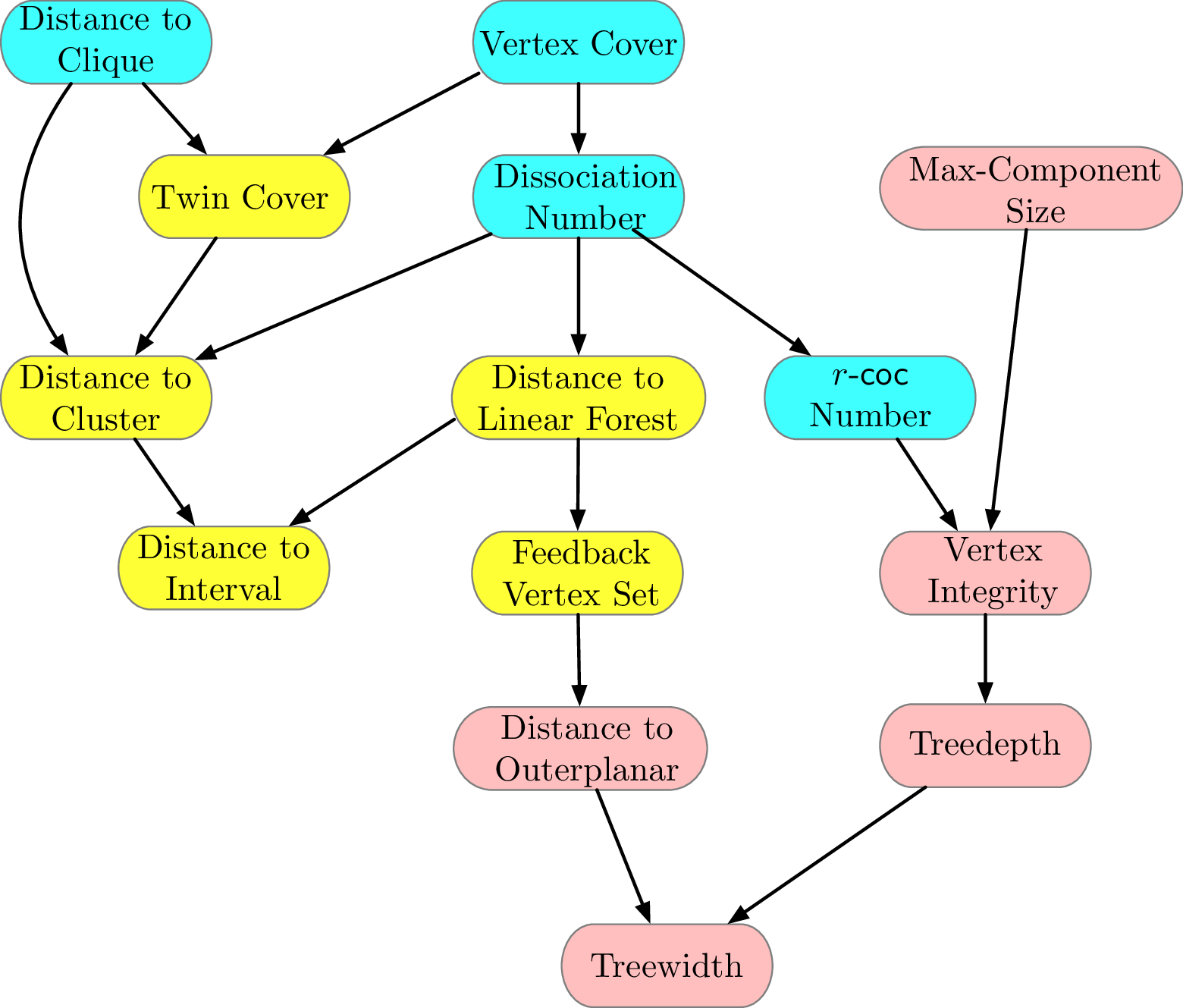}
	\caption{A hierarchy of parameters. A cyan box indicates polynomial sized enumeration kernels. A pink box indicates the non-existence of polynomial sized enumeration kernels. A yellow box indicates open status.
	An edge from a parameter~$\alpha$ to a parameter~$\beta$
below~$\alpha$ means that there is a function~$f$ such that~$\beta \le f(\alpha)$ in every graph.}
\label{fig:hierarchy-parameters}	
\end{figure}
\fi
Our first main result is the following.
%
%Moreover, {\kpath} and {\kcycle} do not admit polynomial kernels when parameterized by the treewidth of the input graphs \cite{BodlaenderJK13}. 
\begin{restatable}{theorem}{thmOne}
\label{thm:k-path-vc-result}
{\enumkpath} parameterized by ${\vc}$ admits an $\OO(n\cdot m\cdot k^2)$-delay enumeration kernel with $\OO({\vc}^2)$ vertices.	
\end{restatable}

The quadratic bound on the number of vertices is the same as for the known kernel for the decision problem~\cite{BodlaenderJK13}. Thus, \Cref{thm:k-path-vc-result} generalizes the best decision kernel to the enumeration problem. 
To prove \Cref{thm:k-path-vc-result} we make use of (i) the ``new expansion lemma''~\cite{FominLLSTZ19,FominLSZ19} and (ii) the enumeration of maximal matchings of a given cardinality with polynomial delay~\cite{KobayashiKW21,KobayashiKW22}.
To the best of our knowledge, this is the first application of the expansion lemma and its variants~\cite{FominLLSTZ19,FominLSZ19,Thomasse10} in the context of enumeration algorithms.

\iflong 
More specifically, to compute the kernel, we use the new expansion lemma on an auxiliary bipartite graph where one partite set corresponds to pairs of vertices from the vertex cover and the other partite set corresponds to vertices of the independent set having degree at least two.
In the kernel, we preserve (1) all vertices from the vertex cover, (2) the vertices of the independent set of degree at least two identified by the new expansion lemma, and (3) a small number of degree-1 vertices of the independent set.
We then define a \emph{signature} of a $k$-path~$P$. 
Roughly speaking the signature is a mapping of all vertices of the vertex cover of~$P$ to their position in~$P$.
The signature is then used to define the notion of \emph{equivalent paths}.
Since in the kernel there might be many equivalent $k$-paths, in the solution lifting algorithm we distinguish whether a $k$-path~$P$ is minimal or not.
If~$P$ is not minimal, we only output~$P$, and otherwise, if~$P$ is minimal, we enumerate all equivalent $k$-paths containing at least one non-kernel vertex.
This ensures that each $k$-path of the input graph is output exactly once.
The expansion computed in the kernelization is essential for guaranteeing the delay of the solution-lifting algorithm: 
whenever we replace an independent set of the kernel with a non-kernel one, the expansion ensures that this sequence can be extended to a $k$-path of~$G$.
\fi

Our next result is a polynomial-delay enumeration kernel when the parameter is ${\diss}$, the dissociation number of the graph. This number is the vertex deletion distance to graphs where each connected component has at most two vertices.   

%Formally, we prove the following result.

\begin{restatable}{theorem}{thmTwo}
\label{theorem:k-path-diss-enum-kernel}
{\enumkpath} parameterized by ${\diss}$ admits an $\OO({\diss}\cdot n)$-delay enumeration kernelization with $\OO({\diss}^3)$ vertices.
\end{restatable}
\iflong
Ideally, we again want to use the new expansion lemma in the kernelization algorithm.
We face a new issue, however: 
The remaining vertices~$I$ of the input graph~$G$ which are not in the dissociation set~$X$ now form components of size~1 and~2.
Furthermore, in a $k$-path~$P$ between two vertices of~$X$ there can be one or two vertices from~$I$. Moreover, vertices contained in size-2 components of~$I$ can also be used as single vertex between two vertices of~$X$ in~$P$.
Consequently, we would need to employ the new expansion lemma for components of~$I$ of size~1 and size~2 individually.
This complicates the application of the expansion lemma. 
Instead, we can achieve a polynomial kernel of cubic size much easier:
We distinguish \emph{rare} and \emph{frequent} vertices of~$I$.
Roughly speaking, a vertex~$u$ is \emph{rare} if there are only a small number of vertices with the same neighborhood with respect to~$X$, and otherwise~$u$ is frequent.
In the kernel we preserve all rare vertices and sufficiently many frequent vertices.
Our solution lifting algorithm then uses similar ideas as the one for~${\vc}$.
\fi
Our approach even gives a kernelization for the more general~$r$-${\coc}$ parameter, which is the vertex deletion distance to graphs whose connected components have size at most~$r$.
\begin{restatable}{corollary}{corTwo}
\label{cor:k-path-r-coc-enum-kernel}
{\enumkpath} parameterized by $r$-${\coc}$ for constant~$r$ admits an~$\OO(r$-${\coc}\cdot n\cdot 2^r)$-delay enumeration kernelization with $\OO((r$-${\coc})^3)$ vertices that can  be computed  in  $\OO(2^r\cdot  r\cdot (r$-${\coc})^2\cdot n)$~time.
\end{restatable}

Ideally, we would like to provide kernels for even smaller parameters such as the vertex integrity or the treedepth~of~$G$. For these parameters, however, already decision kernels are unlikely due to the following folklore result. Here, $\mcp$ denotes the size of a largest connected component. This parameter is at least as large as vertex integrity and treedepth. 
\begin{proposition}
\label{prop:lower-bound-max-cp}
{\enumkpath} has no \pdpsk{} when parameterized by ${\mcp}$ unless {\nka}.
\end{proposition}
\iflong
\begin{proof}
The disjoint union of $2^t$~instances of \textsc{Long-Path} yields a trivial or-cross composition showing that \textsc{Long-Path} has no polynomial size compression for ${\mcp}$ unless {\nka}.
This result transfers to {\enumkpath}.
\end{proof}
\fi

To obtain  \Cref{thm:k-path-vc-result,theorem:k-path-diss-enum-kernel}, we exploited that the length of a longest path can be upper-bounded by $\OO({\vc})$ and $\OO({\diss})$, respectively.
\iflong
This property does not hold for ${\cvd}$, the vertex deletion distance to cliques since the length of a longest path exceeds ${\cvd}$ if the clique is larger than ${\cvd}$.

To resolve the issue of the unbounded path length, our kernelization works in 2 stages:
first the clique and parameter~$k$ are reduced simultaneously and second only clique vertices are removed.
Afterwards, we employ a marking scheme on the remaining clique similar to our marking scheme for~${\diss}$.
For the solution lifting algorithm, we also need to be more careful since the parameter~$k'$ might be much smaller than~$k$.
To ensure that each $k$-path is enumerated exactly once, we use two variants of the solution lifting algorithm. This gives our final main result, a polynomial-delay enumeration kernelization for the parameter ${\cvd}$.
\else
To show that this is not a necessary condition for the existence of  \pdpsk{}s, we consider parameterization by ${\cvd}$, the vertex deletion distance to cliques which can be smaller than~$k$. We obtain the following.
\fi

\begin{restatable}{theorem}{thmThree}
\label{thm-k-path-cvd}
{\enumkpath} parameterized by $\cvd$ admits an $\OO(n\cdot \cvd(G))$-delay enumeration kernelization with $\OO({\cvd}^3)$~vertices.	
\end{restatable}

We also show that our techniques can be adapted to {\enumkcycle}.
\iflong
\paragraph*{Organization of the Paper.} We organize the paper as follows.
In \Cref{sec:prelims}, we introduce the notation and terminologies that are important for this work. 
In particular, we give a formal definition of the enumeration kernelization used in this work.
% Additionally, we also give a short proof why {\enumkpath} admits no polynomial-delay enumeration kernelization of polynomial size when the parameter is ${\mcp}$, i.e. the size of a largest connected component.
After that, in \Cref{sec:k-path-vc,sec:k-path-disoc,sec:k-path-clique}, we provide a \pdpsk{} when the parameters are the vertex cover number ${\vc}$, the dissociation number~${\diss}$, and the distance to clique ${\cvd}$, respectively.
Finally, in \Cref{sec:conclusion}, we conclude with some future research directions.

\else Due to lack of space proofs of results marked with~$(\star)$ are deferred to the full version in the appendix.
\fi

\iflong
\paragraph{Further Related Work.}  
 Wasa \cite{Wasa16arXiv} provided a detailed survey an enumeration problems. 
There is some prior work on enumerating distinct shortest paths, or listing $k$ shortest paths in output-sensitive polynomial time~\cite{Eppstein98,BohmovaHMPSS18,RizziSS14,BirmeleFGMPRS13,FerreiraGRSS14}.
Adamson et al. \cite{AdamsonGM24} provided a polynomial-delay enumeration algorithm to enumerate walks of length exactly $k$.
 Blazej et al.~\cite{BlazejCKSSV22}~explored polynomial kernelizations for variants of {\sc Travelling Salesman Problem} under structural parameterizations.
Recently, there have been several works on  polynomial-delay enumeration algorithms for well-known combinatorial optimization problems.
Kobayashi et al.~\cite{KobayashiKW22}~studied enumeration of large maximal matchings with polynomial delay and polynomial space.
Hoang et al. \cite{HKSS13} studied enumeration of induced paths and induced cycles of length~$k$.

Counting problems, which are related to enumeration in the sense that they extend decision problems, have also seen different formalizations of kernelization. In the compactor technique, formalized by Thurley~\cite{Thur07} based on previous counting algorithms~\cite{NRT05}, there is a kernelization algorithm and an algorithm that computes for each solution of the kernel a number of associated solutions in the input instance. For some counting graph problems, Jansen and~van der Steenhoven~\cite{JansenS23} presented a type of kernelization algorithm that either outputs the total number of solutions or an instance whose size is bounded in the parameter with the same number of solutions.
Independently, Lokshtanov et al. \cite{LokshtanovMSbhZ24} developed a formal framework for kernelization of parameterized counting problems that relies on a kernelization algorithm (called reduce) and a lift algorithm that computes in polynomial time the number of solutions of the input instance from the number of solutions of the kernel. %\todo{CK: THere is also ``Data-compression for Parametrized Counting Problems on Sparse graphs'' but I would say this is sufficiently far away so that we do not need to mention it.}
\fi

\section{Preliminaries}
\label{sec:prelims}

%\paragraph%%preliminaries

%\todo[inline]{define pendant}

\iflong
\paragraph{Sets, Numbers, and Graph Theory.}
For $r \in \nn$, we use $[r]$ to denote the set $\{1,\ldots,r\}$.
We consider simple undirected graphs and use standard graph-theoretic notation~\cite{Diestel-Book}.
For a graph $G = (V, E)$, let $V(G)$ denote the set of \emph{vertices} and $E(G)$ the set of \emph{edges} of $G$.
We write~$n:=|V(G)|$ and~$m:=|E(G)|$.
Given a vertex set~$X \subseteq V(G)$, the graph $G[X]=(X,\{uv \in E(G)\mid u,v\in X\})$ denotes the subgraph of $G$ \emph{induced by $X$}.
 For a set~$X\subseteq V(G)$, we let~$G-X=G[V(G)\setminus X]$ denote the subgraph obtained by deleting the vertices of~$X$.
A vertex~$u$ is {\em isolated} in a graph if no edge of the graph is incident to~$u$.
A {\em pendant} vertex of $G$ is a vertex of degree one in $G$.
A~\emph{$k$-path} (or short, path) in a graph is a finite sequence of distinct vertices $P = (v_1,\ldots,v_k)$ such that $v_iv_{i+1} \in E(G)$ for all~$i\in [k-1]$. We denote by~$V(P)$  the set of vertices in~$P$. 
Given a path $P$ and $v \in V(P)$, we denote by ${\pos}_P(v)$ the index, or informally the position, of $v$ in $P$.
%For a path $P$ with two vertices $u \in V(P)$ but $v \notin V(P)$, we use $(P \setminus \{u\}) \cup \{v\}$ to denote the subgraph 
 A graph $G$ is {\em bipartite} if $V(G)$ can be partitioned into $A \uplus B$ such that for every $uv \in E(G)$, $u \in A$ if and only if $v \in B$. 
 \else
 We use standard graph theoretic notations~\cite{Diestel-Book}. 
 We also use the following tool.
 \fi

\iflong 
\paragraph{Expansion Lemma.}

%\begin{definition}[$q$-Expansion, \cite{CyganFKLMPPS15}]
%\label{defn:q-expansion}
Let $q$ be a positive integer and~$H$ be a bipartite graph with bipartition $(A, B)$.
For $\widehat A \subseteq A$ and $\widehat B \subseteq B$, a set $M \subseteq E(H)$ of edges is called a {\em $q$-expansion} of $\widehat A$ onto $\widehat B$ if 
%the following properties hold.
%\begin{enumerate}[(i)]
	{\bf (i)}  every edge of $M$ has one endpoint in $\widehat A$ and the other endpoint in $\widehat B$,
	{\bf (ii)} every vertex of $\widehat A$ is incident to exactly $q$ edges of $M$, and
	{\bf (iii)}~exactly $q\cdot |\widehat A|$ vertices of $\widehat B$ are incident to some edges of $M$.
%\end{enumerate}	
%\end{definition}
A vertex of $\widehat A\cup \widehat B$ is {\em saturated by $M$} if it is an endpoint of some edge of~$M$, otherwise, it is \emph{unsaturated}.
By definition, all vertices of $\widehat A$ are saturated by~$M$ but~$\widehat B$ may contain some unsaturated vertices.
%Similarly, the vertices of $\widehat B$ that are the endpoints of $M$ are called {\em satu
Observe that a $1$-expansion of $\widehat A$ onto $\widehat B$ is just a matching of $\widehat A$ onto $\widehat B$ that saturates $\widehat A$.
We use the following proposition in our paper.
\fi

\begin{proposition}[New $q$-Expansion Lemma - Lemma 3.2 of \cite{FominLLSTZ19}]
\label{prop:new-q-exp-lemma}
Let $q$ be a positive integer and $H$ be a bipartite graph with bipartition $(A, B)$.
Then, there exists a polynomial-time algorithm that computes $\widehat A \subseteq A$ and $\widehat B \subseteq B$ such that there is a $q$-expansion of $\widehat A$ onto $\widehat B$ such that {\bf (i)}
%\begin{enumerate}[(i)]
%	\item 
	$N_H(\widehat B) \subseteq \widehat A$, and
%	\item 
{\bf (ii)}	$|B \setminus \widehat B| \leq q\cdot |A \setminus \widehat A|$.
%\end{enumerate}
\end{proposition}

\iflong
As a subroutine, we will use the following polynomial-delay algorithm for enumerating all maximum matchings in an auxiliary graph. 

\begin{proposition}[\protect{\cite[paragraph above Theorem~13]{KobayashiKW22}, \cite[Theorem~18]{KobayashiKW21}}]
%\todo{move to graph theory?}
\label{prop:maximum-matching-enumeration}
Given an undirected graph $G$ having $n$ vertices and $m$ edges, all maximum matchings  can be enumerated with $\OO(n\cdot m)$-delay.
\end{proposition}
\fi

\iflong
\paragraph{Graph Parameters.} Given a graph $G$, a set $S \subseteq V(G)$ is a {\em vertex cover} if for every edge $uv\in E(G)$, $u \in S$ or $v \in S$ (or both).
The size of a vertex cover with minimum cardinality is called the {\em vertex cover number} of $G$ and is denoted by ${\vc}(G)$.
A vertex subset $S$ is a {\em dissociation set} \cite{PapadimitriouY82} if $G - S$ has maximum degree at most one.
The {\em dissociation number} of a graph $G$, denoted by ${\diss}(G)$, is the size of a smallest dissociation set of~$G$.
The $r$-${\coc}(G)$ number of a graph is the minimum number of vertex deletions required to obtain a graph where every connected component has size at most~$r$.
Given a graph $G$, let ${\mcp}(G)$, denote the {\em largest component size} of $G$.
In other words, ${\mcp}(G) = \max_{D \in \ccc(G)} |D|$ where ${\ccc}(G)$ denotes the set of all connected components of $G$.
Furthermore, for a graph $G$, the {\em vertex integrity} of $G$, denoted by ${\vi}(G)$ is a measure that indicates how easy it is to break the graph into small pieces.
Formally, ${\vi}(G) = \min_{S \subseteq V(G)} (|S| + \max_{D \in \ccc(G - S)} |D|)$.
A set $S \subseteq V(G)$ is a {\em clique deletion set} if $G - S$ is a clique.
The  {\em distance to clique} of a graph $G$, denoted by ${\cvd}(G)$, is the size of a smallest clique deletion set of~$G$.
\fi

\iflong
\paragraph{Parameterized Complexity and Kernelization.} 
A parameterized problem is a pair $(L,\kappa)$ where~$L\subseteq \Sigma^*$ and~$\kappa:\Sigma^*\to \nn$ is a computable function, called the \emph{parameterization} or simply~\emph{parameter}.
A parameterized problem $(L,\kappa)$ is said to be {\em fixed-parameter tractable} (or FPT in short) if an algorithm $\cA$ that, given $x \in \Sigma^* $, decides in $f(k)\cdot |x|^{\OO(1)}$-time  if $x\in L$, where $f: \nn \rightarrow \nn$ is some computable function. It is natural and common to assume that $\kappa(x)$ is either simply given with $x$ or can be computed in polynomial time from $x$.
Given a graph $G$, our parameterizations ${\vc}(G)$, ${\diss}(G)$, and ${\cvd}(G)$ are NP-hard to compute, hence we only demand computability for~$\kappa$. If one insists on polynomial-time computability, one may equivalently consider parameterization by the size of polynomial-time computable approximation solutions which are at most 3 times as large as ${\vc}(G), {\diss}(G)$, and ${\cvd}(G)$, respectively.

An important technique to design FPT-algorithms is {kernelization} or parameterized preprocessing.
Formally, a parameterized problem $(L,\kappa)$ admits a {\em kernelization} if given $x\in\Sigma^*$, there is an algorithm $\cA$ that runs in polynomial time and outputs an instance $x'\in \Sigma^*$ such that {\bf (i)} $x\in L$ if and only if $x'\in L$, and {\bf (ii)} $|x'| + \kappa(x') \leq g(\kappa(x))$ for some computable function $g: \nn \rightarrow \nn$.
If $g(\cdot)$ is a polynomial function, then $L$ is said to admit a {\em polynomial kernel}.
It is well-known that a decidable parameterized problem is FPT if and only if it admits a kernel \cite{CyganFKLMPPS15}.
For a more detailed discussion on parameterized complexity and kernelization, we refer to \cite{CyganFKLMPPS15,DowneyF13,FominLSZ19}.
\fi

\paragraph{Parameterized Enumeration and Enumeration Kernelization.}
\iflong\else
For details on parameterized complexity and kernelization, we refer to \cite{CyganFKLMPPS15,DowneyF13}. \fi 
Creignou et al.~\cite{CreignouMMSV17} defined parameterized enumeration as follows.
An {\em enumeration problem} over a finite alphabet~$\Sigma$ is a tuple $(L, {\sol})$ such that 
%\begin{enumerate}[(i)]
%\item 
{\bf (i)} $L\subseteq \Sigma^*$ is a decision problem, and
%	\item 
{\bf (ii)} ${\sol}: \Sigma^* \rightarrow 2^{\Sigma^*}$ is a computable function and ${\sol}(x) \neq \emptyset$ if and only if $x \in L$.
%\end{enumerate}
%
If~$x\in L$, then ${\sol}(x)$ is the finite \emph{set of solutions} to $x$.
A {\em parameterized enumeration problem} is a triple $\Pi = (L, {\sol}, \kappa(x))$ where $(L, {\sol})$ is an enumeration problem and $\kappa: \Sigma^* \rightarrow \nn$ is the parameterization.

%We define the parameterization as a function $\kappa(x)$ but 
% or $\kappa(x)$ can be computed in polynomial-time given $x$.
%either a polynomial-time computable function, the .
% Given an instance $x \in \Sigma^*$,  we assume that the parameter $\kappa(x)$ is given with the input or can be computed in polynomial time.
An enumeration algorithm $\cA$ for a parameterized enumeration problem $\Pi$ is a deterministic algorithm that given $x \in \Sigma^*$, outputs ${\sol}(x)$ exactly without duplicates and terminates after a finite number of steps.
%If an enumeration algorithm $\cA$ terminates in $f(\kappa(x))|x|^{\OO(1)}$-time, then $\cA$ is called an {\em FPT-enumeration algorithm}.
For $x \in L$ and $1 \leq i \leq |{\sol}(x)|$, the {\em $i$-th delay} of $\cA$ is the time taken between outputting the $i$-th solution and the $(i+1)$-th solution of ${\sol}(x)$.
The {\em $0$-th delay} is the {\em precalculation time}, that is, the time taken from the start of the algorithm $\cA$ to outputting the first solution.
 The {\em postcalculation} time is the $|{\sol}(x)|$-th delay, that is, the time from the last output to the termination of $\cA$.
If an enumeration algorithm~$\cA$ guarantees that every delay is bounded by $f(\kappa(x))\cdot |x|^{\OO(1)}$, then $\cA$ is an {\em FPT-delay enumeration algorithm} ({\em FPT-delay} algorithm). 
If every delay is bounded by $|x|^{\OO(1)}$, then $\cA$ is a \emph{polynomial-delay enumeration algorithm}.
We will show kernelization results for the following definition of enumeration kernels.

\begin{definition}[\cite{GolovachKKL22}]
\label{defn:poly-delay-enum-kernel}
Let $\Pi = (L, {\sol}, \kappa)$ be a parameterized enumeration problem.
A {\em polynomial-delay enumeration kernel} for $\Pi$ is a pair of algorithms $\cA$ and~$\BB$ with the following properties.
\begin{enumerate}[(i)]
	\item For every instance $x\in \Sigma^*$, the {\em kernelization algorithm} $\cA$ runs in time polynomial in $|x| + \kappa(x)$ and outputs an instance $y\in \Sigma^*$ such that $|y| + \kappa(y) \leq g(\kappa(x))$ for a computable function $g: \nn \rightarrow \nn$, and
	\item \label{cond:sol-lift} for every $s \in \sol(y)$, the {\em solution-lifting algorithm} $\BB$ computes with delay polynomial in $|x| + |y| + \kappa(x) + \kappa(y)$, a nonempty set of solutions $S_s \subseteq {\sol}(x)$ such that $\{S_s \mid s \in {\sol}(y)\}$ is a partition of ${\sol}(x)$. 
\end{enumerate}
The function $g(\cdot)$ is the {\em size} of the polynomial-delay enumeration kernel.
If $g(\cdot)$ is a polynomial function, then $\Pi$ is said to admit a polynomial-delay enumeration kernel of {\em polynomial size}.
\end{definition}

Observe that, by condition (\ref{cond:sol-lift}) in the above definition, $x \in L$ if and only if $y \in L$.
Observe from the definition that every polynomial-delay enumeration kernelization also corresponds to a kernel for the decision version of the corresponding problem.

Before stating our results, we show that the existence of polynomial delay enumeration kernels of a small size also entails running time guarantees concerning the non-polynomial part of the total running time. To put this bound into context, let us mention that the existence of an enumeration algorithm with FPT delay implies that there is an enumeration algorithm where the precalculation time is~$f(\kappa(x))$ and all further delays and the postcalculation time are polynomial~\cite{Meier20}.
The bound for~$f(\cdot)$ is, however, usually prohibitively large. For polynomial-delay enumeration kernels we obtain the following.

\iflong

%\begin{restatable}[$\star$]{proposition}{thmThree}
\begin{proposition}
\label{prop-enumalgo}
  Let~$\Pi = (L, {\sol}, \kappa)$ be a parameterized enumeration problem such that there is an algorithm that outputs for any given~$x\in \Sigma^*$ the set~$\sol(x)$ in~$T(|x|)$~time.
  If~$\Pi$ admits a polynomial-delay enumeration kernel of size~$g$, then~$\Pi$ admits an algorithm with~$T(g(\kappa(x)))+|x|^{\OO(1)}$~precalculation time and polynomial delay.
\end{proposition}
%\end{restatable}

\else

\begin{restatable}[$\star$]{proposition}{thmThree}
%\begin{proposition}
\label{prop-enumalgo}
  Let~$\Pi = (L, {\sol}, \kappa)$ be a parameterized enumeration problem such that there is an algorithm that outputs for any given~$x\in \Sigma^*$ the set~$\sol(x)$ in~$T(|x|)$~time.
  If~$\Pi$ admits a polynomial-delay enumeration kernel of size~$g$, then~$\Pi$ admits an algorithm with~$T(g(\kappa(x)))+|x|^{\OO(1)}$~precalculation time and polynomial delay.
%\end{proposition}
\end{restatable}
\fi

\iflong
\begin{proof}
Since $\Pi$ admits a polynomial-delay enumeration kernel, let $\cA$ be the kernelization algorithm and $\BB$ be the solution-lifting algorithm for the polynomial-delay enumeration kernel of $\Pi$.
 The algorithm first uses~$\cA$ to compute the kernel~$y$ of size~$g(\kappa(x))$ in~$|x|^{\OO(1)}$~time.
 Then it computes in~$T(|y|)$ time the set~$\sol(y)$.
 This concludes the precalculation. 
 {%\color{blue}
If ${\sol}(y) = \emptyset$, then the algorithm stops without outputting any solution to the instance $x$.
Since ${\sol}(x) \neq \emptyset$ if and only if ${\sol}(y) \neq \emptyset$, this step is correct.
Otherwise, ${\sol}(y) \neq \emptyset$.
 }
Then, the algorithm outputs for each~$s\in \sol(y)$ using the solution-lifting algorithm~$\BB$ the corresponding solutions~$S_s\subseteq\sol(x)$. The precalculation time is $T(|y|)+|x|^{\OO(1)}$ to compute the kernel, and all solutions in the kernel and the first solution output by the solution-lifting algorithm. 
 Next, all remaining solutions in~$\sol(x)$ are output with polynomial delay, either as the next solution of some solution~$S_s$ or as the first solution of the next solution of~$S_{s'}$ for the next solution~$s'$ of~$\sol(y)$.  
\end{proof} 
\fi

The above is similar to decision problems where small kernels and  exact exponential-time algorithms directly give concrete running time bounds for FPT-algorithms. 
\todo{F: maybe give an example for what an FPT-delay we would achieve for our vc-result? C: Sure, why not? Diptapriyo: fine with me as long as page-limit permits.}

\section{Parameterization by the Vertex Cover Number}
\label{sec:k-path-vc}

%\paragraph%%Vertex Cover section
 
In this section, we describe a polynomial-delay enumeration kernel for {\enumkpath} when parameterized by ${\vc}$, the size of a minimum vertex cover of the input graph. We first describe the kernelization algorithm and then the solution-lifting algorithm.
We assume without loss of generality that $G$ has no isolated vertex, as isolated vertices cannot be part of any $k$-path for~$k>1$ and {\enumkpath } is trivial for~$k = 1$.

\iflong
\subsection{Marking Scheme and Kernelization Algorithm}
\label{sec:vc-marking-scheme}
\else
\paragraph{Marking Scheme and Kernelization Algorithm.}
\fi

To obtain our kernelization, we mark the vertices of~$G$ that we will keep in the kernel.
Given the input graph $G = (V, E)$ and an integer $k$, we first invoke a linear-time algorithm~\cite{Savage82} that outputs a vertex cover $X$ of $G$ such that $|X| \leq 2\cdot {\vc}(G)$.
We mark all vertices in~$X$.
Now, let $I = V(G) \setminus X$.
We consider a partition $I = I_1 \uplus I_2$ such that $I_1$ is the set of all vertices of $I$ that have degree one in $G$. We call the vertices of~$I_1$ \emph{pendant}.
We then construct an auxiliary bipartite graph $H$ with bipartition $({{X}\choose{2}}, I_2)$ as follows:
a vertex $u \in I_2$ is adjacent to a vertex $\{x, y\} \in {{X}\choose{2}}$ in $H$ if and only if $\{x, y\} \subseteq N_G(u)$.  This bipartite graph construction is inspired by Bodlaender et al.~\cite{BodlaenderJK13} who showed that the $k$-paths of $G$ correspond to matchings in $H$. 
Since~$G$ has no isolated vertex and $I_2$ excludes the pendant vertices, $H$~has no isolated vertex.
Then, we invoke \Cref{prop:new-q-exp-lemma}, the new $q$-expansion lemma, with $q = 3$ on the graph $H$. That is, in polynomial time we obtain a set~$\widehat{A}\subseteq \binom{X}{2}$ and a set~$\widehat{B}\subseteq I_2$ such that there is a 3-expansion~$M\subseteq E(H)$ of~$\widehat{A}$ onto~$\widehat{B}$ such that {\bf (i)}~$N_H(\widehat{B})\subseteq \widehat{A}$ and {\bf (ii)}~$|I_2\setminus \widehat{B}|\le q\cdot |\binom{X}{2}\setminus \widehat{A}|$.
Furthermore, there exists a set~$B'\subseteq \widehat{B}$ of size~$3\cdot |\widehat{A}|$ such that each edge of~$M$ has exactly one endpoint in~$B'$, the vertices in~$\widehat{B}$ saturated by~$M$.
We mark all vertices in~$B'\cup (I_2\setminus \widehat{B})$.
If no 3-expansion exists, that is, if~$\widehat{A}=\emptyset$ and thus also~$\widehat{B}=B'=\emptyset$, then~$I_2\setminus \widehat{B}=I_2$. Consequently the kernel contains all vertices of~$I_2$ in that case, and~$I_2$ is small by property {\bf (ii)} of \Cref{prop:new-q-exp-lemma}.
Otherwise, it is sufficient to mark the vertices~$B'\subseteq \widehat{B}$ which are saturated by~$M$. To prove this, we first define\iflong, in \Cref{vc:k-path-structures}, \fi{} a notion of equivalence for~$k$-paths in~$G$ and~$G'$. Then, in \Cref{lemma:vc-same-signature-existence}, we show that the marked vertices are sufficient to keep for each $k$-path~$P$ in~$G$ an equivalent one in~$G'$.

Moreover, for every $x \in X$, we mark one arbitrary pendant vertex $u \in I_1$ such that $xu \in E(G)$.
Let $I_1^+$ denote the set of all marked pendant vertices of $I_1$.
It is sufficient to mark one vertex of~$I_1^+$ per~$x\in X$ since each $k$-path can contain at most one such vertex (in the beginning or the end).
Now, $G'$ is the subgraph induced by all marked vertices, that is, $V(G')=X\cup I^1_+\cup (I_2\setminus \widehat{B})\cup B'$.
Since $G'$ is an induced subgraph of $G$ and by properties (i) and (ii) of \Cref{prop:new-q-exp-lemma} (new $q$-expansion lemma), we may observe the following.

\iflong

%\begin{restatable}[$\star$]{observation}{obs:k-path-preserve}
\begin{observation}
\label{obs:k-path-preserve}
Let $G'$ be the graph obtained from $G$ after invoking the above marking scheme.
Then, $G'$ has $\OO(|X|^2)$ vertices.
Moreover, any $k$-path of $G'$ is also a $k$-path of $G$.
\end{observation}
%\end{restatable}

\else
\begin{restatable}[$\star$]{observation}{obs:k-path-preserve}
%\begin{observation}
\label{obs:k-path-preserve}
Let $G'$ be the graph obtained from $G$ after invoking the above marking scheme.
Then, $G'$ has $\OO(|X|^2)$ vertices.
Moreover, any $k$-path of $G'$ is also a $k$-path of $G$.
%\end{observation}
\end{restatable}
\fi

\iflong
\begin{proof}
By construction, at most $|X|$ pendant vertices of $I_1$ are marked.
Hence, $|I_1^+| \leq |X|$.
Moreover,  property (ii) of \Cref{prop:new-q-exp-lemma} implies that $|I_2 \setminus \widehat{B}| \leq 3\cdot|{{X}\choose{2}} \setminus \widehat{A}| \leq 3\cdot{{|X|}\choose{2}}$ and $|B'| = 3\cdot|\widehat{A}| \leq 3\cdot{{|X|}\choose{2}}$.
Therefore, $G'$ has $\OO(|X|^2)$ vertices.

The second claim of the statement follows since $G'$ is an induced subgraph of $G$.
\end{proof}
\fi

We let $I':=I\cap V(G')$ denote the independent set in~$G'$. 
Observe that~$I'$ consists of the marked pendant vertices~$I_1^+$, the vertices~$B'$ saturated by the 3-expansion~$M$, and the vertices of~$I_2\setminus \widehat{B}$ which have at least one neighbor in~$\binom{X}{2}\setminus\widehat{A}$.
Moreover, let $H'$ be the subgraph of $H$ induced by~$V(G')$ with bipartition $({{X}\choose{2}}, (I_2 \setminus \widehat{B}) \cup B')$.

%\medskip

\iflong
\subsection{Signatures and Equivalence among $k$-Paths}
\label{vc:k-path-structures}
\else
\paragraph{Signatures and Equivalence among $k$-Paths.}
\fi
While the kernelization algorithm is sufficient to prove that $(G', X, k)$ is equivalent to $(G, X, k)$ for the decision version of the problem, we need to specify some relation between the $k$-paths of $G$ and the $k$-paths of $G'$ for the solution-lifting algorithm.
Based on the description of the marking scheme,  
every $x \in X$ has at most one neighbor $v \in I_1^+$ in the graph~$G'$.
If $x \in X$ has a unique pendant neighbor $v \in I_1$ in $G$, then we call $v$ a {\em 1-pendant} vertex of $I_1$.
Otherwise, if~$x \in X$ has at least two pendant neighbors in $I_1$, then we call each of these pendant neighbors of $x$ a {\em 2-pendant} vertex of $I_1$.  
This distinction is important since a $k$-path in~$G$ may use an unmarked vertex~$v$ of~$I_1$, but then~$v$ has to be a 2-pendant vertex.

Let $\RR \subseteq I'$ denote the set of all vertices of $G'$ that are in $I_2 \setminus \widehat{B}$ or 1-pendant vertices of~$I_1$.
Roughly speaking, $\RR$ contains the vertices of~$I'$ (and also of~$I$) which are somewhat unique; in contrast,~$I'\setminus\RR$ contains the vertices which have many similar vertices in~$I\setminus I'$.
Our definition of equivalence for $k$-paths will essentially make paths equivalent if they have the same interaction with~$X\cup \RR$, where the vertices of~$\RR$ are included because of their ``uniqueness''. We formalize this interaction as follows.
%We have to prove some properties that establishes some relations between all $k$-paths of $G$ and all $k$-paths in $G'$.
\begin{definition}
\label{defn:vc-signature-P}
Let $P=(v_1,\ldots,v_k)$ be a $k$-path of $G$ or $G'$.
Then, the function ${\rmf}_P: V(P) \cap (X \cup \RR) \rightarrow [k]$ with ${\rmf}_P(v_i) = i$ is called the {\em signature} of~$P$. 
\end{definition}

\begin{figure}[t]
\centering
\begin{adjustbox}{width=\textwidth}
\begin{tikzpicture}

%%%%%%%%%%%%%%%%%%%%%%%%%%%%%%%%%%%%%%%%%%%%%%%%%%%%%%%%%%%%%%%%%%%%%%%%%%%%%%%%%%%%%%%%%%%%%%%%%%%%
%%%%%%%%% PART a: auxillary bipratite graph %%%%%%%%%%%%%%%%%%%%%%%%%%%%%%%%%%%%%%%%%%%%%%%%%%%%%%%%
%%%%%%%%%%%%%%%%%%%%%%%%%%%%%%%%%%%%%%%%%%%%%%%%%%%%%%%%%%%%%%%%%%%%%%%%%%%%%%%%%%%%%%%%%%%%%%%%%%%%
\begin{scope}[xshift=7.2cm]

\node[label=left:{$b)$}](rr) at (-1.7, 2.2) []{};

\draw[rounded corners,dotted,very thick] (-1.4, 1.0) rectangle (1.0, 2.1) {};
\node[label=above:{$B'$}](tt) at (-0.2, 1.9) []{};
\draw[rounded corners,dotted,very thick] (1.6, 1.0) rectangle (3.0, 2.1) {};
\node[label=above:{$\widehat{B}\setminus B'$}](tt) at (2.3, 1.9) []{};
\draw[rounded corners,dotted,very thick] (3.6, 1.0) rectangle (4.0, 2.1) {};
\node[label=above:{$I_2\setminus\widehat{B}$}](tt) at (3.8, 1.9) []{};
\draw[rounded corners] (-1.7, 0.8) rectangle (4.3, 1.9) {};
\node[label=left:{$I_2$}](X) at (-1.6, 1.2) []{};

\node[label=above:{$i_1$}](V1) at (-1.2, 1.3) [shape = circle, draw, fill=black, scale=0.11ex]{};
\node[label=above:{$i_2$}](V2) at (-0.2, 1.3) [shape = circle, draw, fill=black, scale=0.11ex]{};
\node[label=above:{$i_3$}](V3) at (0.8, 1.3) [shape = circle, draw, fill=black, scale=0.11ex]{};
\node[label=above:{$i_4$}](V4) at (1.8, 1.3) [shape = circle, draw, fill=black, scale=0.11ex]{};
\node[label=above:{$i_5$}](V5) at (2.8, 1.3) [shape = circle, draw, fill=black, scale=0.11ex]{};
\node[label=above:{$i_6$}](V6) at (3.8, 1.3) [shape = circle, draw, fill=black, scale=0.11ex]{};

\draw[rounded corners] (-1.7, -0.7) rectangle (4.3, 0.1) {};
\node[label=left:{$\binom{X}{2}$}](X) at (-1.5, -0.3) []{};
\node[label=below:{$x_{2,3}$}](X1) at (-0.4, -0.2) [shape = circle, draw, fill=black, scale=0.11ex]{};
\node[label=below:{$x_{3,4}$}](X2) at (1.3, -0.2) [shape = circle, draw, fill=black, scale=0.11ex]{};
\node[label=below:{$x_{5,6}$}](X3) at (3.0, -0.2) [shape = circle, draw, fill=black, scale=0.11ex]{};

\path [-,line width=0.2mm](X1) edge (V1);
\path [-,line width=0.2mm](X1) edge (V2);
\path [-,line width=0.2mm](X1) edge (V3);
\path [-,line width=0.2mm](X1) edge (V4);

\path [-,line width=0.2mm](X2) edge (V4);
\path [-,line width=0.2mm](X2) edge (V5);

\path [-,line width=0.2mm](X3) edge (V6);

%%%%%%%%%%%%%%%%%%%%%% matched edges %%%%%%%%%%%%%%%%%%%%%%%%%%%%%%%%%%%%%%%%%%%%%%%%%%%%%%%%%%%%%%%
\draw[red,line width=3pt]  (X1) to (V1);
\draw[red,line width=3pt]  (X2) to (V4);

\draw[blue!40,line width=3pt]  (X1) to (V4);
\draw[blue!40,line width=3pt]  (X2) to (V5);

\draw[dashed,dash pattern=on 5pt off 5pt,red,line width=3pt]  (X3) to (V6);
\draw[dashed,dash pattern=on 5pt off 5pt,blue!40,dash phase=5pt,line width=3pt]  (X3) to (V6);

\end{scope}

%%%%%%%%%%%%%%%%%%%%%%%%%%%%%%%%%%%%%%%%%%%%%%%%%%%%%%%%%%%%%%%%%%%%%%%%%%%%%%%%%%%%%%%%%%%%%%%%%%%%
%%%%%%%%% PART b: input graph %%%%%%%%%%%%%%%%%%%%%%%%%%%%%%%%%%%%%%%%%%%%%%%%%%%%%%%%%%%%%%%%%%%%%%
%%%%%%%%%%%%%%%%%%%%%%%%%%%%%%%%%%%%%%%%%%%%%%%%%%%%%%%%%%%%%%%%%%%%%%%%%%%%%%%%%%%%%%%%%%%%%%%%%%%%
\begin{scope}[xshift=-7.2cm]

\node[label=left:{$a)$}](rr) at (5.5, 2.2) []{};

\draw[rounded corners] (5.5, -0.7) rectangle (11.5, 0.1) {};
\node[label=left:{$X$}](X) at (5.6, -0.3) []{};

\draw[rounded corners] (5.5, 0.8) rectangle (11.5, 1.9) {};
\node[label=left:{$I$}](X) at (5.6, 1.2) []{};

\draw[rounded corners,dotted,very thick] (5.8, 1.0) rectangle (8.2, 2.1) {};
\node[label=above:{$B'$}](tt) at (7.0, 1.9) []{};
\draw[rounded corners,dotted,very thick] (8.8, 1.0) rectangle (10.2, 2.1) {};
\node[label=above:{$\widehat{B}\setminus B'$}](tt) at (9.5, 1.9) []{};
\draw[rounded corners,dotted,very thick] (10.8, 1.0) rectangle (11.2, 2.1) {};
\node[label=above:{$I_2\setminus\widehat{B}$}](tt) at (11.0, 1.9) []{};

\node[label=above:{$i_1$}](V11) at (6.0, 1.3) [shape = circle, draw, fill=black, scale=0.11ex]{};
\node[label=above:{$i_2$}](V21) at (7.0, 1.3) [shape = circle, draw, fill=black, scale=0.11ex]{};
\node[label=above:{$i_3$}](V31) at (8.0, 1.3) [shape = circle, draw, fill=black, scale=0.11ex]{};
\node[label=above:{$i_4$}](V41) at (9.0, 1.3) [shape = circle, draw, fill=black, scale=0.11ex]{};
\node[label=above:{$i_5$}](V51) at (10.0, 1.3) [shape = circle, draw, fill=black, scale=0.11ex]{};
\node[label=above:{$i_6$}](V61) at (11.0, 1.3) [shape = circle, draw, fill=black, scale=0.11ex]{};

\node[label=below:{$x_1$}](X11) at (6.0, -0.2) [shape = circle, draw, fill=black, scale=0.11ex]{};
\node[label=below:{$x_2$}](X21) at (7.0, -0.2) [shape = circle, draw, fill=black, scale=0.11ex]{};
\node[label=below:{$x_3$}](X31) at (8.0, -0.2) [shape = circle, draw, fill=black, scale=0.11ex]{};
\node[label=below:{$x_4$}](X41) at (9.0, -0.2) [shape = circle, draw, fill=black, scale=0.11ex]{};
\node[label=below:{$x_5$}](X51) at (10.0, -0.2) [shape = circle, draw, fill=black, scale=0.11ex]{};
\node[label=below:{$x_6$}](X61) at (11.0, -0.2) [shape = circle, draw, fill=black, scale=0.11ex]{};

\path [-,line width=0.2mm](X21) edge (V11);
\path [-,line width=0.2mm](X21) edge (V21);
\path [-,line width=0.2mm](X21) edge (V31);
\path [-,line width=0.2mm](X21) edge (V41);

\path [-,line width=0.2mm](X31) edge (V11);
\path [-,line width=0.2mm](X31) edge (V21);
\path [-,line width=0.2mm](X31) edge (V31);
\path [-,line width=0.2mm](X31) edge (V41);
\path [-,line width=0.2mm](X31) edge (V51);

\path [-,line width=0.2mm](X41) edge (V41);
\path [-,line width=0.2mm](X41) edge (V51);

\path [-,line width=0.2mm](X51) edge (V61);

\path [-,line width=0.2mm](X61) edge (V61);

\path [-,line width=0.2mm](X11) edge (X21);
\path [-,line width=0.2mm](X41) edge (X51);

%%%%%%%%%%%%%%%%%%%%%%%%%%%%%%%%%%%%% the two paths %%%%%%%%%%%%%%%%%%%%%%%%%%%%%%%%%%%%%%%%%%%%%%%%
\draw[dashed,dash pattern=on 5pt off 5pt,red,line width=3pt]  (X11) to (X21);
\draw[dashed,dash pattern=on 5pt off 5pt,blue!40,dash phase=5pt,line width=3pt]  (X11) to (X21);

\draw[dashed,dash pattern=on 5pt off 5pt,red,line width=3pt]  (X41) to (X51);
\draw[dashed,dash pattern=on 5pt off 5pt,blue!40,dash phase=5pt,line width=3pt]  (X41) to (X51);

\draw[dashed,dash pattern=on 5pt off 5pt,red,line width=3pt]  (X51) to (V61);
\draw[dashed,dash pattern=on 5pt off 5pt,blue!40,dash phase=5pt,line width=3pt]  (X51) to (V61);

\draw[dashed,dash pattern=on 5pt off 5pt,red,line width=3pt]  (X61) to (V61);
\draw[dashed,dash pattern=on 5pt off 5pt,blue!40,dash phase=5pt,line width=3pt]  (X61) to (V61);

\draw[dashed,dash pattern=on 5pt off 5pt,red,line width=3pt]  (X31) to (V41);
\draw[dashed,dash pattern=on 5pt off 5pt,blue!40,dash phase=5pt,line width=3pt]  (X31) to (V41);

\draw[red,line width=3pt]  (X21) to (V11);
\draw[red,line width=3pt]  (X31) to (V11);
\draw[red,line width=3pt]  (X41) to (V41);

\draw[blue!40,line width=3pt]  (X21) to (V41);
\draw[blue!40,line width=3pt]  (X31) to (V51);
\draw[blue!40,line width=3pt]  (X41) to (V51);
\end{scope}
\end{tikzpicture}
\end{adjustbox}

\caption{Part~$a)$ shows a graph~$G$ and its decomposition into the modulator~$X$ and the independent set~$I$. 
Here, the vertices in~$I_1$ are not shown.
The two paths~$P_1=(x_1,x_2,i_1,x_3,i_4,x_4,x_5,i_6,x_6)$ (in \textcolor{red}{red}) and~$P_1=(x_1,x_2,i_4,x_3,i_5,x_4,x_5,i_6,x_6)$ (in \textcolor{blue!50}{blue}) have the same signature~${\rmf}(x_1)=1,{\rmf}(x_2)=2,{\rmf}(x_3)=4,{\rmf}(x_4)=6,{\rmf}(x_5)=7,{\rmf}(i_6)=8,{\rmf}(x_6)=9$.
Part~$b)$ shows the corresponding auxiliary bipartite graph and the highlighted edges show the edges occupied (see paragraph above \Cref{lemma:edge-occupation-vc} for a definition) of both paths~$P_1$ and~$P_2$.
}
\label{fig-vc-example-signature}
\end{figure}
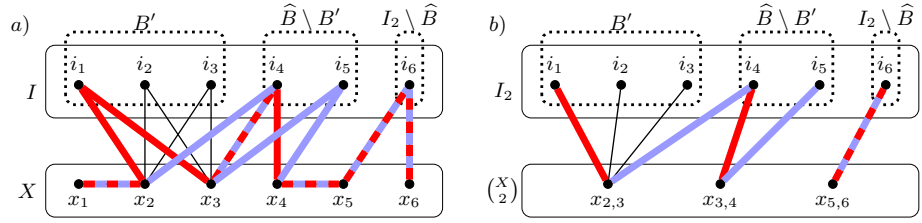

Informally, the {\em signature} of $P$ is the mapping of the vertices of $X \cup \RR$ to their position in $P$; for an example see \Cref{fig-vc-example-signature}. \iflong\else In other words, the signature is the restriction of the position fuction ${\pos}_P$, which maps each vertex to its index, to $X \cup \RR$. \fi
%
%\medskip
%
Observe that one can compute the signature~${\rmf}_P$ of a given $k$-path in linear time. 

%For an example for the signature of a $k$-path see \Cref{fig-vc-example-signature}.

In order to establish a relation between the $k$-paths of $G$ and the $k$-paths of~$G'$, we define the notion of equivalent $k$-paths as follows.
\begin{definition}
\label{defn:vc-equivalent-k-paths}
Two $k$-paths $P_1$ and $P_2$ of $G$ are said to be {\em equivalent} if
\begin{enumerate}[(i)]
	%\item $V(P_1) \cap (X \cup \RR) = V(P_2) \cap (X \cup \RR)$,
	\item \label{item-1-pendant} $P_1$ and $P_2$ have the same signature.
	\item \label{item-2-pendant} $P_1$ starts (respectively ends) at a 2-pendant vertex of $I$ if and only if $P_2$ starts (respectively ends) at a 2-pendant vertex of $I$.
\end{enumerate}
\end{definition}

An example of two equivalent $k$-paths is shown in \Cref{fig-vc-example-signature}.
Note that \Cref{item-1-pendant} of \Cref{defn:vc-equivalent-k-paths} implies that~$V(P_1) \cap (X \cup \RR) = V(P_2) \cap (X \cup \RR)$.
Also, note that in \Cref{item-2-pendant} of \Cref{defn:vc-equivalent-k-paths} it is sufficient to consider 2-pendant vertices, since all 1-pendant vertices are contained in~$\RR$.

To obtain $k$-paths from a signature, we use another mapping \emph{extension} from $I \setminus \RR$ to the indices of~$P$ such that the combination of both the mappings yields a $k$-path in $G$ (or~$G'$). %containing at least one vertex that is not in $G'$.
Its formal definition is given as follows.

\begin{definition}
\label{defn:vc-extn-signature-P}
Let $P'$ be a $k$-path in $G'$ and let $I^* \subseteq I$.
A mapping~${\rmg}: I^* \rightarrow [k]$ is an {\em extension} of the signature ${\rmf}_{P'}$ if
\iflong
\begin{enumerate}[(i)]
	\item ${\rmg}$ is injective, and
	%\item $I^* \setminus I' = I^* \setminus V(G') \neq \emptyset$, and
	\item there is a~$k$-path $P$ in $G$ such that~$\rmf_{P'} \cup \rmg$ is the position function $\pos_{P}$.
\end{enumerate}
\else
$(i)$  ${\rmg}$ is injective, and $(ii)$ there is a~$k$-path $P$ in $G$ such that~$\rmf_{P'} \cup \rmg$ is the position function $\pos_{P}$.
\fi
\end{definition}

Note that if~$I^*\subseteq I'$, then the $k$-path~$P$ obtained by~$\rmf_{P'} \cup \rmg$ is a $k$-path in~$G'$.

The equivalence definition is given for $k$-paths of~$G$ but since~$G'$ is an induced subgraph of~$G$ this also implicitly defines equivalence between $k$-paths of~$G'$ and between $k$-paths of~$G$ and $k$-paths of $G'$.
It is clear that the above definition is an equivalence relation for the paths of~$G$.
Hence, a $k$-path $P$ is equivalent to itself. 
Consequently, every $k$-path~$P$ of~$G'$ has some equivalent $k$-path in~$G$, namely~$P$ itself.
We now show that the reverse is also true.

\iflong
%\begin{restatable}[$\star$]{lemma}{lemma:vc-same-signature-existence}
\begin{lemma}
\label{lemma:vc-same-signature-existence}
Let $P$ be a $k$-path of $G$.
Then, $G'$ has a $k$-path $P'$ that is equivalent to~$P$.	
\end{lemma} 
%\end{restatable}

\else

\begin{restatable}[$\star$]{lemma}{lemma:vc-same-signature-existence}
%\begin{lemma}
\label{lemma:vc-same-signature-existence}
Let $P$ be a $k$-path of $G$.
Then, $G'$ has a $k$-path $P'$ that is equivalent to~$P$.	
%\end{lemma} 
\end{restatable}

\fi

\iflong
\begin{proof}
Let $P$ be a $k$-path of $G$ and let~$Y := (X \cup \RR) \cap V(P)$.
If $P$ is a $k$-path of $G'$, then~$P$ is equivalent to itself and we are done.
Thus, in the following we assume that $P$ is not a $k$-path in $G'$.
Hence, $P$ contains at least one vertex that is not in the kernel.
More precisely, $P$ contains a vertex from $I_1 \setminus I_1^+$ or  a vertex from $\widehat{B} \setminus B'$ (or both).
Let ${\rmf}_P: Y \rightarrow [k]$ denote the signature of $P$ and let ${\rmg}: Z \rightarrow [k]$ denote the extension of ${\rmf}_P$ for $Z:= V(P)\setminus Y \subseteq I$.
We prove that there is a different subset $Z' \subseteq I'$ with an extension~$\rmg':Z'\to [k]$ such that $\rmf_{P} \cup \rmg'$ gives a $k$-path of $G'$.

Recall that 1-pendant vertices and vertices in~$I_2\setminus \widehat{B}\subseteq \RR\subseteq I'$ are part of the signature and thus~$P$ and~$P'$ contain the same subset of these vertices which appear on the exact same positions.
Our strategy to obtain~$Z'$ is two steps as follows:
in Step~1, we check if the first vertex~$u$ of~$P$ is contained in~$Z$.
If yes, we replace~$u$ by a vertex~$w$ from~$I'$ preserving the equivalence of the paths and add~$w$ to~$Z'$.
We handle the last vertex~$u$ of~$P$ analogously.
Observe that~$|Z'|\le 2$ after Step~1.
In Step~2, we replace all intermediate vertices of~$V(P)\cap Z$.
Here, we exploit the 3-expansion~$M$.
More precisely, since we used at most 2~vertices to replace the first and last vertex, for each vertex~$x\in \widehat{A}$, there is at least one neighbor~$y\in B'$ remaining such that~$M$ contains the edge~$xy$.
Next, we provide the details for both steps.

\emph{Step 1:}
We initialize $Z' := \emptyset$.
Furthermore, we distinguish the case whether~$P$ contains vertices of~$I_1\setminus I_1^+$, or only vertices of~$I_2\setminus \widehat{B}$.

If $P$ contains a vertex $u \in I_1 \setminus I_1^+$, then $u$ is the start (or end) vertex of $P$ and $u$ is a 2-pendant vertex in $G$ (since 1-pendant vertices are part of the signature).
Formally, ${\rmg}(u) = 1$ (or ${\rmg}(u) = k$).
Let~$x \in X$ be the unique neighbor of~$u$ in $G$. 
It follows that ${\rmf}_P(x) = 2$ (or ${\rmf}_P(x) = k-1$, respectively).
By the algorithm of the marking scheme, there exists a vertex~$u' \in I_1^+$ such that $N(u') = \{x\}$.
We set $Z' = Z' \cup \{u'\}$ and ${\rmg'}(u') = 1$ (or ${\rmg'}(u') = k$, respectively). That is, $u'$ is the first vertex (or the last vertex, respectively) of $P'$.
This ensures that condition (ii) is satisfied from \Cref{defn:vc-equivalent-k-paths}.

In the following, we assume that all remaining vertices of~$Z$ stem from~$I_2\setminus \widehat{B}$.
Recall that we still aim to replace the vertices of $V(P) \cap \widehat{B}$ that are endpoints of~$V(P)$.
First, observe that by definition every vertex in $\widehat{B} \cap V(P)$ has at least two neighbors in~$X$. 
If the first vertex~$u$ of $P$ is contained in~$\widehat{B}$, then ${\rmg}(u) = 1$ and there is a vertex~$x \in X$ such that ${\rmf}_P(x) = 2$.
% By definition, $B \subseteq I_2$ and $u$ is not an isolated vertex in $G$.
% Hence, $u$ has at least two neighbors in $X$.
Let $x$ and some other arbitrary but fixed vertex~$x'$ be two neighbors of $u$ in~$G$. 
Since $u \in \widehat{B}$, it follows from property~(i) of \Cref{prop:new-q-exp-lemma} that $\{x, x'\} \in \widehat{A}$. By property~(ii) of \Cref{prop:new-q-exp-lemma} and the choice of $q = 3$ in the marking scheme, it follows that there are three vertices $u', \widehat{u}, u'' \in B'$ such that $(\{x, x'\}, u'), (\{x, x'\}, \widehat{u}), (\{x, x'\}, u'') \in M$.
We set~$Z' := Z' \cup \{u'\}$ and $
\rmg'(u')=1$, that is, $u'$ is the first vertex of $P'$.
Note that it is irrelevant which vertex of~$u', \widehat{u}, u''$ we choose.
%Similarly
The last vertex~$u$ of~$P$ can be handled analogously when~$u$ is from~$\widehat{B}$: 
If~$u$ is from $\widehat{B}$, then there is $y \in X$ such that ${\rmf}_P(y) = k-1$ and ${\rmg}(u) = k$.
%Same as before, we can argue that $v$ has at least two neighbors $y, y' \in X$ and it follows from item (i) of Proposition \ref{prop:new-q-exp-lemma} that $\{y, y'\} \in A$.
As above, there are three vertices $u', \widehat{u}, u'' \in B'$ such that $(\{y, y'\}, \widehat{u}), (\{y, y'\}, u'), (\{y, y'\}, u'') \in M$ where~$y'$ is another arbitrary but fixed neighbor of~$u$ in~$G$.
We set~$Z' := Z' \cup \{u'\}$  and $\rmg'(u')=k$, that is,~$u'$ is the last vertex of~$P'$.

\emph{Step 2:}
Finally, it remains to replace the intermediate vertices of $V(P) \cap \widehat{B}$. 
These vertices have a predecessor~$b_1$ and a successor~$b_2$ in~$P$ which are both vertices from~$X$. 
Recall that~$M$ is a 3-expansion of~$\widehat{A}$ into~$\widehat{B}$ and note that~$B'\subseteq \widehat{B}$ is the set of vertices saturated by~$M$. that is, $|B'|=3\cdot|\widehat{A}|$.
Recall that~$B'\subseteq I'\setminus\RR$.
Using~$M$, we now replace all intermediate vertices of $V(P) \cap B$ by endpoints of $M$.
For each intermediate vertex~$w \in V(P) \cap \widehat{B}$ we do the following. 
Then, there are three vertices $w',\widehat{w}, w'' \in B'$ such that $(\{b_1, b_2\}, w'), (\{b_1, b_2\}, \widehat{w}), (\{b_1, b_2\}, w'') \in M$. 
Of these three vertices, at least one, say~$w'$ is not contained in~$Z'$ since (1) at most two of them have been added when replacing the endpoints of~$V(P)\setminus \widehat{B}$, and (2) since no two edges of~$M$ have the same endpoint in~$B'$, vertex~$w'$ is not used for another replacement of an intermediate vertex.
We set~$Z' := Z' \cup \{w'\}$ and~$\rmg'(w')=\rmg(w)$.

This completes the construction of the set $Z' \subseteq I'\setminus \RR$ and~$\rmg'$.
By construction, $\rmf\cup \rmg'$ gives a $k$-path~$P'$ in $G'$.
Observe that $P'$ is equivalent to $P$ since~they have the same signature and since~$P'$ starts or end at a 2-pendant vertex if and only if~$P$ does.
\end{proof}

\fi

We now observe a structural property that illustrates a correspondence between the matchings in the auxiliary bipartite graph $H$ and the $k$-paths of $G$. This allows the solution-lifting algorithm to make use of polynomial-delay algorithms for matching enumeration.
If $xu, uy \in E(G)$ are two edges of $P$ such that $x, y \in X$ and $u \in I$, then the pair $\{x, y\}$ is adjacent to $u$ in $H$.
In such a situation, the path $P$ is said to {\em occupy} the edge $(\{x, y\}, u)$ of $H$.
We state the following lemma that also holds true for $G'$ (for an example see \Cref{fig-vc-example-signature}).

\iflong
%\begin{restatable}[$\star$]{lemma}{lemma:edge-occupation-vc}
\begin{lemma}
\label{lemma:edge-occupation-vc}
Let $P$ be a $k$-path in $G$ and let $H = ({{X}\choose{2}}, I_2)$ be the auxiliary bipartite graph.
Then, the edges of $H$ that are occupied by $P$ form a matching in $H$.
Moreover, given a $k$-path~$P$ of $G$ and the auxiliary bipartite graph $H$, the edges of $H$ that are occupied by $P$ can be computed in $\OO(n+m)$ time. 
\end{lemma}
%\end{restatable}

\else

\begin{restatable}[$\star$]{lemma}{lemma:edge-occupation-vc}
%\begin{lemma}
\label{lemma:edge-occupation-vc}
Let $P$ be a $k$-path in $G$ and let $H = ({{X}\choose{2}}, I_2)$ be the auxiliary bipartite graph.
Then, the edges of $H$ that are occupied by $P$ form a matching in $H$.
Moreover, given a $k$-path~$P$ of $G$ and the auxiliary bipartite graph $H$, the edges of $H$ that are occupied by $P$ can be computed in $\OO(n+m)$ time. 
%\end{lemma}
\end{restatable}

\fi

\iflong
\begin{proof}
\iflong
We prove the statements in the given order. 
\fi 
For the first part, if $P$ occupies only one edge of $H$, then the statement is trivially true.
Otherwise, let $(\{x, y\}, u)$ and $(\{x',y'\}, u')$ be two distinct edges of $H$ that are occupied by $P$.
Consider the edges $xu, uy, x'u', u'y' \in E(G)$.
Since $P$ is a path and each of these four edges are in $P$, we conclude that they must be distinct edges of $P$.
In fact, $(x, u, y)$ and $(x',u',y')$ form two distinct subpaths of $P$ with~$u, u' \in I$.
Thus, $\{x,y\}$ must be different from the pair $\{x',y'\}$ and $u$ must be a different vertex from~$u'$.
Hence, the edges of $H$ that are occupied by $P$ do not share endpoints, that is, they form a matching. 

For the second part, let $P$ be a $k$-path of $G$.
As the algorithm remembers $H = ({{X}\choose{2}}, I_2)$, the set $X$ can be computed in~$\OO(n+m)$~time.
For every three consecutive vertices $x, u, y \in V(P)$, we check if $x, y \in  X$ and $u \in I$.
If this happens, then we output $(\{x, y\}, u) \in E(H)$ as an edge of~$H$ that is occupied by $P$.
Since the number of vertices in $P$ is $k$, this procedure computes all edges of $H$ that are occupied by $P$ in time $\OO(k +n + m)$.
\end{proof}
\fi

\iflong
\subsection{A Challenge: Avoiding Duplicate Enumeration}
\label{sec:vc-challenges}
\else
\paragraph{Avoiding Duplicate Enumeration.}
\fi
To ensure and correctly design the solution-lifting algorithm, we need to make sure that every $k$-path~$P$ of~$G$ is outputted for exactly one~$k$-path~$P'$ of~$G'$.
The natural idea is that we output~$P$ for a path~$P'$ that is equivalent. 
In case two distinct $k$-paths~$P_1$ and~$P_2$ of~$G'$ are equivalent to a $k$-path $P$ of $G$,
then both $P_1$ and $P_2$ have the same signature where~$X'= V(P_1)\cap (X \cup \RR) = V(P_2)\cap (X \cup \RR)$ and for every $u \in X'$, ${\rmf}_{P_1}(u) = {\rmf}_{P_2}(u)$.
In such a case, if we enumerate all $k$-paths of $G$ that are equivalent to a $k$-path $P'$ of $G$, then the same $k$-path of $G$ can be output multiple times.
To circumvent this issue, we introduce a lexicographic order on equivalent paths which allows to output equivalent paths only for solution paths that are minimal with respect to this order.  

We formalize the order as follows. First, we assume some fixed total order~$\prec$ on the vertices of~$G$ and thus also on the vertices of~$G'$. 
Now, for a path~$P'$ consider the subsequence $(v_1,\ldots,v_q)$ of $P'$ containing only the vertices of~$I\setminus \RR$ and call this the~$I$-sequence of~$P'$. 
In other words, the $I$-sequence is the order of vertices of the extension of~$P$.
We say that a path~$P'$ is \emph{lexicographically smaller} than a path~$P^*$ when the $I$-sequence of~$P'$ is lexicographically smaller than the~$I$-sequence of~$P^*$.

\iflong
%\begin{restatable}[$\star$]{lemma}{obs:vc-lexicographically-smallest}
\begin{lemma}
\label{obs:vc-lexicographically-smallest}
Let $P$ be a $k$-path in $G'$.
There is an $\OO(k \cdot(m' + n'))$-time algorithm~that checks if $P$ is the lexicographically smallest $k$-path among all~$k$-paths in~$G'$ with the same~signature.
\end{lemma}
%\end{restatable}

\else

\begin{restatable}[$\star$]{lemma}{obs:vc-lexicographically-smallest}
%\begin{lemma}
\label{obs:vc-lexicographically-smallest}
Let $P$ be a $k$-path in $G'$.
There is an $\OO(k \cdot(m' + n'))$-time algorithm~that checks if $P$ is the lexicographically smallest $k$-path among all~$k$-paths in~$G'$ with the same~signature.
%\end{lemma}
\end{restatable}

\fi

\iflong
\begin{proof}
  Consider the $I$-sequence~$(v_1,\ldots,v_q)$ of~$P$. 
  For each vertex~$v_i$ in the~$I$-sequence, starting with~$v_1$, we check whether there is another $k$-path~$P^*$ with the same signature and an~$I$-sequence~$(u_1,\ldots,u_q)$ such that~$v_j=u_j$ for all~$j<i$ and~$u_i\prec v_i$. 
  The path~$P$ is lexicographically smallest with its signature if and only if the test fails for all~$i\in [q]$. 
  To perform the test  efficiently for all given~$i$, we first construct an auxiliary bipartite graph~$\hat{H}$ and a matching~$\hat{M}$ in~$\hat{H}$  and then use both to perform for the test for each~$i$.

  Recall that~$H'$ is the auxiliary bipartite graph with bipartition~$({{X}\choose{2}}, V(G')\cap I_2)$ where an edge corresponds to a vertex of~$I_2$ that has two neighbors in~$X$. The graph~$\hat{H}$ is essentially the induced subgraph of~$H'$ consisting of the vertices of the~$I$-sequence, of their neighbors in~$P$ and all further vertices of~$I_2$ that are not in~$\RR$.   
More precisely,~$\hat{H}$ consists of
  \begin{itemize}
  \item one part containing the vertices~$\{x,y\}\in {{X}\choose{2}}$ such that~$x$ and~$y$ are in~$P$ the predecessor and successor, respectively, of some~$v_i$,
  \item the vertex set~$(V(G')\cap I_2)\setminus \RR$,
  \item and the edge set~$\{\{x,y\},u\}\mid x,y\in X \land \{x,y\}\in N(u)\}$. 
  \end{itemize}
  We add some further vertices that deal with the case that the~$I$-sequence contains the first or last vertex of~$P$.
 {%\color{blue} 
 If~$v_1$ is the first vertex of~$P$, then let~$w_1$ denote its neighbor in~$P$.
 } 
  Add~$w_1$ to~$\hat{H}$, make~$w_1$ in~$\hat{H}$ adjacent to all vertices in~$N_G(w_1)\cap V(\hat{H})\cap I_2$. If~$v_q$ is the last vertex of~$P$, then add its neighbor~$w_q$ in~$P$ in a similar fashion. If they are added, then~$w_1$ and~$w_q$ are in the same part of the bipartition as the vertices of~${{X}\choose{2}}$. This completes the construction of~$\hat{H}$. Now, the matching~$\hat{M}$ consists of the edges occupied by~$P$ in~$\hat{H}$ plus the edges~$\{w_1,v_1\}$ and~$\{w_q,v_q\}$ if~$w_1$ and~$w_q$ have been added, respectively. For each~$v_i$, we let~$w_i$ denote its neighbor in the edge contained in~$\hat{M}$. Observe that in this notation, one part of~$\hat{H}$ consists exactly of~$\{w_1,\ldots,w_q\}$. 
  
  For each index~$i$, we now perform the test using~$\hat{H}$ and~$\hat{M}$ follows. For every vertex~$v_j$ with~$j<i$, remove all edges from~$\hat{H}$ that are incident with~$w_j$ except~$\{w_j,v_j\}$. For the vertex~$v_i$, remove from~$\hat{H}$ the edge~$\{w_i,v_i\}$ and all edges~$\{w_i,u\}$ such that~$v_i\prec u$. Now, determine whether the resulting graph has a matching that saturates~$\{w_1,\ldots,w_q\}$. If yes, then this matching corresponds to another $k$-path~$P^*$ in~$G'$ with the same signature, since the matching defines an extension of the signature. By construction, the~$I$-sequence~$(u_1,\ldots,u_q)$ of this path fulfills~$v_j=u_j$ for all~$j<i$ and~$u_i\prec v_i$. Conversely, if such a path~$P^*$ exists, then the edges of~$\hat{M}$ occupied by this path form a matching that saturates~$\{w_1,\ldots,w_q\}$. Hence, the test is correct for~$i$ which shows the correctness of the overall procedure.

  To see the running time, observe that~$\hat{H}$ and~$\hat{M}$ can be computed in~$\OO(n'+m')$ time and have size~$\OO(n'+m')$ since~$\hat{H}$  only contains~$q$ vertices of~$X\choose 2$. Afterwards, for each~$i\in [q]$, we first need to modify~$\hat{H}$ and~$\hat{M}$ as described above, which can be clearly done in~$\OO(n'+m')$ time. Now the remaining time is spent on computing a maximum matching and checking whether it saturates~$\{w_1,\ldots,w_q\}$. Since we are given a matching of size~$|q|-1$ (the matching $\hat{M}\setminus \{w_i,v_i\}$) this needs only the computation of one augmenting path using a standard algorithm for computing bipartite matchings. Since  augmenting paths can be computed in~$\OO(n+m)$ time, we obtain a total running time of~$\Oh(n+m)$ for the test for each~$i$. Since~$q<k$, the total running time of the algorithm amounts to~$\OO(k\cdot (n+m))$.
\end{proof}
\fi

\iflong
\subsection{The Solution-Lifting Algorithm}
\label{sec:vc-challenge-resolve}
\else
\paragraph{Solution-Lifting Algorithm.}
\fi
Given a $k$-path $P$ of $G'$, let $F(H, P)$ be the set of edges from $H$ that are occupied by~$P$.
It follows from \Cref{lemma:edge-occupation-vc} that $F(H, P)$ forms a matching in $H$.
Without loss of generality, assume that~$P$ is the lexicographically smallest $k$-path for the signature~${\rmf}_P: V(P) \cap (X \cup \RR) \rightarrow [k]$.
Let~$\cA$ denote the endpoints of $F(H, P)$ that are in ${{X}\choose{2}}$.
If we are to enumerate the collection of all $k$-paths in $G$ that intersect $V(G) \setminus V(G')$ and are equivalent to $P$, then every enumerated $k$-path $P^*$ must occupy edges that also form a matching in $H$ and satisfy the conditions (i) and (ii) of \Cref{defn:vc-equivalent-k-paths}.
The following lemma states how we can enumerate such $k$-paths of~$G$.

\begin{lemma}
\label{lemma:vc-enum-k-path-poly-delay}
Let $P$ be a $k$-path of $G'$.
% for the signature ${\rmf}: V(P) \cap (X \cup \RR) \rightarrow [k]$.
Then, there exists an algorithm with delay $\OO(n\cdot m\cdot k^2)$ that enumerates all the $k$-paths $P^*$ of $G$ exactly once such that
\iflong
\begin{enumerate}[(i)]
	\item $P^*$ contains at least one vertex from $G$ that is not in~$G'$, and%\todo{I think the unique non-kernel vertex can also be a pendant vertex}
	\item $P^*$ is equivalent to $P$.
\end{enumerate} 
\else
\textbf{(i)} $P^*$ contains at least one vertex from $G$ that is not in~$G'$, and \textbf{(ii)} $P^*$ is equivalent to $P$.
\fi
\end{lemma}
\begin{proof}
Let $P$ be a $k$-path of $G'$ with signature ${\rmf}_P: V(P) \cap (X \cup \RR) \rightarrow [k]$.
We invoke \Cref{lemma:edge-occupation-vc} to compute the edges $F(H, P)$ of~$H$ that are occupied by~$P$ in~$\OO(k+n+m)$~time.
For each position~$i$ in~$P$ that is not used by a vertex in~$V(P)\cap (X \cup \RR)$ and each vertex~$u\in V(G)\setminus V(G')$, we enumerate all paths~$P^*$ in~$G$ that are equivalent to~$P$ and where~$\pos_{P^*}(u)$ is the first vertex of~$P^*$ that is not contained in~$V(G')$. 
For this, we build an auxiliary graph~$H^*$ with bipartition~$A^*$,~$B^*$ based on~$H$ and then enumerate all maximal matchings in~$H^*$.  
The graph~$H^*$ is constructed as follows:
\begin{itemize}
\item First, in~$H^*$ we keep only those vertices~$\{a,b\}\in \binom{X}{2}$ where~$a$ and~$b$ are connected by a vertex~$w$ outside of~$X \cup \RR$ in~$P$. 
Note that~$w\in B'$ for each such vertex.
These vertices are added to~$A^*$.

\item If~$1<i<k$ and~$\pos_P(a)=i-1$ and~$\pos_P(b)=i+1$, then keep~$u$ as the only neighbor of~$\{a,b\}$ in~$H^*$; otherwise, discard the choice of~$u$ and~$i$. 
Additionally, if~$i=1$ check if~$u$ is adjacent to the second vertex~$x\in X$ of~$P$. 
If this is not the case, then  discard the choice of~$u$ and~$i$. 
Symmetrically, if~$i=k$, then check if~$u$ is adjacent to the penultimate vertex~$y\in X$ of~$P$. If this is not the case, then discard the choice of~$u$ and~$i$.

\item For each vertex~$\{a,b\}$ in~$H^*$ with~$\pos_{P}(a)<\pos_P(b)<i$ remove the edges between~$\{a,b\}$ and any vertex~$v$ such that~$v\notin V(G')$ or~$v\in\RR$. 
In other words, only keep the edges~$(\{a,b\},v)$ where~$v\in B'$.

\item For each vertex~$\{a,b\}$ in~$H^*$ with~$i<\pos_{P}(a)<\pos_P(b)$ remove the edges between~$\{a,b\}$ and any vertex~$v$ such that~$v\in\RR\cup \{u\}$, that is, we only keep those edges where~$v\in\widehat{B}\setminus\{u\}$.

\item Finally, if~$i\neq 1$ and the first vertex of~$P$ is not from~$X\cup \RR$, then let~$x\in X$ denote the second vertex of~$P$ and add a new vertex~$\{x,x^*\}$ to~$H$ and make~$\{x,x^*\}$ in~$H^*$ adjacent to any vertex from~$I \setminus (\RR \cup \{u\})$ that is in~$G$ a neighbor of~$x$.
 Similarly, if $i\neq k$ and the last vertex of $P$ is not from $X\cup \RR$, then let $y\in X$ denote the penultimate vertex of $P$. 
 Add a new vertex $\{y,y^*\}$ to~$H^*$ and make $\{y,y^*\}$ adjacent to every vertex of $I \setminus (\RR \cup \{u\})$ that is a neighbor of~$y$. If added,~$\{x,x^*\}$ and~$\{y,y^*\}$ belong to~$A^*$.
 
\end{itemize}
\iflong 
Now, for each such graph~$H^*$, we enumerate all matchings that saturate all vertices in~$A^*$. Using the algorithm of \Cref{prop:maximum-matching-enumeration},
 for each enumerated matching~$M$, we output the following path~$P^*$: The path~$P^*$ is defined via~$\pos_{P^*}$. 
 \else 
 Now, for each such graph~$H^*$, we enumerate all matchings that saturate all vertices in~$A^*$. This can be achieved by enumerating all maximum matchings which can be done with $\OO(n\cdot m)$-delay~\cite[paragraph above Theorem~13]{KobayashiKW22}, \cite[Theorem~18]{KobayashiKW21} on graphs with~$n$ vertices and~$m$ edges. 
 For each enumerated matching~$M$, we output the following path~$P^*$, defined via~$\pos_{P^*}$. 
 \fi

 For the vertex~$u$, we set~$\pos_{P^*}(u):=i$. The path~$P^*$ has the same signature as~$P$, that is, for each vertex~$v\in X\cup \RR$, we set~$\pos_{P^*}(v):=\pos_{P}(v)$. For each edge~$\{\{a,b\},v\}\in M$ with~$a,b\in X$, assuming~$\pos_P(a)<\pos_P(B)$, we set~$\pos_P^*(v)=\pos_P(a)+1$. For the edge~$\{\{x,x^*\},v\}\in M$, if it exists, we set~$\pos_{P^*}(v)=1$. 
 For the edge~$\{\{y,y^*\},v\}\in M$, if it exists, we set~$\pos_{P^*}(v)=k$. 
 Now, $\pos_{P^*}$ indeed gives a~$k$-path: since~$A^*$ is saturated by~$M$ and since the signatures of~$P^*$ and~$P$ agree, the codomain of~$\pos_{P^*}$ is~$[k]$. 
 Moreover, by the construction of~$H^*$ and the fact that~$M$ is a matching, every vertex that is not in~$X\cup \RR$ is adjacent to its predecessor and successor (except for the first and last vertex if they are not in~$X\cup\RR$).

  Thus, every outputted path is equivalent to~$P$ and by the inclusion of~$u$ it contains at least one vertex that is not from~$G'$.  
  \iflong
  It remains to show that
  no path is output twice and that every path that is equivalent to~$P$ and contains some vertex~$u$ that is not in~$G'$ is output. To this end consider a path~$P^*$ output by the algorithm. 
  Let~$u$ denote the first vertex of~$P^*$ that is not in~$G'$ and assume~$\pos_{P^*}(u)=i$. 
  Now~$P^*$ is not output for some other choice of~$u'\notin V(G')$ and~$i'$: If~$i'<i$, then~$u$ is not the first vertex that is not from~$V(G')$ in~$G$ because~$u'$ is fixed to be at position~$i'$. Otherwise, if~$i'>i$, then the construction of~$H^*$ ensures that the vertex~$u$ is not adjacent to any vertex~$\{a,b\}$ with~$\pos_P(b)<i'$. Hence,~$u$ is not at position~$i$ in the graphs enumerated for~$u'$ and~$i'$. Finally, two different paths enumerated for the particular choice of~$u$ and~$i$ differ in at least one matching edge and thus there is at least one position where the two paths differ.  

To see that~$P^*$ is indeed output, note that the vertices of~$V(P^*)\setminus (X\cup \RR)$ have different predecessors and successors in~$P^*$. Thus, the edge set~$M$ obtained by adding for each~such vertex~$v\neq u$ an edge~$\{\{a,b\},v\}$ if~$\pos_{P^*}(a)+1=\pos_{P^*}(v) =\pos_{P^*}(b)-1$, the edge~$\{\{x,x^*\},v\}$
if $\pos_{P^*}(v)=1$ and $\pos_{P^*}(x)=2$, and the edge~$\{\{y,y^*\},v\}$
if $\pos_{P^*}(y)=k-1$ and $\pos_{P^*}(v)=k$ is a subset of~$E(H^*)$ and, by construction a matching. Thus,~$M$ is output \iflong at some point \fi during the matching enumeration for~$u$ and~$i$. The path constructed from~$M$ is precisely~$P^*$.

\else

The proof that no path is output twice is deferred to the full version.

\fi

It remains to show the bound on the delay.
\iflong 
Observe that there are $k$~choices for~$i$ and up to $n$~choices for~$u$.
For each fixed pair~${i,u}$, we now use \Cref{prop:maximum-matching-enumeration} yielding a delay of $\OO(n\cdot m)$.
Kobayashi et al.~\cite[paragraph above Theorem~13]{KobayashiKW22}, \cite[Theorem~18]{KobayashiKW21} argue that at most $n$~times an augmenting path in $\OO(m)$~time needs to be found. 
 \else 
 Observe that there are $k$~choices for~$i$ and up to $n$~choices for~$u$.
For each fixed pair~${i,u}$, we now use the algorithm of Kobayashi et al.~\cite[paragraph above Theorem~13]{KobayashiKW22}, \cite[Theorem~18]{KobayashiKW21} yielding a delay of $\OO(n\cdot m)$.
They argue that at most $n$~times an augmenting path in $\OO(m)$~time needs to be found. 
 \fi 
In other words, the delay can also be bounded by~$\OO(\nu(G)\cdot m)$, where~$\nu(G)$ is the size of a maximum matching.
In our case~$\nu(G)=k$ since~$A^*$ contains at most $k$~vertices.
Thus, the overall delay is $\OO(n\cdot m\cdot k^2)$.
\end{proof}

\iflong
We are now ready to prove our main theorem statement.
\fi

{\thmOne*}

\begin{proof}
Our enumeration kernelization has two parts.
The first part is the kernelization algorithm and the second part is the solution-lifting algorithm.
Without loss of generality, let $(G, X, k)$ be the input instance such that $X$ is a vertex cover of $G$ having at most $2\cdot{\vc}(G)$ vertices.
The kernelization algorithm described by the {\em marking scheme} outputs the instance $(G', X, k)$.
%We can encode the vertices using $\OO(\log_2 {\vc})$ bits and
By \Cref{obs:k-path-preserve}, $G'$ has $\OO(|X|^2)=\OO({\vc}^2)$ vertices.

Our solution-lifting algorithm works as follows.
By construction, $(\widehat{B} \setminus B') \cup (I_1 \setminus I_1^+)$ are the only vertices of $G$ that are not in $G'$.
Let $P'$ be a $k$-path of $(G', X, k)$ such that $X' = (X \cup \RR) \cap V(P')$.
Moreover, let ${\rmf}_{P'}: X' \rightarrow [k]$ be the signature of $P'$.
Our first step is to output $P'$ itself. 
Afterwards, we invoke \Cref{obs:vc-lexicographically-smallest} and check whether $P'$ is the lexicographically smallest path of $G'$ having the signature same as ${\rmf}_{P'}: X' \rightarrow [k]$.
If $P'$ is the lexicographically smallest, then we invoke \Cref{lemma:vc-enum-k-path-poly-delay} to enumerate all the $k$-paths that are equivalent to $P'$ and contain at least one vertex that is not from $V(G')$.
Otherwise, if~$P'$ is not the lexicographically smallest for the signature ${\rmf}_{P'}: X' \rightarrow [k]$, 
then the enumeration algorithm outputs only the path $P'$.

Since the enumeration algorithm invoked by \Cref{lemma:vc-enum-k-path-poly-delay} runs with a delay of~$\OO(n\cdot m\cdot k^2)$ and it never enumerates any other path $P \neq P'$ in $G$ that is contained in $G'$, this completes the proof that we have a polynomial-delay enumeration kernelization with $\OO({\vc}^2)$ vertices.
\end{proof}

\iflong
\subsection{Extension to Other Path and Cycle Variants}
\label{subsec:vc-other-path-cycles}

%\paragraph{Extension to Other Path and Cycle Variants.}
We extend our positive results to {\enumkcycle} and variants where all paths/cycles of length {\em at least} $k$ need to be outputted.
In each of these variants, our adaptation exploits a fundamental property that a path/cycle of length at least $k$ has at most $2{\vc} + 1$ vertices.
Because a path of $G$ can have have all the vertices from a vertex cover $X$ and at most $|X| + 1$ vertices from $G - X$.
Hence, in the kernelization algorithm will preserve the length of a longest path.
Therefore, we can prove the following result.

\begin{corollary}
\label{cor:vc-at-least-k-path}
{\enumkpathAll} parameterized by ${\vc}$ admits a $\OO(n\cdot m\cdot k^2)$-delay enumeration kernel with $\OO({\vc}^2)$ vertices.
\end{corollary}

\begin{proof}
Since the above argumentation about the path length holds true, the marking scheme and kernelization algorithm works similarly as before. 
For every $\ell \geq k$, given an $\ell$-path $P'$ of~$G'$ (the kernel), the solution-lifting algorithm outputs a collection of $\ell$-paths of $G$ that are equivalent to $P'$.
Rest of the arguments remain the same.
\end{proof}

Now, we explain how the proof of \Cref{thm:k-path-vc-result} can be adopted to a polynomial-delay enumeration kernel for {\enumkcycle} parameterized by ${\vc}$.

\begin{corollary}
\label{corollary:vc-k-cycle-result}
{\enumkcycle} parameterized by ${\vc}$ admits a $\OO(n\cdot m\cdot k^2)$-delay enumeration kernel with $\OO({\vc}^2)$ vertices.
\end{corollary}

\begin{proof}
Our kernelization algorithm works similar as the proof of \Cref{thm:k-path-vc-result}.
But, the objective is to enumerate all cycles of length exactly $k$, so, we do not have to store any pendant vertex in the kernel.
The rest of the ideas work similarly as before and the solution-lifting algorithm also works similarly as before.
\end{proof}

Finally, combining the adaptations for {\enumkpathAll} and {\enumkcycle} can be utilized to prove the following result.

\begin{corollary}
\label{corollary:k-cycle-at-least-enum}
{\enumkcycleAll} admits a $\OO(n\cdot m\cdot k^2)$-delay enumeration kernelization with $\OO({\vc}^2)$ vertices.
\end{corollary}
\else

% \paragraph{Extension to Other Path and Cycle Variants.}

Our techniques can also be used for {\enumkcycle} and variants where all paths/cycles of length {\em at least} $k$ need to be outputted.
% In each of these variants, we exploit the fact that a path/cycle of length at least $k$ has at most $2{\vc} + 1$~vertices, because a path of $G$ can have have all the vertices from a vertex cover $X$ and at most $|X| + 1$ vertices from $G - X$.
% Hence, in the kernelization algorithm will preserve the length of a longest path.

\begin{proposition}[$\star$]
{\enumkpathAll}, {\enumkcycle}, and {\enumkcycleAll} parameterized by ${\vc}$ admits a $\OO(n\cdot m\cdot k^2)$-delay enumeration kernel with $\OO({\vc}^2)$ vertices.
\end{proposition}

\fi

\section{Parameterization by Dissociation Number}
\label{sec:k-path-disoc}

%\paragraph%%Dissociation number result

In this section, we present a polynomial-delay enumeration kernel for {\enumkpath} parameterized by the dissociation number~$\diss(G)$ of the input graph~$G$.
Given the instance~$(G,k)$, we initially find a $3$-approximation~$X$ of a minimal dissociation set.
Clearly,~$X$ can be found in polynomial time by greedily adding vertex subsets of size~$3$ whose induced subgraph is connected.
Let~$I= V(G)\setminus X$.
By construction,~$G[I]$ consists of connected components of size at most~$2$.
Furthermore, we assume without loss of generality that the vertices of~$I$ are assigned some label with an ordering in the input graph $G$.
When we construct our kernel with graph~$G'$, for the vertices in~$I\cap V(G')$, we maintain the relative ordering of the vertices of these vertices.
This will be utilized later in the solution-lifting algorithm.

The key observation in our kernelization is that since~$G[I]$ consists of connected components of size at most~$2$, the length~$k$ of the path can be upper-bounded by~$3\cdot(|X|+1)$.

\paragraph{Strategy.}
In our marking scheme we mark connected components of~$G[I]$ as \emph{rare} and \emph{frequent} such that all possible structures of a $k$-path can be preserved in the kernel.
Roughly speaking, a connected component~$C$ of~$G[I]$ is rare, if the number of connected components in~$G[I]$ having the same neighborhood as~$C$ in~$X$ is 'small' (linear bounded in~$|X|$), and otherwise~$C$ is frequent.
We mark all vertices in rare components and sufficiently many vertices (linear in~$|X|$) in frequent components.
This ensures that the number of marked vertices is cubic in~$|X|$.
Parameter~$k$ is not changed.

Next, we define the \emph{signature}~$\rmf_{P'}$ of each $k$-path~$P'$ in~$G$ as the mapping of the rare vertices and the vertices in the modulator~$X$ of~$P'$ to their indices in~$P'$ and an \emph{extension} of~$P'$ is another mapping from frequent vertices to the indices of~$P'$ such that the combination of both mappings yields a $k$-path in~$G$ which contains at least one non-kernel vertex.
Based on the signature, we construct equivalence classes~$\PP_i$ of $k$-paths in~$G$ and equivalence classes~$\PP_i'$ of $k$-paths in~$G'$.
We then verify that our kernel does not miss any equivalence class~$\PP_i$ of~$G$ and vice versa, that is, that there is a one-to-one correspondence between the equivalence classes of~$\PP_i$ and~$\PP_i'$.
Afterwards, we define a \emph{suitable} $k$-path in each equivalence class~$\PP_i'$ of~$G'$.
Then, if~$P_i'\in\PP_i'$ is not suitable, we only output~$P_i'$, and if~$P_i'$ is suitable we output all $k$-paths in~$\PP_i$ which contain at least one non-kernel vertex, that is, all $k$-paths of~$\PP_i\setminus \PP_i'$.

\subsection{Marking Scheme and Kernelization}

\paragraph{Marking Scheme.}
Our kernel consists of~$X$ and some connected components of~$G[I]$. 
In order to detect the necessary connected components of~$G[I]$, we invoke the following marking scheme distinguishing \emph{rare} and \emph{frequent} components.
Then, the kernel comprises of all rare components and sufficiently many frequent components.

Formally, let~$p=2\cdot |X|+3$ and let~$u,w\in X$.
If~$u$ and~$w$ have at most $p$~common neighbors in~$I$, then mark all vertices in the corresponding connected components of~$G[I]$ as \emph{rare} and we say that the pair~$(u,w)$ is \emph{$1$-rare}.
Otherwise, the pair~$(u,w)$ is \emph{$1$-frequent}.
If for~$u$ and~$w$ there exist at most $p$~connected components~$C_i=\{x_i,z_i\}$ of~$G[I]$ such that~$G[\{u,w,x_i,z_i\}]$ contains a path~$(u,x_i,z_i,w)$ or~$(w,x_i,z_i,u)$, then mark all vertices in the corresponding connected components of~$G[I]$ as \emph{rare} and we say that the pair~$(u,w)$ is \emph{2-rare}.
Otherwise, the pair~$(u,w)$ is \emph{2-frequent}.
We also say that~$C_i$ is $(u,w)$-connecting.
Finally, if~$u$ has at most $p$~neighbors in~$I$, then mark all vertices in the corresponding connected components of~$G[I]$ as \emph{rare} and we say that~$u$ is \emph{0-rare}.
Otherwise, $u$ is a \emph{0-frequent} vertex.
To unify later arguments, we abuse notation and say that a pair~$(u,w)$ is 0-rare or 0-frequent. 
In this case, vertex~$w$ is ignored and thus only vertex~$u$ is considered.

Note that the 0-rare and 0-frequent vertices are necessary for the start and end of a $k$-path~$P$ since~$P$ might start or end with vertices from~$I$.
By~$\RR\subseteq I$ we denote the set of all \emph{rare} vertices.
All remaining vertices~$\FF$ of~$I$ are called \emph{frequent}.
In other words,~$V(I)=\FF\cup\RR$ and~$\FF\cap \RR =\emptyset$.

\paragraph{Kernelization.}
To obtain our kernel we exploit the marking scheme as follows:
We consider the graph~$G'$ induced by all vertices of~$X$, all rare vertices, and some of the frequent vertices.
More precisely, for each pair~$(u,w)$ of vertices from the modulator which is $1$-frequent we first chose $p$~arbitrary common neighbors of~$u$ and~$w$ in~$G[I]$ and second, we add the vertices of the connected components of these vertices in~$G[I]$ to~$G'$.
Similarly, for each $2$-frequent pair~$(u,w)$ we add the vertices of exactly $p$~arbitrary $(u,w)$-connecting components of~$G[I]$ to~$G'$.
Finally, for each~$u\in X$ with at least $p$~neighbors in~$I$, we add the vertices of exactly $p$ connected components of~$G[I]$ containing at least one neighbor of~$u$ to~$G'$.
To conclude:~$V(G')=X\cup I'$ where~$\RR\subseteq I'\subseteq I$.

Note that it is necessary to set~$p=2\cdot |X|+3$ since a connected component~$C$ of~$G[I]$ can contain two marked frequent common neighbors of two vertices in the modulator and simultaneously, $C$ can be $(u,w)$-connecting for another pair~$(x,z)$ of vertices in the modulator.
Now, if~$C$ is used between~$x$ and~$z$ in a $k$-path, then two common neighbors are blocked for one pair of vertices in the modulator.
This possibility requires us to set~$p=2\cdot |X|+3$ instead only~$p=|X|+2$.

We observe the following for the kernel~$G'$.

\begin{observation}
\label{obs-kernel-diss}
Let~$G'$ be the graph obtained from~$G$ after invoking the marking scheme described above.
Then,~$G'$ has $\OO(|X|^3)$~vertices and~$G'$ can be computed in $\OO(|X|^2\cdot n)$~time.
Furthermore, any $k$-path of~$G'$ is also a $k$-path of~$G$.
\end{observation}

\begin{proof}
First, we show the bound on the number of vertices.
Observe that for each pair~$(u,w)$ of vertices from~$X$, we mark at most $4p$~rare vertices to~$G'$: at most $2p$~vertices if~$(u,w)$ is $1$-rare (up to $p$~common neighbors and up to~$p$ remaining vertices in the connected components of these vertices in~$G[I]$) and at most $2p$~vertices if~$(u,w)$ is $2$-rare.
Hence, we mark at most $4p\cdot |X|^2$~vertices as rare and all of them are part of~$G'$.
Similarly, for each pair~$(u,w)$ we add at most $4p$~frequent vertices to~$G'$: $2p$~vertices, if~$(u,w)$ is $1$-frequent and $2p$~many if~$(u,w)$ is $2$-frequent.
Furthermore, we add at most $2p$~vertices for each vertex~$u$ in the modulator to~$G'$; at most $p$~neighbors of~$u$ and at most $p$~other vertices in the corresponding connected components.
Hence, we add at most $4p\cdot |X|^2+2\cdot |X|$~vertices to~$G'$.
Since~$p=2\cdot |X|+3$, we obtain that~$|V(G')|\in\OO(|X|^3)$.
The time bound follows by checking for each pair of vertices from~$X$ whether they are 0-rare, 1-rare, or 2-rare, respectively.

The second part is a direct consequence from the fact that~$G'$ is an induced subgraph of~$G$.
\end{proof}

The choice of~$p$ gives a guarantee that if two vertices~$v_\ell$ and~$v_r$ of the modulator are $t$-frequent for some~$t\in[0,2]$, then independent of the structure of any $k$-path~$P$ in~$G$, there is at least one $(v_\ell,v_r)$-path on $t+2$~vertices in~$G[I'\cup\{v_\ell,v_r\}]$ whose internal vertices, that is, the vertices in~$G[I']$, are not contained in~$P$.
We denote this property as the \emph{prolongation property} that we formally define as follows.
Suppose that a given $k$-path $P$ of $G$ (respectively of $G'$) fulfills the following property -- for every pair of vertices $(x, y) \in {{X}\choose{2}}$, if $(x, y)$ is $t$-frequent for some $t \in [0, 2]$, then there is at least one path $Q_{x, y}$ having $t+2$ vertices between~$x$ and $y$ in $G[\{x, y\} \cup I']$ such that the internal vertices of $Q_{x, y}$ are not contained in~$P$.
If a path $P$ fulfills this property, then $P$ is said to fulfill {\em prolongation property}.
%This property holds true for every $k$-path $P$ of $G'$ as well.
This property will be crucial in the solution-lifting algorithm for guaranteeing polynomial delay.

\begin{lemma}
\label{lemma:diss-prolongation-property}
Every $k$-path~$P$ of $G$ (respectively of $G'$) fulfills the prolongation property.
\end{lemma}

\begin{proof}
We only give the proofs for the $k$-paths of $G$ since by \Cref{obs-kernel-diss} $k$-paths of~$G'$ are also $k$-paths of~$G$. 
Let $P$ be a $k$-path of $G$ and ~$v_{\ell}$ and~$v_r$ be two vertices of the modulator~$X$ such that all vertices in~$P$ between~$v_\ell$ and~$v_r$ are frequent.
In other words, assume that~$v_\ell$ and~$v_r$ are $t$-frequent for some~$t\in[0,2]$.
First, observe that since~$P$ contains at most $|X|$~vertices of~$X$, there are at most $|X|+1$~consecutive sequences of frequent vertices in~$P$ (recall that the frequent and rare vertices form a partition of~$I$).
Second, since each of these consecutive sequences comprises a connected component of~$G[I]$, which have at most $2$~vertices, in total at most $2\cdot |X|=p-1$~consecutive sequences corresponding to~$v_\ell$ and~$v_r$ are hit by other vertices from the modulator.
%The argumentation for 0-frequent vertices in~$X$ works analog.
Hence, the prolongation property is verified.
\end{proof}

\subsection{Signature and Equivalence Classes}

We have to prove some properties that establish some relations between all $k$-paths of~$G$ and all $k$-paths in~$G'$.
Similar to \Cref{sec:k-path-vc}, we define the mappings~${\rmf}$ and~${\rmg}$, and the notion of equivalent paths which will be used in the solution-lifting algorithm.

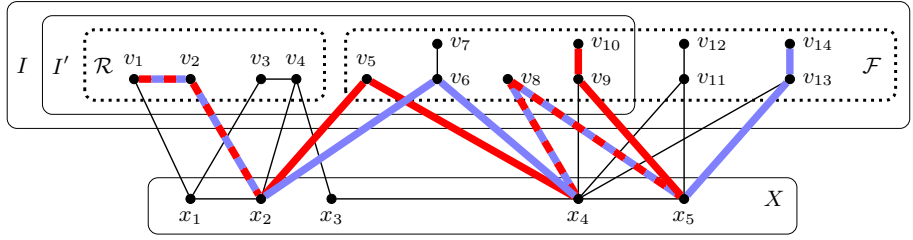
\begin{figure}[t]
\centering
\begin{adjustbox}{width=\textwidth}
\begin{tikzpicture}
\draw[rounded corners,dotted,very thick] (-1.0, 1.2) rectangle (2.4, 2.2) {};
\node[label=right:{$\RR$}](X) at (-1.1, 1.7) []{};
\draw[rounded corners,dotted,very thick] (2.7, 1.2) rectangle (10.5, 2.2) {};
\node[label=right:{$\FF$}](X) at (9.8, 1.7) []{};
\draw[rounded corners] (-1.6, 1.0) rectangle (6.8, 2.4) {};
\node[label=right:{$I'$}](X) at (-1.7, 1.7) []{};
\draw[rounded corners] (-2.1, 0.8) rectangle (10.7, 2.6) {};
\node[label=right:{$I$}](X) at (-2.2, 1.7) []{};

\node[label=above:{$v_1$}](V1) at (-0.3, 1.5) [shape = circle, draw, fill=black, scale=0.11ex]{};
\node[label=above:{$v_2$}](V2) at (0.5, 1.5) [shape = circle, draw, fill=black, scale=0.11ex]{};

\node[label=above:{$v_3$}](V3) at (1.5, 1.5) [shape = circle, draw, fill=black, scale=0.11ex]{};
\node[label=above:{$v_4$}](V4) at (2.0, 1.5) [shape = circle, draw, fill=black, scale=0.11ex]{};

\node[label=above:{$v_5$}](V5) at (3.0, 1.5) [shape = circle, draw, fill=black, scale=0.11ex]{};
\node[label=right:{$v_6$}](V6) at (4.0, 1.5) [shape = circle, draw, fill=black, scale=0.11ex]{};
\node[label=right:{$v_7$}](V7) at (4.0, 2.0) [shape = circle, draw, fill=black, scale=0.11ex]{};

\node[label=right:{$v_8$}](V8) at (5.0, 1.5) [shape = circle, draw, fill=black, scale=0.11ex]{};
\node[label=right:{$v_9$}](V9) at (6.0, 1.5) [shape = circle, draw, fill=black, scale=0.11ex]{};
\node[label=right:{$v_{10}$}](V10) at (6.0, 2.0) [shape = circle, draw, fill=black, scale=0.11ex]{};

\node[label=right:{$v_{11}$}](V11) at (7.5, 1.5) [shape = circle, draw, fill=black, scale=0.11ex]{};
\node[label=right:{$v_{12}$}](V12) at (7.5, 2.0) [shape = circle, draw, fill=black, scale=0.11ex]{};
\node[label=right:{$v_{13}$}](V13) at (9.0, 1.5) [shape = circle, draw, fill=black, scale=0.11ex]{};
\node[label=right:{$v_{14}$}](V14) at (9.0, 2.0) [shape = circle, draw, fill=black, scale=0.11ex]{};

\path [-,line width=0.2mm](V1) edge (V2);
\path [-,line width=0.2mm](V3) edge (V4);
\path [-,line width=0.2mm](V6) edge (V7);
\path [-,line width=0.2mm](V9) edge (V10);
\path [-,line width=0.2mm](V11) edge (V12);
\path [-,line width=0.2mm](V13) edge (V14);

\draw[rounded corners] (-0.1, -0.7) rectangle (9.1, 0.1) {};
\node[label=right:{$X$}](X) at (8.4, -0.2) []{};
\node[label=below:{$x_1$}](X1) at (0.5, -0.2) [shape = circle, draw, fill=black, scale=0.11ex]{};
\node[label=below:{$x_2$}](X2) at (1.5, -0.2) [shape = circle, draw, fill=black, scale=0.11ex]{};
\node[label=below:{$x_3$}](X3) at (2.5, -0.2) [shape = circle, draw, fill=black, scale=0.11ex]{};
\node[label=below:{$x_4$}](X4) at (6, -0.2) [shape = circle, draw, fill=black, scale=0.11ex]{};
\node[label=below:{$x_5$}](X5) at (7.5, -0.2) [shape = circle, draw, fill=black, scale=0.11ex]{};

\path [-,line width=0.2mm](X1) edge (V1);
\path [-,line width=0.2mm](X1) edge (V3);
\path [-,line width=0.2mm](X2) edge (V2);
\path [-,line width=0.2mm](X2) edge (V4);
\path [-,line width=0.2mm](X2) edge (V5);
\path [-,line width=0.2mm](X2) edge (V6);
\path [-,line width=0.2mm](X3) edge (V4);
\path [-,line width=0.2mm](X4) edge (V5);
\path [-,line width=0.2mm](X4) edge (V6);
\path [-,line width=0.2mm](X4) edge (V8);
\path [-,line width=0.2mm](X4) edge (V9);
\path [-,line width=0.2mm](X4) edge (V11);
\path [-,line width=0.2mm](X4) edge (V13);
\path [-,line width=0.2mm](X5) edge (V8);
\path [-,line width=0.2mm](X5) edge (V9);
\path [-,line width=0.2mm](X5) edge (V11);
\path [-,line width=0.2mm](X5) edge (V13);

\path [-,line width=0.2mm](X1) edge (X2);
\path [-,line width=0.2mm](X3) edge (X5);

\draw[dashed,dash pattern=on 5pt off 5pt,red,line width=3pt]  (V1) to (V2);
\draw[dashed,dash pattern=on 5pt off 5pt,blue!50,dash phase=5pt,line width=3pt]  (V1) to (V2);

\draw[dashed,dash pattern=on 5pt off 5pt,red,line width=3pt]  (V2) to (X2);
\draw[dashed,dash pattern=on 5pt off 5pt,blue!50,dash phase=5pt,line width=3pt]  (V2) to (X2);

\draw[dashed,dash pattern=on 5pt off 5pt,red,line width=3pt]  (X4) to (V8);
\draw[dashed,dash pattern=on 5pt off 5pt,blue!50,dash phase=5pt,line width=3pt]  (X4) to (V8);

\draw[dashed,dash pattern=on 5pt off 5pt,red,line width=3pt]  (V8) to (X5);
\draw[dashed,dash pattern=on 5pt off 5pt,blue!50,dash phase=5pt,line width=3pt]  (V8) to (X5);

\draw[red,line width=3pt]  (X2) to (V5);
\draw[red,line width=3pt]  (X4) to (V5);
\draw[red,line width=3pt]  (X5) to (V9);
\draw[red,line width=3pt]  (V9) to (V10);

\draw[blue!50,line width=3pt]  (X2) to (V6);
\draw[blue!50,line width=3pt]  (X4) to (V6);
\draw[blue!50,line width=3pt]  (X5) to (V13);
\draw[blue!50,line width=3pt]  (V13) to (V14);

\end{tikzpicture}
\end{adjustbox}

\caption{A graph~$G$ with modulator~$X$ such that each connected component in~$G[I]$ has at most 2~vertices.
The two paths~$P_1=(v_1,v_2,x_2,v_5,x_4,v_8,x_5,v_9,v_{10})$ (in \textcolor{red}{red}) and~$P_2=(v_1,v_2,x_2,v_6,x_4,v_8,x_5,v_{13},v_{14})$ (in \textcolor{blue!50}{blue}) have the same signature.
Furthermore, $P_1$ is entirely contained in~$G'$, while~$P_2$ contains the two non-kernel vertices~$v_{13}$ and~$v_{14}$.}
\label{fig-diss-example-signature}
\end{figure}

\paragraph{Definition of Signature.}
First, we define the signature~$\rmf$ of a $k$-path.

\begin{definition}
\label{def-diss-sig}
Let~$P$ be a $k$-path of~$G$ (respectively, of $G'$) and~$P = (v_1,\ldots,v_k)$ be the sequence of the $k$~vertices of the path~$P$.
A function~$\rmf_P: V(P) \cap (X\cup\RR) \rightarrow [k]$ is called a \emph{signature} of~$P$ if~$\rmf_P(v_i) = i$.
\end{definition}

Observe that~${\rmf}$ is an injective mapping.
Clearly,~${\rmf}_{P}$ can be computed in linear-time.
See \Cref{fig-diss-example-signature} for an example of the signature of a path.

\paragraph{Equivalent paths and Equivalence classes of $k$-paths.}
Now, we define equivalence classes of the $k$-paths based on the mappings~$\rmf$ and~$\rmg$.

\begin{definition}
\label{def-diss-ext}
Let~$P_1$ and~$P_2$ be two $k$-paths.
Then, $P_1$ is \emph{equivalent} to~$P_2$ if and only if~$\rmf_{P_1}=\rmf_{P_2}$.
\end{definition}

The above definition defines the notion of equivalence not only between two $k$-paths of~$G$, but also between two $k$-paths of~$G'$, and between one $k$-path of $G$ and one $k$-path of $G'$, respectively.

Since the above definition of equivalent $k$-paths provides an equivalence relation, we are now able to define the equivalence classes for the set of $k$-paths of~$G$.
Let~$\PP$ be the set of all $k$-paths of~$G$ and let $\PP=\PP_1\cup\PP_2\cup\ldots\cup\PP_s$ be the partition of $\PP$ into $s$ equivalence classes where $s$ is the number of equivalence classes.
By~$\PP_i'\subseteq \PP_i$ we denote all $k$-paths of~$G$ in class~$\PP_i$ which are also $k$-paths in the kernel~$G'$.
In other words, $\PP_i'$ is the collection of all $k$-paths in $G'$ such that $\PP_i' \subseteq \PP_i$.

\paragraph{Definition of Extension.}
Next, we define the mapping~${\rmg}$.
Intuitively, for a given signature~${\rmf}_{P}$ an \emph{extension} consists of a set~$I^*\subseteq \FF\subseteq I$ of vertices which are also mapped to indices in~$[k]$ such that the combined function is bijective, that is, each vertex is mapped to a unique index of~$[k]$ and these $k$~vertices form a $k$~path in that specific ordering in~$G$.
%Furthermore, such a path is required to contain at least one vertex which is not contained in the kernel.
%By definition such a vertex is contained in~$I$ and has to be frequent, that is, has to be contained in~$\FF$.

\begin{definition}
For some~$I^* \subseteq I$ we call a mapping~$\rmg: I^* \rightarrow [k]$ an {\em extension} of the signature~$\rmf_{P}$ of~$P$ if 
\begin{itemize}
	\item $\rmg$ is an injective function,
	%%%\item $I^* \setminus V(G')=I^*\setminus I' \neq \emptyset$,
	\item ${\rmg}(u)\ne{\rmf}_P(w)$ for each~$u\in I^*$ and each~$w\in V(P)\cap (X\cup\RR)$, and
	\item the vertices in~$(V(P)\cap (X\cup\RR))\cup I^*$ form a $k$-path in~$G$ in the order specified by~$\rmf_{P}$ and~$\rmg$.
\end{itemize}
\end{definition}

\subsection{Challenges}

The objective of the solution-lifting algorithm is that ``for each equivalence class $\PP_i$ we choose a unique \emph{suitable} $k$-path $P_i'\in\PP_i'$ that is used to enumerate all $k$-paths in~$\PP_i\setminus \PP_i'$ with polynomial-delay and also~$P_i'$''.
For every other $k$-path~$\widetilde{P_i} \in \PP_i' \setminus \{P_i'\}$, the solution-lifting algorithm only outputs $\widetilde{P_i}$.
To make this algorithm work, we have to overcome the following challenges:

\begin{enumerate}
	\item Verify that~$\PP_i'\ne\emptyset$, that is, for each equivalence class~$\PP_i$ there exists at least one $k$-path $P'$ which is entirely contained in the kernel, that is, in $G'$.
	\item Define the notation of \emph{suitable} $k$-paths of $G'$.
	\item Given a $k$-path~$P_i'\in\PP_i'$, check whether $P_i'$ is a suitable $k$-path or not in polynomial time.
	\item Ensure polynomial delay for the enumeration of all $k$-paths in~$\PP_i\setminus \PP_i'$.
\end{enumerate}

\subsection{Resolving the Four Challenges}
\label{sec-diss-solve-challenges}

This section is devoted to describe how we resolve each of the above mentioned challenges.
The following lemma illustrates a crucial structural characterization that resolves the first challenge.

\begin{lemma}
\label{lemma:dissociation-challenge-1-eqv-class}
If $\PP_i$ is a nonempty equivalence class of $k$-paths in $G$, then $\PP_i' \neq \emptyset$.
\end{lemma}

\begin{proof}
Let $P_i\in\PP_i$.
We prove the lemma by showing that there is a $k$-path $P_i'$ in $G'$ that is equivalent to $P_i$.
This is sufficient to prove that $\PP_i' \neq \emptyset$.
By definition for each $P_i,\widehat{P_i}\in\PP_i$, ${\rmf}_{P_i} = {\rmf}_{\widehat{P_i}}$.
Therefore, both $P_i$ and~$\widehat{P_i}$ coincide on all non-frequent vertices.
It means that we have~$Z=V(P_i)\setminus \FF=V(\widehat{P_i})\setminus \FF$ where $\FF$ is the set of all frequent vertices.
And also for each~$w\in Z$, $\pos_{P_i}(w)=\pos_{\widehat{P_i}}(w)$.
This implies that both~$P_i$ and~$\widehat{P_i}$ have the same sequences of consecutive indices such that the vertices at these indices stem from~$\FF$ (frequent vertices).
Let~$g$ now be such a sequence of length~$t\in[0,2]$ and let~$v_\ell\in P_i\cap (X\cup\RR)$ be the vertex left of this sequence and let~$v_r\in P_i\cap (X\cup\RR)$ be the vertex right of this sequence (by definition these vertices are in the modulator~$X$).
Observe that the vertices~$v_\ell$ and~$v_r$ are $t$-frequent by definition.
Also, note that~$v_\ell$ or~$v_r$ might not exists, if~$P$ starts/ends with~$g$; in particular both might not exist if~$k\le 2$.
Due to \Cref{lemma:diss-prolongation-property}, $P_i$ and  $\widehat{P_i}$ admit prolongation property.
We exploit this prolongation property to construct~$P_i'$ as follows.
According to the construction of~$G'$ (kernelization algorithm), there are at least $q\ge p= 2\cdot |X|+3$ paths~$Q_1, Q_2, \ldots, Q_q$ in~$G'$ of length~$2+t$ from~$v_\ell$ to~$v_r$ that are internally vertex-disjoint.
Since each $k$-path consists of at most $|X|$~vertices from~$X$, there are at most $|X|+1$~sequences.
Since each such sequence contains at most $2$~vertices, at most $2\cdot |X|+2=p-1$~paths~$Q_i$ are hit by any $k$-path.
Hence, there is at least one path~$Q_j$ which is not hit.
Thus there exists also some~$P_i'\in\PP_i'$ which contains $Q_j$ as subpath.
\end{proof}

\paragraph{Defining the notion of suitable $k$-path in $G'$.}
Now, we move on to define the notion of ``suitable $k$-path'' in $G'$ that addresses the second challenge.
Recall that $G'$ (the graph outputted by the kernelization algorithm) is an induced subgraph of the input graph $G$.
We maintain the relative ordering of the vertices in~$I'=I\cap V(G')$ as they were in~$V(G)$.
It is possible for us to encode these vertices using $\OO(\log_2 ({\diss(G)}))$~bits and by construction,~$G'$ has $\OO(|X|^3)$~vertices, hence $\OO({\diss(G)}^3)$~vertices.
Therefore,~$G'$ can be encoded using $\OO({\diss(G)}^3\log_2 ({\diss(G)}))$~bits. 
We exploit this small encoding of the vertices in~$I'$ as follows.

Informally speaking, a $k$-path~$P_i'\in\PP_i'$ is \emph{suitable} if and only if the ordering of the frequent vertices of~$P_i'$ is minimal among all $k$-paths in the equivalence class~$\PP_i'$.
Formally, let~$\PP_i$ be an equivalence class of the $k$-paths in~$G$ and let~$\PP_i'\subseteq \PP_i$ be the corresponding equivalence class of $k$-paths in~$G'$ (recall that~$\PP_i'\ne\emptyset$ for each~$i$ according to Challenge~$1.$).
By definition of~$\PP_i'$, any two paths~$P_1',P_2'\in\PP_i'$ only differ in some frequent vertices.
Let~$P_j'\in\PP_i'$ be a $k$-path and $F_j$ be the sequence of frequent vertices appearing in $P_j'$ according to the corresponding extension~${\rmg}$.
Then, $\order(P_j')$ is the sequence of $|F_j|$~indices in which the $x$th entry corresponds to the index of the vertex in~$I'$ that appears at the $x$th position in~$F_j$. 
Note that each vector~$\order$ has the same length for a fixed equivalence class~$\PP_i'$.
Furthermore, observe that for any two distinct $k$-paths $P_1',P_2'\in\PP_i'$ we have~$\order(P_1')\ne \order(P_2')$

Based on~$\order$, we define an ordering~$\triangleleft$ of the $k$-paths in~$\PP_i'$.

\begin{definition}
Let $P_1',P_2'\in\PP_i'$.
We write~$\order(P_1')\triangleleft\order(P_2')$ if and only if~$\order(P_1')$ is lexicographically smaller than~$\order(P_2')$.
We say~$P_1'\triangleleft P_2'$ if and only if~$\order(P_1')\triangleleft\order(P_2')$.
\end{definition}

Informally speaking, $\order(P_1')$ is lexicographically smaller than $\order(P_2')$ if the sequence of indices appearing in $\order(P_1')$ is lexicographically smaller than the sequence of indices appearing in $\order(P_2')$.
Clearly,~$\triangleleft$ is a total ordering of the $k$-paths in~$\PP_i'$.
This allows us to define the suitable $k$-path of each equivalence class.

\begin{definition}
A path $P_i' \in {\PP}_i'$ is \emph{suitable} if and only if $P_i'$ is the minimal path of~$\PP_i'$ with respect to $\triangleleft$.
\end{definition}

Since~$\triangleleft$ is a total ordering, the minimal path is well-defined and unique.

\paragraph{Check whether a $k$-path $P_i'\in\PP_i'$ is suitable.}
Let~$P_i\in\PP_i'$ with ordering~$\order(P_i)$ of its frequent vertices.
We have to check whether there is another $k$-path~$P_i'\in\PP_i'$ such that~$\order(P_i')\triangleleft \order(P_i)$.
Informally, we try to find~$P_i'$ by replacing the frequent vertices of~$P_i$ by frequent vertices with smaller indices in~$G'$ in increasing order of positions.
Recall that each connected component of~$G[I]$ (and thus also~$G[I']$) is frequent or rare.
In the following we do \emph{not} replace each frequent vertex of~$P_i$ individually, instead we replace consecutive sequences of frequent vertices in~$P_i$ (recall that they either have length~$1$ or~$2$).
With this intuition, our next lemma illustrates a proof how we can check in polynomial time whether a given $k$-path $P_i$ is suitable or not.

\begin{lemma}
\label{lemma:algo-suitable-check}
Given a $k$-path $P_i \in \PP_i'$ of $G'$, there is an $\OO(|X|\cdot n)$-time algorithm that correctly decides whether $P_i$ is suitable or not.
\end{lemma}

\begin{proof}	
Let $P_i$ be a $k$-path in $G'$ and $P_i \in \PP_i'$ and let $F_i\subseteq\FF\subseteq I$ denotes the set of frequent vertices appearing in $P_i$ and let~$\order(P_i)$ be the sequence of~$F_i$ as the appear in~$P_i$.
Consider position~$j$ of~$\order(P_i)$, that is, the index of the vertex of~$F_i$ appearing at the $j$th position in $\order(P_i)$.
Additionally, we assume that we already verified that there is no lexicographically smaller $k$-path~$P_i'\in\PP_i'$ having a vertex with smaller index in~$I'$ in a position of~$\order(P_i)$ preceding~$j$.
Furthermore, let~$p_j$ be the corresponding $j$th vertex in $F_i$ according to~$\order(P_i)$ (in other words, the $j$th frequent vertex in~$P_i$).
Suppose that~$v_1\in V(P_i)\cap X$ is the vertex of~$P_i$ preceding~$p_j$ (recall that we want to replace a consecutive sequences of frequent vertices in~$G[I']$ in a single step and that all vertices of a connected component in~$G[I]$ are either all rare or all frequent).
Also, let~$v_2\in V(P_i)\cap X$ be the first vertex in the modulator~$X$ in~$P_i$ after~$p_j$.
Note that~$v_2$ is the vertex appearing directly after $p_j$ or  the vertex of~$P_i$ appearing two positions after~$p_j$ and this is determined by the ${\rmf}_{P_i}$.
In contrast, $v_1$ is the vertex of the modulator appearing in~$P_i$ such that~$p_j \in \FF$ is the vertex appearing directly after~$v_1$.
Here, the other option that~$p_j$ appears two positions after~$v_1$ is not possible since we replace a consecutive sequence of frequent vertices in a single step.
Also note that~$v_1$ or~$v_2$ might not exists if~$P_i$ starts/ends with frequent vertices; in particular if~$k\le 2$ both~$v_1$ and~$v_2$ might not exists.
Without loss of generality, assume these both~$v_1$ and~$v_2$ exist.
This implies that~$v_1$ and~$v_2$ are  $t$-frequent vertices for some $t \in [1, 2]$

Now, let~$U_j$ be the set of sequences of frequent vertices (which all have length 1 or 2 determined by the structure of~$P_i$) which can be plugged in for vertex~$p_j$ (and possibly also for its frequent successor in~$P_i$).
Note that if~$v_2$ is the successor of~$p_j$ then~$U_j$ is precisely the set of common frequent neighbors of~$v_1$ and~$v_2$ which are not already used in a position proceeding~$j$.
Formally, $U_j=\{w:w\in (N_{G'}(v_1)\cap N_{G'}(v_2)\cap \FF)\setminus (\bigcup_{j'<j}U_{j'})\}$ where~$U_{j'}$ is the set of frequent vertices used in position~$j'$ in~$\order(P_i)$.
%\todo[inline]{Diptapriyo: Got confused with $v_r, v_1, v_2$. I think we always consider the pair $v_1, v_2 \in X$ in these lines.}
Otherwise, if~$v_2$ is not the successor of~$p_j$ then~$U_j$ is defined similarly; $U_j$ is the set of all frequent connected components (having size~$2$) in~$G[I']$ such that one vertex is a neighbor of~$v_1$ and the other is a neighbor of~$v_2$ and none of these vertices is already contained in~$U_{j'}$ for any~$j'<j$.

For completeness, note that if~$v_1$ or~$v_2$ does not exist, then the definition of $U_j$ is analogous.
For example, if $v_1$ exists but $v_2$ does not, then $p_j$ is either the last vertex or the penultimate vertex of $P_i$ such that the vertex next to $p_j$ appears in the same component of $G[I']$ that is frequent.
If~$p_j$ is the last vertex of $P_i$, then $U_j$ is the set of frequent vertices adjacent to $v_1$ but not in $U_{j'}$ for any $j' < j$.
In case, $p_j$ is the penultimate vertex of $P_i$, then $U_j$ is the set of frequent components that are adjacent to $v_1$ but not in $U_{j'}$ for any $j' < j$.
If $v_1$ does not exist, but $v_2$ exists, then $p_j$ is either the first vertex of $P_i$ or the second vertex of $P_i$. The definition of $U_j$ is analogous.
If both $v_1, v_2$ do not exist, then  $k\le 2$.
Then, then~$U_1$ is any path of length~$t$ in~$G[I']$.

From the prolongation property (\Cref{lemma:diss-prolongation-property}), we on the one side conclude that~$U_j$ without the already used frequent vertices is non-empty, and on the other hand that for each position after~$j$ in~$\order(P_i)$ we can also find frequent vertices, that is, there are frequent vertices that can be used in the position $j$ of $\order(P_i)$.
This guarantees us that this partial solution, where we only plugged in frequent vertices until position~$j-1$, can always be extended to some $k$-path of~$G'$:
Note that since~$V(G')\in\OO(|X^3|)\in\OO(n)$ and each connected component in~$G[I']$ has size at most~$2$, we have~$U_j\in\OO(|X|^3)$ and~$U_j$ can also be computed in $\OO(|X|^3)\in\OO(n)$~time in both cases (whether the sequence has length~$1$ or~$2$).
Now, we simply have to check whether~$p_j$ has the smallest index in~$G'$ among all vertices in~$U_j$ (we check the first vertex in each sequence).
If~$p_j$ is the smallest, then we cannot replace~$p_j$ by a vertex with smaller index in~$G'$ and hence we consider the next position~$j+1$ in $\order(P_i)$.
Otherwise, if~$p_j$ is not the smallest, that is, there is a vertex~$q_j$ which has smaller index, then by the choice of the number of frequent vertices, we know that there is still a $k$-path and we have verified that~$P_i$ is \emph{not} suitable.
Since we have to invoke this test at most~$|X|+1$~times, we can check in $\OO(|X|\cdot|X|^3)\in\OO(|X|\cdot n)$~time whether~$P_i$ is suitable or not.
\end{proof}

\paragraph{Solution-Lifting Algorithm.}
Suppose we are given some $k$-path~$P'\in\PP_i'$.
First we invoke \Cref{lemma:algo-suitable-check} to check in $\OO(|X|^4)\in\OO(|X|\cdot n)$-time whether~$P'$ is suitable or not.
If~$P'\in\PP_i'$ is \emph{not} suitable, then we only output~$P'$ and no other $k$-path.
Clearly, this case is doable with constant delay.

It remains to consider the case that~$P'$ is suitable.
Similar to the proof ideas of \Cref{lemma:algo-suitable-check}, we again replace each consecutive sequences of frequent vertices in~$P'$ in a single step.
The enumeration is done in two phases. 
In Phase~1, we fix the smallest index~$i$ of~$\order(P')$ which is replaced by a frequent non-kernel vertex.
In Phase~2, we replace the remaining frequent vertices of~$P'$.
More precisely, in Phase~2.1, we replace the frequent vertices of~$P'$ in~$\order(P')$ starting at position~$i+1$ by frequent vertices (which might be contained in the kernel).
Note that we only require that the first such vertex we replace is a non-kernel vertex, that is, stems from~$V(G)\setminus V(G')$, to ensure that we enumerate a $k$-path which contains at least one non-kernel vertex and that we do not miss any $k$-path having the same signature and the same rare vertices at the same positions in~$P'$.
Then in Phase~2.2, we replace the frequent vertices of~$P'$ in~$\order(P')$ from position~$i-1$ down to position~$1$ by frequent vertices from the kernel.
Here, it is important to only choose kernel vertices to not violate our guess of index~$i$ in Phase~1.
More precisely, the replacement of the $j$th sequence works as follows: 
We insert one possible candidate sequence~$c_j$ and then we replace the next sequence. 
After all $k$-path in~$\PP_i$ having sequence~$c_j$ at position~$j$ are enumerated, we then consider the next candidate sequence~$c_{j'}$ for position~$j$.
This process ensures that all $k$-paths in~$\PP_i\setminus \PP_i'$ are enumerated.
%Note, that in one possibility we do not change any of these vertices to not miss a $k$-path of~$G$.
Finally, we also output path~$P'$.
The following lemma gives both the phases in detail.

\begin{lemma}
\label{lemma:solution-lifting-suitable-path}
Let $P'$ be a suitable $k$-path in $G'$.
There exists an enumeration algorithm that outputs $P'$ itself and all $k$-paths of $G$ that are equivalent to $P'$ and contain at least one vertex from $V(G) \setminus V(G')$ with $\OO(n\cdot |X|)$-delay.
\end{lemma}

\begin{proof}
Let~$s$ be the number of consecutive sequences of frequent vertices in~$P'$ and let~$g_i$ the $i$th such sequence. Recall that~$s\le|X|+1$ and~$g_i$ has length~1 or~2.
By definition, the vertex~$v^i_1$ proceeding the first vertex in~$g_i$ and the vertex~$v^i_2$ following the last vertex of~$g_i$ are both from the modulator~$X$.
Note that for~$g_1$ ($g_s$) vertex~$v^1_1$ ($v^s_2$) might not exists if~$P'$ starts with~$g_1$ (ends with~$g_s$). 

Now, we perform Phase~1:
For each~$i\in[s]$ we replace~$g_i$ by a sequence of frequent non-kernel vertices.
To do so, we find in $\OO(n)$-time all candidate consecutive sequences~$c_i$ of frequent non-kernel vertices of the same length as~$g_i$ such that the first vertex of~$c_i$ is adjacent to~$v^i_1$ and the last vertex of~$c_i$ is adjacent to~$v^i_2$.
Note that if one of~$v^i_1$ or~$v^i_2$ does not exist, then the corresponding adjacency requirement is not needed.
For example, if both~$v^i_1$ and~$v^2_i$ do not exists, then if~$g_i$ has length~1, then it is any non-kernel vertex and if~$g_i$ has length~2, then it is any size~2 component of~$G[I]$ not containing any kernel vertex. 
Afterwards, for each such candidate~$c_i$ we replace~$g_i$ with~$c_i$ to obtain a new $k$-path~$P_i'$ which is the input of Phase~2.
Clearly, $P_i'$ is a $k$-path in~$G$ but not in~$G'$.
Note that when we created path~$P_i'$ we directly go to Phase~2.
In this phase we output all $k$-paths of~$G$ which are not $k$-paths of~$G'$ which contain sequence~$g_j$ at th $j$th sequence.
Afterwards, we compute the next such path with Phase~1.

Now, we are in a position to start the Phase~2. In this phase, we replace the remaining sequences~$g_j$ with~$j\ne i$.
The Phase~2.1 works as follows.
We basically do the same as in Phase~$1$; but now the vertices of the candidates for~$c_j$ can also be contained in the kernel, that is, also the current candidate of~$P_i'$ used in this position is allowed.
But all vertices used to replace sequences~$g_{j'}$ for~$j'\in[i,j-1]$ cannot be used anymore to ensure that no vertex appears more than once in the path.
More precisely, we first consider all possibilities for candidates~$c_j$ where~$c_j$ is any sequence of vertices from~$G[I]$ which can be inserted into~$g_j$ (except the already used vertices) and create a new $k$-path~$P_j'$ for each such possibility. 
Here, we assume that all candidates have some arbitrary but fixed ordering in which they are processed; this can be achieved by the labels of the vertices.
Again, after we created path~$P_j'$, we increase increase~$j$ by one and do the same, until~$j=s$.
Note that since we replace~$g_j$ by all possible candidates~$c_j$ (except the already used vertices), the resulting vertex sequence might contain some vertices twice and it might thus be not a valid $k$-path.
This is not an issue, since Phase~2 ensures that the vertices appearing twice are replaced by different vertices which are not used in the sequence.
More precisely, our algorithm ensures that each replacement~$c_{j'}$ for some sequence~$g_{j'}$ (for any~$j'\in[s]$) uses different vertices than all previous replacement.
Consequently, if all $s$~positions are considered it is guaranteed that the sequence is a $k$-path.

Phase~2.2 is almost identical:
The only difference is that all candidate sequences~$c_j$ are only allowed to consist of kernel vertices which are not already used to replace some other sequences~$g_{j'}$ to not violate our guess of Phase~1 and to ensure that each vertex appears at most once in the path.
After all $s$~sequences have been replaced by suitable candidates, we output the resulting $k$-path.

Next, we explain the backtracking once a $k$-path~$P$ is enumerated.
Consider the sequence~$\pi:=\{1,2,\ldots, i-2,i-1,s,s-1,\ldots,i+1,i\}$, that is, the reverse ordering of the indices considered in both phases to replace the sequences.
For any~$i\in[s]$, let~$c_i$ be the $i$th sequence of~$P$ which was used to replace sequence~$g_i$ of~$P'$.
Furthermore, let~$C_i$ be the set of all candidate sequences which can be used for~$g_i$ and let~$\phi_i$ be the ordering of~$C_i$ which is used in the corresponding phase.
Now, let~$j$ be the smallest index according to~$\pi$ such that~$c_j$ is not the last sequence of~$C_j$ which does not contain any vertices which is used in any~$c_{j'}$ where~$j'>j$.
Note that~$j$ can be determined in $\OO(n\cdot |X|)$~time.
If~$j$ does not exist, then all $k$-paths of~$G$ which are equivalent to~$P'$ which contain at least one non-kernel vertex are enumerated.
Otherwise, we replace~$c_j$ by the next such candidate~$c'_j$ of~$C_j$ according to~$\phi_j$ which does not use any vertices of any~$c_{j'}$ for some~$j'>j$.
Now, we restart Phase~$2$ with the $j-1$th entry of~$\pi$.
In other words, we reset~$\phi_{j'}$ to the first sequence in~$C_{j'}$ for any entry~$j'$ of~$\phi$ which corresponds to a smaller index than~$j$ in~$\pi$.

We next show that for any guess we make in any phase, there is always a $k$-path in~$G$ which respects all these guesses, that is, there is a $k$-path for any possibility of replacing the current candidate sequence~$g_j$ of~$P'$ by any candidate sequence~$c_j$ which is not already used.
More precisely, we consider index~$j\in[s]$ and show that there exists at least one candidate~$c_j$.
Recall that the corresponding sequence~$g_j$ of~$P'$ is frequent.
The prolongation property ensures that independently of which vertices are used to replace the sequences~$g_z$ with~$z\ne j$, there exists at least one candidate~$c_j$ which is not hit by the other replacements.
Thus, there exists always at least one candidate.
Hence, each branch we create leads to at least one $k$-path.
Note that each of these $k$-paths contains at least one vertex from~$V(I\setminus I')$ because of our replacement in Phase~1.

It remains to analyze the delay of this algorithm.
In Phase~$1$, for some given index~$i$ we can check in $\OO(n)$~time whether there exists an appropriate non-kernel vertex. 
If yes, we do this replacement and continue with Phase~2, and if not we consider the next index.
Since there are $\OO(|X|)$~choices for~$i$, Phase~1 requires at most $\OO(n\cdot |X|)$~time to terminate or to provide a valid candidate for Phase~2.
Let~$P$ be a $k$-path which is an input for Phase~2.
For each~$g_j$ there exist there are at most $\OO(n)$~candidates, that is, all vertices from~$I$ (or~$I'$ for the indices before the chosen index of Phase~1).
Thus, in $\OO(n)$~time the next candidate for~$c_j$ can be determined or we can end this branch of the enumeration.
Since the number of indices is bounded by~$|X|+1$, we obtain an overall delay of~$\OO(n\cdot |X|)$.
\end{proof}

%\paragraph{Main Theorem.}
After each of the above mentioned four challenges have been resolved, we are ready to finally prove the correctness of the solution-lifting algorithm and henceforth show the main result for the parameter dissociation number~${\diss}$.

{\thmTwo*}

%\begin{theorem}
%\label{theorem:k-path-diss-enum-kernel}
%{\enumkpath} parameterized by ${\diss}$ admits a $\OO({\diss}\cdot n)$-delay enumeration kernelization with $\OO({\diss}^3)$ vertices.
%\end{theorem}

\begin{proof}
Our enumeration kernelization has two parts: the first part is the kernelization algorithm exploiting the marking scheme and the second part is the solution-lifting algorithm.
The fact that the kernel has $\OO({\diss}^3)$~vertices follows directly from our marking scheme (\Cref{obs-kernel-diss}).
Using \Cref{lemma:algo-suitable-check}, given a $k$-path $P'$ of $G'$, we can check whether it is a suitable $k$-path or not in $\OO(n\cdot |X|)$-time.
If $P'$ is a suitable $k$-path, then due to \Cref{lemma:solution-lifting-suitable-path}, we can justify the delay bound of~$\OO({\diss}\cdot n)$.

It remains to show that each $k$-path~$P$ of~$G$ is enumerated exactly once.
We distinguish the cases whether~$P$ is also a $k$-path in~$G'$ or not.
If~$P$ is also a $k$-path in~$G'$, then our algorithm enumerates~$P$ exactly once:
each $k$-path of~$G'$ is enumerated exactly once when~$P$ is the input of the corresponding call of the solution-lifting algorithm and all other $k$-paths which are enumerated contain at least one vertex of~$V(G)\setminus V(G')$ (this is ensured by the fourth challenge).

Hence, it remains to consider the case that~$P$ is \emph{not} a $k$-path in~$G'$.
We begin by showing that~$P$ is enumerated at least once.
Let~$\rmf_{P}$ be the signature of~$P$ and let~$\rmg(P)$ be the corresponding extension using the frequent vertices~$\FF_P$.
According to the definition of~$P$ we have~$\FF_P\setminus \FF'\ne\emptyset$ where~$\FF'$ are the frequent vertices in~$G'$.
Let~$f\in \FF_P\setminus\FF'$ be the vertex with lowest index according to~$\rmg(P)$.
According to challenge~$1$, there exists a $k$-path~$\widehat{P}$ in~$G'$ which is suitable and equivalent to~$P$, that is,~$P$ and~$\widehat{P}$ only differ in some frequent vertices.
Consider the solution-lifting algorithm with the choice of~$\widehat{P}$ and the choice that~$\widehat{f}$, the frequent vertex at the same position as~$f$ in~$\widehat{P}$, is the (frequent) vertex with lowest index in~$\widehat{P}$ which is replaced by the (frequent) non-kernel vertex~$f$.
From now on, the solution-lifting algorithm chooses the frequent vertices in~$P$ to replace the corresponding frequent vertices in~$\widehat{P}$.
Thus,~$P$ is enumerated at least once.

Second, we show that~$P$ is enumerated at most once.
Assume towards a contradiction that~$P$ is enumerated at least twice.
Since the solution-lifting algorithm outputs only $k$-paths having non-kernel vertices for suitable $k$-paths of~$G'$ and since for each such path only some frequent vertices are exchanged, we conclude that both enumerations of~$P$ result from the same $k$-path~$\widehat{P}$ of~$G'$ which is suitable and equivalent to~$P$.
Furthermore, both times, the same first frequent vertex of~$\widehat{P}$ needs to be replaced by the corresponding vertex of~$P$ in Phase~1.
For Phase~2, consider the step of the algorithm where~$P$ is enumerated the first time and let~$f$ be the last vertex which was replaced to obtain~$P$ and assume~$f$ is at position~$j$.
Since~$f$ is only plugged once into~$j$ as ensured by our back-tracking procedure and in each further iteration of the algorithm in each position the next candidate is used (according to the given fixed order), we conclude that~$P$ was enumerated only once.
\end{proof}

Note that the same approach also works for the  larger parameter vertex cover~${\vc}$ which  we  studied in  \Cref{sec:k-path-vc}.
This  would then result in a  kernel having $\OO({\vc}^3)$ vertices which is  larger than the $\OO({\vc}^2)$~vertex  kernel which we presented in \Cref{thm:k-path-vc-result}.
Note, however,  that  the delay of \Cref{thm:k-path-vc-result} is larger than the delay  of  $\OO({\vc}\cdot n)$ of  the adaption of the  above  approach for dissociation number.

\subsection{Extension to Larger Component Size}

Our technique to obtain the enumeration kernel for the dissociation number is not limited to this specific parameter.
In fact, this technique also works works for the smaller parameter $r$-component order connectivity~$r$-${\coc}$.
This parameter is the size of a minimal modulator such that the number of vertices in each remaining component of is at most~$r$, where~$r$ is some fixed constant.
Clearly, the dissociation number corresponds to the special case~$r=2$.
Next, we argue which adaptions of the technique for the dissociation number are necessary to obtain an analog result for the $r$-${\coc}$~number.

Initially, we use an $(r+1)$-approximation (instead of a 3-approximation) to obtain a minimal $r$-${\coc}$~set.
In the marking scheme we again partition the vertices in~$I=V(G)\setminus X$ into \emph{rare} and \emph{frequent} vertices.
Now, we set~$p=r\cdot |X|+1$ instead of~$p=2\cdot |X|+1$.
Also, for each two vertices in the modulator~$X$ we perform the marking process for each~$i\in[r]$ (instead of~$i\in[2]$) to mark these vertices are $i$-rare or $i$-frequent, respectively.
Note that verifying whether there is a path of length~$i\in[r]$ in a  connected component having $r$~vertices  requires $\OO(r^i)$~time.
Thus, this check requires $\OO(2^r)$~time.
Then, the resulting kernel~$G'$ consists of~$X$, all rare vertices, and up to $p$~connected components containing the internal vertices of the paths which make each two vertices~$x_1,x_2\in X$ $i$-frequent for~$i\in[r]$.
The single neighbors are treated similar.
Analog to \Cref{obs-kernel-diss}, this marking scheme has the result that~$G'$ has size~$\OO(r^2\cdot |X|^3)=\OO(|X|^3)$ since~$r$ is a constant.
For  the  running time, observe  that for each  of  the  $|X|^2$~pairs  of  vertices  from  the  modulator and each  choice  of~$i\in[r]$ we  need $\OO(2^r\cdot n)$~time to check whether this  pair is  $i$-frequent.
Furthermore,~$G'$ is an induced subgraph of~$G$ and thus each $k$-path of~$G'$ is also a $k$-path of~$G$.
Again, the choice of~$p$ guarantees us the \emph{prolongation} property, with the only distinction that now~$t\in[r]$ instead of~$t\in[2]$.

For the solution-lifting algorithm, we exploit the same definitions of the mappings~${\rmf}$ and~${\rmg}$.
Also, the definition of equivalent paths is analog.
Furthermore, the solutions and proofs to the four challenges is similar, the only difference is that we need to consider the case~$t\in[r]$ instead of~$t\in[2]$.
Now the analysis of the running time of challenge~3 and~4 is slightly more comprehensive: whenever we checked all possibilities for the first vertex in some sequence~$g_i$; for example for the vertex with minimal index in~$P$ which is a non-kernel vertex, we now need to find a path of length up to~$r$ in a connected component of size at most~$r$.
Via brute-force such a path can be found in $\OO(2^r)$~time.
Thus, the running time for challenges~3 and the delay of challenge~4 increases to~$\OO(|X|\cdot |X|^3\cdot 2^r)=\OO(n\cdot |X|\cdot 2^r)$.
Note that both times are still polynomial since~$r$ is a constant.

Hence, we obtain the following.

{\corTwo*}
%\begin{corollary}
%\label{cor:k-path-r-coc-enum-kernel}
%{\enumkpath} parameterized by $r$-${\coc}$ for constant~$r$ admits a $\OO(r$-${\coc}\cdot n\cdot 2^r)$-delay enumeration kernelization with $\OO((r$-${\coc})^3)$ vertices.
%Furthermore, the kernel  can  be computed  in  $\OO(2^r\cdot  r\cdot (r$-${\coc})^2\cdot n)$~time.
%\end{corollary}

Observe that the factor of~$2^r$ in  the running time  of  the  kernelization is unavoidable since we need to  solve the  problem  of finding a path of length~$i\in[r]$ in  components  of size~$r$.
Recall that for  \textsc{Long-Path}, according  to \Cref{prop:lower-bound-max-cp}, there is  no  polynomial  kernel  for  $\mcp$, the size of a largest connected component, unless  {\nka}.
Note that~$\mcp$ is  a  larger parameter than the  vertex integrity  and  thus  \textsc{Long-Path} has no  polynomial kernel  for the  vertex integrity unless {\nka}.
Thus, it  is  unlikely that  \Cref{cor:k-path-r-coc-enum-kernel}  can be lifted  to  the  smaller  parameter vertex  integrity.

\subsection{Extension to Other Path and Cycle Variants}

We now extend our positive result of {\enumkpath} to {\enumkcycle} and the variants where all paths/cycles of length \emph{at least}~$k$ need to be outputted.
In all adaptions we exploit that each path of length at least~$k$, has length at most~$(r+1)\cdot (|X|+1)$: the path can contain at most all vertices from the modulator and the vertices of at most $(|X|+1)$~connected components of the remaining graph.
Thus, in the kernelization we can preserve the length of a longest path.
We start with {\enumkpathAll}.

\begin{corollary}
\label{cor-diss-k-path-at-least}
{\enumkpathAll} parameterized by $r$-${\coc}$ for constant~$r$ admits a $\OO(r$-${\coc}\cdot n\cdot 2^r)$-delay enumeration kernelization with $\OO((r$-${\coc})^3)$ vertices.
Furthermore, the kernel  can  be computed  in  $\OO(2^r\cdot  r\cdot (r$-${\coc})^2\cdot n)$~time.
\end{corollary}
\begin{proof}
Because of the above argumentation about the maximal path length, the marking scheme and the subsequent kernelization works analog.
Now, our definitions of signature (\Cref{def-diss-sig}) and extension (\Cref{def-diss-ext}) need to be extended:
We not only define them for $k$-paths, instead we define both of them for $k^*$-paths where~$k\le k^*\le (r+1)\cdot (|X|+1)$.
Afterwards, the proof is analog.
\end{proof}

Next, we adapt our result for enumerating cycles of length~$k$.

\begin{proposition}
\label{prop-diss-k-cycle}
{\enumkcycle} parameterized by $r$-${\coc}$ for constant~$r$ admits a $\OO(r$-${\coc}\cdot n\cdot 2^r)$-delay enumeration kernelization with $\OO((r$-${\coc})^3)$ vertices.
Furthermore, the kernel  can  be computed  in  $\OO(2^r\cdot  r\cdot (r$-${\coc})^2\cdot n)$~time.
\end{proposition}
\begin{proof}
The marking scheme and the kernelization work analog.
The definition of the signature is slightly adjusted since a cycle has no designated starting vertex:
Let~$C$ be a $k$-cycle and let~$v^*\in C\cap(X\cup\RR)$ be the vertex of the modulator or rare with minimal index.
Clearly, $v^*$ has 2~neighbors~$v_1$ and~$v_2$ in~$C$.
Let~$v'\in\{v_1,v_2\}$ be the vertex with minimal index.
Now,~$C$ is transformed into a $k$-path~$P_C$ which starts with~$v^*$, has~$v'$ as its second vertex and uses the same vertices of~$C$ in the same order.
Afterwards, the definition of suitable and of the equivalence classes is analog.

In the solution-lifting algorithm, we do not have the case of single neighbors anymore; they occurred if the path started/ended with vertices from~$I$.
Afterwards, the proof is analog.
\end{proof}

Now, the result for {\enumkcycleAll} follows by combining the adaptations for paths of length at least~$k$ (\Cref{cor-diss-k-path-at-least}) and $k$-cycles (\Cref{prop-diss-k-cycle}).

\begin{corollary}
\label{cor-diss-k-cycle-at-least}
{\enumkcycleAll} parameterized by $r$-${\coc}$ for constant~$r$ admits a $\OO(r$-${\coc}\cdot n\cdot 2^r)$-delay enumeration kernelization with $\OO((r$-${\coc})^3)$ vertices.
Furthermore, the kernel  can  be computed  in  $\OO(2^r\cdot  r\cdot (r$-${\coc})^2\cdot n)$~time.
\end{corollary}

\section{Parameterization by Distance to Clique}
\label{sec:k-path-clique}

%\paragraph%%Clique Deletion Set size

\iflong
We present a polynomial-delay enumeration kernel for {\enumkpath} parameterized by the vertex deletion distance~$\cvd(G)$ of the graph~$G$ to a clique.
Again, for an input instance~$(G,k)$ we initially compute a $2$-approximation~$X$ of~$\cvd(G)$ by greedily adding~$2K_1$s to~$X$.
Let~$\ell=|X|$ and $C=V(G)\setminus X$ be the remaining clique.
Again, we assume an arbitrary but fixed ordering~$\prec$ on the vertices of~$C$.
We exploit the ordering~$\prec$ to define a minimal element in each equivalence class of $k'$-paths in~$G'$ which is then used in the solution-lifting algorithm.

In sharp contrast to the vertex cover number~$\vc(G)$ (see \Cref{sec:k-path-vc}), and more general the dissociation number~$\diss(G)$ (see \Cref{sec:k-path-disoc}), the vertex deletion distance~$\cvd(G)$ to a clique does \emph{not} provide any bound for the length~$k$ of a path.
Hence, in the kernelization we modify both the graph~$G$ \emph{and} the parameter~$k$.

\paragraph{Strategy.}
Initially, we employ a marking scheme similar to the one we used for the dissociation number~$\diss(G)$ (see \Cref{sec:k-path-disoc}).
More precisely, we mark sufficiently many vertices in~$C$ as \emph{rare} and \emph{frequent} such that all possible structures of a $k$-path can be preserved in the kernel.
Our kernelization comprises two steps:
First, we remove unmarked vertices of the clique and simultaneously we decrease~$k$ by the number of removed vertices.
This ensures that the new parameter~$k'$ is bounded in~$\ell\in\OO(\cvd(G))$.
Now, the remaining clique might still be unbounded in~$\cvd(G)$.
Thus, in a second step, we remove unmarked vertices of the clique until also the size of the remaining clique is bounded in some function depending only on~$\cvd(G)$.
It is necessary to do this in two steps since we have to ensure that all signatures are preserved ($k'$ does not get too small) and that we do not obtain an instance with no $k'$-path ($k'$ is larger than the number of vertices in the kernel). 
%In contrast to~$\vc(G)$ (see \Cref{sec:k-path-vc}) and~$\diss(G)$ (see \Cref{sec:k-path-disoc}) the decrease of~$k$ is now necessary to preserve all structures of $k$-paths.
%The decrease of~$k$ is also the reason for the above case distinction whether~$k$ is much smaller than~$|C|$: if yes,~$k'$ would be smaller than zero.

Our solution lifting algorithm uses similar ideas as the solution lifting algorithm for the dissociation number~$\diss(G)$ (see \Cref{sec:k-path-disoc}):
Initially, we define the \emph{signature}~$\rmf_{P'}$ of each $k'$-path~$P'$ in~$G'$.
Afterwards, based on these signatures we define equivalence classes and define a suitable $k'$-path of each equivalence class.
Then, we present a base version of the solution lifting algorithm \texttt{SolLift}.
This base version is adapted to obtain an algorithm \texttt{SolLiftTest} to check whether an $k'$-path is suitable or not.
Afterwards, \texttt{SolLift} and \texttt{SolLiftTest} are used to verify that the equivalence classes of the reduced graph~$G'$ are identical to those of the input graph~$G$.
In contrast to~$\diss(G)$, we cannot pick the suitable path~$P'$ in~$G'$ of an equivalence class to enumerate all $k$-paths in~$G$ having the same signature and output the path itself if~$P'$ is not suitable:
$P'$ has only length~$k'$ and thus we have to add $k-k'$~vertices to~$P'$.
We need to do this in a clever way to ensure that each $k$-path in~$G$ is only enumerated once.
For this we use two further variants \texttt{SolLiftS} for suitable paths and \texttt{SolLiftNS} for non-suitable paths.
With each call of \texttt{SolLiftNS} we output \emph{exactly one} $k$-path having signature~$\rmf_{P'}$, and with \texttt{SolLiftS} we output all $k$-paths having signature~$\rmf_{P'}$ not outputted by any call of \texttt{SolLiftNS}.
%%%%%%%%%%%%%%%%%%%%%%%%%%%%%%%%%%%%%%%%%%%%%%%%%%%%%%%%%%%%%

\else

\fi

\iflong
\subsection{Marking Scheme and Kernelization}

\paragraph{Marking Scheme.}
Similar to~$\diss(G)$ we mark sufficiently many common neighbors in~$C$ of two vertices in the modulator~$X$, and we also mark sufficiently many neighbors in~$C$ of each vertex in~$X$.
Let~$x,y\in X$ and~$p=2(\ell+1)$.
If~$|N(x)\cap N(y)\cap C|\le p$, then mark all these vertices as \emph{rare}, and otherwise, mark~$p$ of these vertices (arbitrarily) as \emph{frequent}.
We treat the case~$N(x)\cap C$ similar: if~$|N(x)\cap C|\le p$, then mark all these vertices as \emph{rare}, and otherwise, mark~$p$ of these vertices (arbitrarily) as \emph{frequent}.
Now, by~$\mathcal{R}$ we denote all vertices marked as \emph{rare} for at least one vertex~$x\in X$ or for at least one pair of vertices~$x,y\in X$ and by abusing notation we denote all vertices in~$\mathcal{R}$ as \emph{rare}.
Again, by abusing notation, all remaining marked vertices are denoted by~$\mathcal{F}$ and called \emph{frequent}.
The number of vertices in~$\mathcal{R}\cup\mathcal{F}$ can be bounded as for~$\diss(G)$ (see \Cref{obs-kernel-diss}).

\begin{observation}
\label{obs-mark-cvd}
The above-described marking scheme marks at most $2p\cdot\ell^2\le 4\cdot (\ell+1)^3$~vertices as rare or frequent.

\end{observation}

\begin{proof}
For each pair of vertices in the modulator~$X$ at most $p$~common neighbors in the clique are marked as rare or frequent.
This results in at most $p\cdot \ell^2$~marked vertices.
Furthermore, for each vertex in~$X$ at most $p$~neighbors in the clique are marked.
Thus, the claimed bound follows.
\end{proof}

\paragraph{Kernelization.}
We may safely assume that~$k\le |V(G)|=|C|+|X|$ since otherwise there is no $k$-path in~$G$.
Let~$q=\min(|C|-4\cdot (\ell+1)^3,k-4\ell)$. 
We remove $q$~unmarked vertices of~$C$ arbitrarily, resulting in graph~$G''$ with partition~$V(G'')=X\cup C''$.
Simultaneously, we set~$k'=k-q$.
Afterwards, if~$|C''|> 4\cdot(\ell+1)^3$, we remove $|C''|-4\cdot(\ell+1)^3$~unmarked vertices of~$C''$ arbitrarily, resulting in graph~$G'$ with partition~$V(G')=X\cup C'$.
Note that for the vertices in~$C'$, we maintain their relative ordering implied by the labels of all vertices in~$C$.
Observe that for each removal there are sufficiently many unmarked vertices in~$C$ or~$C''$ since by \Cref{obs-mark-cvd} there are at most $4\cdot(\ell+1)^3$~vertices marked in~$C$.

\begin{lemma}
\label{lem-kernel-dist-clique}
The graph~$G'$ has $\OO(\cvd(G)^3)$~vertices and~$k'\in\OO(\cvd(G))$.
\end{lemma}

\begin{proof}
Our kernelization algorithm ensures that~$|C'|\le 4\cdot(\ell+1)^3$.
Since~$X$ is a 2-approximation of~$\cvd(G)$, we have~$|X|\le 2\cdot\ell$.
Thus,~$G'$ has at most $2\cdot\ell+4\cdot(\ell+1)^3\in\OO(\ell^3)$~vertices.
Furthermore, note that~$k'=\min(k,4\cdot\ell)$.
Hence,~$k'\le 4\cdot\ell$.
The statement follows since~$\ell\le 2\cdot\cvd(G)$.
\end{proof}

\subsection{Signature and Equivalence Classes}

Similar to \Cref{sec:k-path-vc,sec:k-path-disoc} we define the mappings~${\rmf}$ and~${\rmg}$, and the notation of equivalent $k$-paths to establish relations between paths in~$G$ and $k'$-paths in~$G'$.

\paragraph{Definition of Signature.}

Here, the signature~$\rmf_{P}$ is not simply the mapping of vertices in the modulator and of the rare vertices to their index in~$P$ (as in \Cref{sec:k-path-disoc}).
We need a more refined definition since now~$k$ might be much bigger than~$k'$.
For this, we replace consecutive subsequences of vertices in~$P$ where no vertex is in the modulator~$X$ or rare by new symbols or remove it. 
Next, we provide a formal definition.
For this, let~$g_C$ and~$g_N$ be two new \emph{gap-symbols}.
These two symbols are used to encode different connection types between two vertices of the modulator to the clique.
More precisely, symbol~$g_C$ is used if between two vertices of the modulator in~$P$ there is \emph{exactly one} common neighbor of these two vertices.
Symbol~$g_N$ is used to encode the other case, that is, there are many vertices between two vertices from the modulator in~$P$ (including a rare vertex).
An example is shown in \Cref{fig-clique-example-signature}.

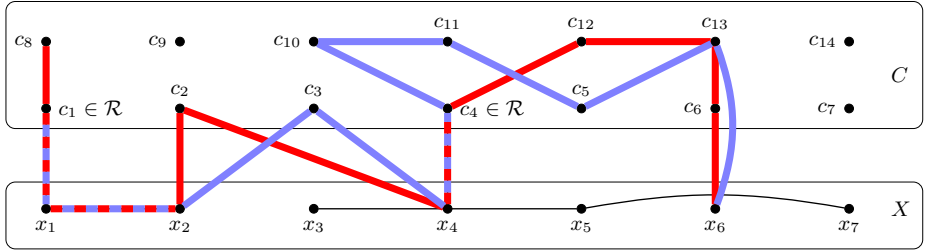
\begin{figure}[t]
\centering
\begin{adjustbox}{width=\textwidth}
\begin{tikzpicture}
\draw[rounded corners] (-0.6, -0.6) rectangle (13.1, 0.4) {};
\node[label=below:{$x_1$}](X1) at (0, 0) [shape = circle, draw, fill=black, scale=0.11ex]{};
\node[label=below:{$x_2$}](X2) at (2, 0) [shape = circle, draw, fill=black, scale=0.11ex]{};
\node[label=below:{$x_3$}](X3) at (4, 0) [shape = circle, draw, fill=black, scale=0.11ex]{};
\node[label=below:{$x_4$}](X4) at (6, 0) [shape = circle, draw, fill=black, scale=0.11ex]{};
\node[label=below:{$x_5$}](X5) at (8, 0) [shape = circle, draw, fill=black, scale=0.11ex]{};
\node[label=below:{$x_6$}](X6) at (10, 0) [shape = circle, draw, fill=black, scale=0.11ex]{};
\node[label=below:{$x_7$}](X7) at (12, 0) [shape = circle, draw, fill=black, scale=0.11ex]{};
\node[label=right:{$X$}](X) at (12.4, 0) []{};

\draw[rounded corners] (-0.6, 1.2) rectangle (13.1, 3.1) {};
\node[label=right:{$c_1\in\RR$}](C1) at (0, 1.5) [shape = circle, draw, fill=black, scale=0.11ex]{};
\node[label=above:{$c_2$}](C2) at (2, 1.5) [shape = circle, draw, fill=black, scale=0.11ex]{};
\node[label=above:{$c_3$}](C3) at (4, 1.5) [shape = circle, draw, fill=black, scale=0.11ex]{};
%\node[label=above:{$c_3\in\FF$}](C31) at (4, 1.6) {};
\node[label=right:{$c_4\in\RR$}](C4) at (6, 1.5) [shape = circle, draw, fill=black, scale=0.11ex]{};
\node[label=above:{$c_5$}](C5) at (8, 1.5) [shape = circle, draw, fill=black, scale=0.11ex]{};
\node[label=left:{$c_6$}](C6) at (10, 1.5) [shape = circle, draw, fill=black, scale=0.11ex]{};
\node[label=left:{$c_7$}](C7) at (12, 1.5) [shape = circle, draw, fill=black, scale=0.11ex]{};

\node[label=left:{$c_8$}](C8) at (0, 2.5) [shape = circle, draw, fill=black, scale=0.11ex]{};
\node[label=left:{$c_9$}](C9) at (2, 2.5) [shape = circle, draw, fill=black, scale=0.11ex]{};
\node[label=left:{$c_{10}$}](C10) at (4, 2.5) [shape = circle, draw, fill=black, scale=0.11ex]{};
\node[label=above:{$c_{11}$}](C11) at (6, 2.5) [shape = circle, draw, fill=black, scale=0.11ex]{};
\node[label=above:{$c_{12}$}](C12) at (8, 2.5) [shape = circle, draw, fill=black, scale=0.11ex]{};
\node[label=above:{$c_{13}$}](C13) at (10, 2.5) [shape = circle, draw, fill=black, scale=0.11ex]{};
%\node[label=above:{$c_{13}\in\FF$}](C131) at (10, 2.6) {};
\node[label=left:{$c_{14}$}](C14) at (12, 2.5) [shape = circle, draw, fill=black, scale=0.11ex]{};
\node[label=right:{$C$}](C) at (12.4, 2) []{};

\path [-,line width=0.2mm](X1) edge (X2);
\path [-,line width=0.2mm](X3) edge (X5);
\path [-,line width=0.2mm, bend right=10](X7) edge (X5);

\path [-,line width=0.2mm](X1) edge (C8);
\path [-,line width=0.2mm](X2) edge (C2);
\path [-,line width=0.2mm](X4) edge (C2);
\path [-,line width=0.2mm](X2) edge (C3);
\path [-,line width=0.2mm](X4) edge (C3);
\path [-,line width=0.2mm](X4) edge (C4);
\path [-,line width=0.2mm](X6) edge (C6);
\path [-,line width=0.2mm, bend right=20](X6) edge (C13);

\draw[dashed,dash pattern=on 5pt off 5pt,red,line width=3pt]  (C1) to (X1);
\draw[dashed,dash pattern=on 5pt off 5pt,blue!50,dash phase=5pt,line width=3pt]  (C1) to (X1);

\draw[dashed,dash pattern=on 5pt off 5pt,red,line width=3pt]  (X1) to (X2);
\draw[dashed,dash pattern=on 5pt off 5pt,blue!50,dash phase=5pt,line width=3pt]  (X1) to (X2);

\draw[dashed,dash pattern=on 5pt off 5pt,red,line width=3pt]  (X4) to (C4);
\draw[dashed,dash pattern=on 5pt off 5pt,blue!50,dash phase=5pt,line width=3pt]  (X4) to (C4);

\draw[red,line width=3pt]  (C1) to (C8);
\draw[red,line width=3pt]  (X2) to (C2);
\draw[red,line width=3pt]  (C2) to (X4);
\draw[red,line width=3pt]  (C4) to (C12);
\draw[red,line width=3pt]  (C12) to (C13);
\draw[red,line width=3pt]  (C13) to (C6);
\draw[red,line width=3pt]  (C6) to (X6);

\draw[blue!50,line width=3pt]  (X2) to (C3);
\draw[blue!50,line width=3pt]  (X4) to (C3);
\draw[blue!50,line width=3pt]  (C4) to (C10);
\draw[blue!50,line width=3pt]  (C11) to (C10);
\draw[blue!50,line width=3pt]  (C5) to (C11);
\draw[blue!50,line width=3pt]  (C5) to (C13);
\draw[blue!50,line width=3pt, bend right=20]  (X6) to (C13);

%\draw[black,dotted,line width=1.2] (C1) circle (5.5pt);
%\draw[black,dotted,line width=1.2] (C2) circle (5.5pt);
%\draw[black,dotted,line width=1.2] (C3) circle (5.5pt);
%\draw[black,dotted,line width=1.2] (C4) circle (5.5pt);
%\draw[black,dotted,line width=1.2] (C6) circle (5.5pt);
%\draw[black,dotted,line width=1.2] (C13) circle (5.5pt);

%\node[label=right:{$\in\RR$}](T) at (0, 1.5) {};
%\node[label=right:{$\in\RR$}](T) at (6, 1.5) {};
%\node[label=right:{$\in\FF$}](T) at (2, 1.6) {};
%\node[label=right:{$\in\FF$}](T) at (4, 1.6) {};
%\node[label=right:{$\in\FF$}](T) at (10.1, 1.5) {};
%\node[label=right:{$\in\FF$}](T) at (10.0, 2.5) {};
\end{tikzpicture}
\end{adjustbox}

\caption{A graph~$G$ with clique~$C$ and modulator~$X$ where~$c_1$ and~$c_4$ are rare, and all other clique vertices are frequent. 
The two paths~$P_1=(c_8,c_1,x_1,x_2,c_2,x_4,c_4,c_{12},c_{13},c_6,x_6)$ (in \textcolor{red}{red}) and~$P_2=(c_1,x_1,x_2,c_3,x_4,c_4,c_{10},c_{11},c_5,c_{13},x_6)$ (in \textcolor{blue!50}{blue}) have the same signature~$(c_1,x_1,x_2,g_C,x_4,c_4,g_N,x_6)$.}
\label{fig-clique-example-signature}
\end{figure}

\begin{definition}
\label{def-sig-dist-clique}
For each consecutive subsequence~$(v_i, s,v_{j})$ of~$P$ where~$s=(v_{i+1},\ldots, v_{j-1})$ such that~$v_i,v_{j}\in X\cup\mathcal{R}$ and~$v_q\notin X\cup\mathcal{R}$ for~$q\in[i+1,j-1]$, we do the following:

%for this, we distinguish whether~$v_i$ and~$v_{j}$ are contained in~$X$ or~$\mathcal{R}$.

\begin{enumerate}
\item If~$v_i,v_{j}\in\mathcal{R}$: We remove~$s$ from~$P$, and we do not add any gap-symbols.

\item If~$v_i\in X$ and~$v_{j}\in\mathcal{R}$, or~$v_i\in \mathcal{R}$ and~$v_{j}\in X$: We replace~$s$ by~$(g_N)$.
We treat the cases that~$P_i$ starts with a consecutive sequence~$s$ of non-rare clique vertices followed by a vertex from the modulator~$X$ or that~$P_i$ ends with a consecutive sequence~$s$ of non-rare clique vertices after a vertex from the modulator~$X$ analog: we replace~$s$ by~$(g_N)$.

\item If~$v_i,v_{j}\in X$: If~$s$ has length~$1$, we replace~$s$ by~$(g_C)$.
Otherwise, if~$s$ has length at least~$2$, we replace~$s$ by~$(g_Ng_N)$.
Here, it is important to add two symbols~$g_N$ and not only one.
\end{enumerate}

The resulting sequence is denoted as \emph{signature}~$\rmf_{P}$.
\end{definition}

Note that~$\rmf_{P}$ only consists of gap-symbols~$g_N$ and~$g_C$, rare vertices, and vertices from the modulator~$X$.
Also note that~$\rmf_{P}$ can be an empty sequence; this is the case if~$P$ only consists of non-rare clique vertices. 
Observe that each consecutive subsequence of vertices from~$P$ not containing any vertex from~$X\cup\mathcal{R}$ is replaced by a (possible empty) sequence of gap-symbols having the same or smaller length.
Thus,~$\rmf_{P}$ contains at most $|P|$~symbols/vertices.
Furthermore, for each vertex in~$X\cap P$, we add at most $2$~gap-symbols.
Thus,~$\rmf_{P}$ contains at most $2\cdot |X|$~gap-symbols.
Also, note that~$\rmf_{P}$ can be computed in linear time.
The signature of a $k'$-path~$P'$ in~$G'$ is defined analogously.

\paragraph{Equivalent paths and Equivalence classes.}
Again, we exploit~$\rmf_{P}$ and~$\rmg_{P}$ to define equivalence classes of $k$-paths in~$G$ and $k'$-paths in~$G'$.

\begin{definition}
\label{def-equivalence-dist-clique}
Let~$P_1$ and~$P_2$ be two $k$-paths.
$P_1$ and~$P_2$ are \emph{equivalent} if and only if~$\rmf_{P_1}=\rmf_{P_2}$.
\end{definition}

This definition also applies for two $k'$-paths and a combination of both.
Analog to \Cref{sec:k-path-disoc}, based on the notion of equivalent paths we now define equivalence classes:
Let~$\mathcal{P}$ be the set of all $k$-paths in~$G$.
Furthermore, let~$\mathcal{P}=\mathcal{P}_1\cup\mathcal{P}_2\cup\ldots\cup\mathcal{P}_s$ be the partition into equivalence classes where~$s$ is the number of  equivalence classes.
Analog, let~$\mathcal{P'}=\mathcal{P'}_1\cup\ldots\cup\mathcal{P'}_{s'}$ be the equivalence classes of $k'$-paths in~$G'$.

\paragraph{Definition of Extension.}

Next, we define the mapping \emph{extension}~$\rmg(P)$ of the signature~$\rmf_{P}$ of a $k$-path~$P$ in~$G$ or a $k'$-path~$P$ in~$G'$.

\begin{definition}
\label{def-ext-dist-clique}
Let~$P'$ be a $k'$-path in~$G'$. 
A path~$P=\rmg(P')$ is an \emph{extension} of~$P'$ if~$P$ is a $k$-path in~$G$ with~$\rmf_{P}=\rmf_{P'}$.
\end{definition}

An extension in~$G'$ is defined analogously.
Observe that for constructing an extension we need to remove the gap-symbols and add consecutive sequences of (non-)kernel vertices from the clique without violating the constraint that the resulting sequence is a $k$-path in~$G$ (or in~$G'$, respectively) having the same signature.
These properties are ensured by our solution lifting algorithm which first replaces the gap-symbols by (non-)kernel vertices from the clique such that the resulting sequence is a path of length~$|\rmf_{P}|$.
This is always possible by our definition of frequent vertices, as we show.
Afterwards, we only have to add consecutive sequences of (non-)kernel vertices from the clique to obtain a $k$-path.
Note that in both steps, we cannot add rare vertices since they are part of the signature.

\subsection{Challenges}
The basic idea of the solution lifting algorithm is analog to the one for the dissociation number~$\diss(G)$ (see \Cref{sec:k-path-disoc}): for each equivalence class~$\PP_i$ we choose a \emph{suitable} $k'$-path~$P_i'\in\PP_i'$ which is used to enumerate all $k$-paths in~$\PP_i$ with polynomial-delay.
Now, we face more obstacles.
First, it is not clear that there is a one to one correspondence between the equivalence class $\PP$ and~$\PP'$.
Second, since~$k'$ might be smaller than~$k$, we cannot simply output the $k'$-path~$P_i'$ if~$P_i'$ is not suitable.
Instead, we have to output at least one $k$-path of the equivalence class~$\PP_i$ corresponding to the equivalence class~$\PP_i'$ containing path~$P_i'$.
More precisely, we design our solution algorithm such that the algorithm outputs \emph{exactly one} equivalent $k$-path of~$G$ if~$P_i'$ is not suitable.
%\todo[inline]{F: I would remove the following; I added the sentence above.
%{\color{blue} DM: Looks okay to me.}}
Consequently, we cannot output all $k$-paths in~$\PP_i$ if~$P_i'$ is suitable.
Overall, we need to overcome the following challenges (some of which are similar to \Cref{sec:k-path-disoc}):

\begin{enumerate}
\item Verify that there is a one-to-one correspondence between equivalence classes~$\PP_i'$ and~$\PP_i$, that is,
for each $k$-path~$P$ in~$G$ there exists at least one $k'$-path~$P'$ is~$G'$ such that~$\rmf_{P}=\rmf_{P'}$.

\item Define the notation of \emph{suitable}.

\item Check in polynomial time for some $k'$-path~$P_i'\in\PP_i'$ whether~$P_i'$ is suitable or not.

\item If~$P_i'\in\PP_i'$ is \emph{not} suitable, output \emph{exactly one} $k$-path~$P\in\PP_i$ in polynomial time, and otherwise, if~$P_i'$ is suitable, then enumerate all $k$-paths in~$\PP_i$ with polynomial delay without outputting the $k$-paths outputted by non-suitable paths in~$\PP_i'$. 
This ensures that each $k$-path of~$G$ is enumerated exactly once.
\end{enumerate}

\subsection{Resolving the Four Challenges}
\label{sec-clique-solve-challenges}

First, we present our definition of a suitable path.
Second, we present a base variant of the solution lifting algorithm \texttt{SolLift}.
This base version is then adapted to test whether a $k'$-path is suitable.
The actual solution lifting algorithm used to enumerate $k$-paths consists of two further variants of \texttt{SolLift}; one for suitable $k'$-paths, and one for non suitable $k'$-paths.

\paragraph{Definition of suitable.}
Recall that in the kernel we maintain the relative ordering of the vertices~$V(G')$ as they were in~$V(G)$.
We can encode the vertices~$V(G')$ using $\OO(\log ({\cvd(G)}))$~bits and by construction,~$G'$ has $\OO(|X|^3)$~vertices, hence $\OO({\cvd(G)}^3)$~vertices.
Therefore,~$G'$ can be encoded using $\OO({\cvd(G)}^3\log ({\cvd(G)}))$~bits and this assumption is justified.

Let~$P_i'\in\PP_i'$.
By~$\orderrr(P_i')$ we denote the sequence of indices of~$P_i'$.
We also use a second sequence~$\orderr(P_i')$ for our definition of suitable:
We remove all non-rare clique vertices from~$P_i'$ whose predecessor and successor in~$P_i'$ are both from~$C'$.
This results in the sequence~$Q_i'$ and the resulting sequence of indices is denoted by~$\orderr(P_i')$.
Observe that~$Q_i'$ is a $|\rmf_{P_i'}|$-path in~$G'$:
First, as $P_i'$ is a path in~$G'$ and we only remove non-rare vertices from $C'$ whose predecessor and successor are in $C'$, after each removal we still have a path.
Thus,~$Q_i'$ is a path in~$G'$.
Second, observe that both~$Q_i'$ and~$\rmf_{P_i'}$ contain all vertices of~$P_i'\cap(X\cup\RR)$.
Furthermore, since each gap-symbol in~$\rmf_{P_i'}$ has at least one neighbor in $X$ and a gap-symbol is only inserted into~$\rmf_{P_i'}$ for a maximal sequence~$s$ of vertices in~$P_i'$ not containing any vertex of~$X\cup\RR$ such that the predecessor or the successor of~$s$ is from~$X$, we conclude that~$|Q_i'|=|\rmf_{P_i'}|$.
For example, in the example of \Cref{fig-clique-example-signature}, we have~$\orderr(P_i')=(c_1,x_1,x_2,c_3,x_4,c_4,c_{13},x_6)$.

Based on~$\orderr$ and~$\orderrr$, we define a ordering~$\triangleleft$ of all $k'$-paths in~$\PP_i'$.

\begin{definition}
\label{def-ordering-dist-clique}
Let~$P_1',P_2'\in\PP_i'$.
We write~$\orderr(P_1')\triangleleft\orderr(P_2')$ (and~$\orderrr(P_1')\triangleleft\orderrr(P_2')$) if and only if~$\orderr(P_1')$ ($\orderrr(P_1')$) is lexicographically smaller than~$\orderr(P_2')$ ($\orderrr(P_2')$).
We say~$P_1'\triangleleft P_2'$ if and only if $a)$~$\orderr(P_1')\triangleleft \orderr(P_i')$, or $b)$~$\orderr(P_i')=\orderr(P_2')$, and~$\orderrr(P_1')\triangleleft \orderrr(P_2')$.
\end{definition}

Note that~$\orderrr(P_1')= \orderrr(P_2')$ implies~$P_1'=P_2'$.
Thus,~$\triangleleft$ is a total ordering of~$\PP_i'$.
This allows as to define the suitable path of each equivalence class.

\begin{definition}
\label{def-suitable-dist-clique}
A path~$P_i'$ is \emph{suitable} if and only if~$P_i'$ is the minimal path of~$\PP_i'$ with respect to~$\triangleleft$.
\end{definition}

Since~$\triangleleft$ is a total ordering, the minimal path is well-defined and unique.
By~$\pTwo(\PP_i')$ we denote the minimal path of~$\PP_i'$ and by~$\pOne(\PP_i')$ we denote the path corresponding to~$\orderr$ of this minimal path.
Both notations are used in \texttt{SolLiftTest} to check whether a path~$P_i'$ is suitable or not.

\paragraph{Solution Lifting Algorithm.}
Let~$P_i'\in\PP_i'$ be a $k'$-path in~$G'$ with signature~$\rmf_{P_i'}$.
Recall that in order to obtain a $k$-path $P_i$ from $\rmf_{P_i'}$, we need to remove the gap-symbols and add consecutive sequences of vertices from~$C$ which are not rare.
As discussed earlier, ideally one wants to insert each non-rare clique vertex in each gap; but this is not possible:
A vertex from the modulator~$X$ in~$P_i'$ cannot have each vertex from~$C$ as predecessor or successor; only some frequent vertices are possible at this position ensuring the path property.
Because of this technical hurdle, our solution lifting algorithm consists of two phases, which we next describe informally.

\begin{enumerate}[{Phase} 1:]
\item By definition, the predecessor and successor of a gap-symbol~$g_C$ in~$\rmf_{P_i'}$ are two vertices from the modulator~$X$.
Thus, we have to replace~$g_C$ by a common neighbor of these two vertices.
Similar, either the predecessor or the successor (not both, see \Cref{def-sig-dist-clique}) of a gap-symbol~$g_N$ in~$\rmf_{P_i'}$ is from the modulator~$X$ (this is well-defined).
Consequently, we have to replace~$g_N$ by a neighbor of this vertex.
Since we marked sufficiently many vertices as frequent, such replacements are always possible.
Let~$Q_i'$ be the resulting $|\rmf_{P_i'}|$-path.

\item 
Now, we add further non-rare clique vertices to~$Q_i'$ to obtain a $k$-path~$P_i$.
More precisely, we add (possible empty) sequences of non-rare clique vertices between any two clique vertices of~$Q_i'$.
\end{enumerate}

Next, we provide formal details for both phases.

\paragraph{Phase~1.} Recall that the aim of Phase~$1$ is to replace each gap-symbol~$g_N$ or~$g_C$ in~$\rmf_{P_i'}$ by \emph{exactly one} non-rare clique vertex.
Also, recall that the total number of gap-symbols in~$\rmf_{P_i'}$ is at most~$2\cdot |X|<p$; the threshold for marking clique vertices as rare or frequent.
Next, we describe how we replace each gap-symbol~$g_C$ and~$g_N$.
We replace them from lowest to highest index in~$\rmf_{P_i'}$.
Let~$L$ be the current set of non-rare clique vertices we used for this replacement.
Initially,~$L=\emptyset$.

Let~$u,w\in X\cap P_i'$ be the predecessor and successor of~$g_C$, respectively, and let~$D_C=(N(u)\cap N(w)\cap C)\setminus (\mathcal{R}\cup L)$ be the common neighbors of~$u$ and~$w$ in the clique which are not rare and not already used.
For each vertex~$v\in D_C$, we branch in the possibility of replacing~$g_C$ with~$v$.
Then, we add~$v$ to~$L$.
We show that each of these branches leads to at least one valid extension which is a $k$-path of $G$.
For gap-symbol~$g_N$ we proceed similarly:
By definition of~$g_N$, either its predecessor or successor (not both, see \Cref{def-sig-dist-clique}), say~$u$, in~$P_i'$ is from the modulator~$X$.
Let~$D_N=(N(u)\cap C)\setminus(\mathcal{R}\cup L)$ be the non-rare neighbors of~$u$ in the clique which are not already used.
Again, for each~$v\in D_N$, we branch in the possibility of replacing~$g_N$ with~$v$ and we show that each of these branches leads to at least one valid extension, that is, a $k$-path in~$G$.
Afterwards, we add~$v$ to~$L$.

It remains to show that for each gap-symbol there is at least one vertex we can replace this gap-symbol with, that is,~$|D_C|\ge 1$ and~$|D_N|\ge 1$ for each step of the replacement.
%By definition of the gap-symbols~$g_C$ and~$g_N$, there was at least one non-rare vertex~$v$ from the clique~$C$ in the position of that symbol.
%Thus,~$v\in D_C\cup L$ or~$v\in D_N\cup L$, respectively.
By our definition of frequent vertices, we have~$|D_C\cup L|\ge p$ and~$|D_N\cup L|\ge p$.
Since~$L$ contains at most one vertex for each gap-symbol in~$\rmf_{P_i'}$, we obtain that~$|L|\le 2\cdot|X|<p$.
Hence,~$|D_C|\ge 1$ and~$|D_N|\ge 1$ for each branching.
Thus, each branching leads to a $|\rmf_{P_i'}|$-path in~$G$ containing no gap-symbols.

%\todo[inline]{Diptapriyo: it seems that the case of replacing two gap symbols $g_N g_N$ case is missing in both the phases.}

\paragraph{Phase~2.} 
Recall that the aim of this phase is to add non-rare clique vertices to~$Q_i'$ to obtain $k$-paths of~$G$ having the same signature as~$P_i'$.
For this, we add (possible empty) sequences of (non-)kernel clique vertices to~$Q_i'$ in the following cases: $a)$ between two consecutive clique vertices from~$Q_i'$, and $b)$ in the front/end of~$Q_i'$ if the first/last vertex of~$Q_i'$ is from the clique.
For this, let~$p_1,p_2,\ldots, p_t$ with~$t\le |Q_i'|+1$ be the positions of~$Q_i'$ such that~$p_i$ is the $i$th position in~$Q_i'$ fulfilling~$a)$ or~$b)$.

Let~$L=C\setminus(\mathcal{R}\cup V(Q_i'))$ be the set of non-rare clique vertices not already contained in~$Q_i'$.
Note that~$|L|\ge k-|Q_i'|$, that is, there are sufficiently many vertices to add:
By definition of our kernel we have~$|C\setminus C'|\ge k-k'$ and by definition of path~$Q_i'$ we know that all vertices in~$V(P_i')\setminus V(Q_i')$ are available in~$L$ and contained in~$C'$.

Intuitively, for each~$Z\subseteq L$ with~$|Z|=k-|Q_i'|$ and each ordering~$\pi$ of~$Z$ we enumerate all $k$-paths which can be obtained by adding the vertices in~$Z$ according to ordering~$\pi$ to~$Q_i'$ in the positions~$p_1,\ldots, p_t$, that is, into each~$p_i$ we insert a (possibly empty) subsequent sequence of~$\pi(Z)$ such that each vertex in~$Z$ is assigned to exactly one position.

Now, we provide a precise description how Phase $2$ works algorithmically.
The algorithm starts with the lexicographically smallest subset~$L'$ of~$L$ having $k'-|Q_i'|$~vertices and in each further iteration the algorithm considers the (unique) successor subset~$L'$ according to this lexicographical order.
Furthermore, the algorithm starts with the lexicographical smallest ordering~$\pi_1$ of~$L'$ and all orderings are sorted according to the lexicographical ordering of the resulting sequence.
In the following, let~$Z$ be such a set.
Consequently, in each further iteration the algorithm chooses the (unique) successor ordering~$\pi_j$.
Note that all possibilities how many vertices of~$Z$ in order~$\pi$ are added into each positions~$p_1,\ldots , p_t$ can be described by a vector~$\vec{v}$ having~$t$ entries with values from~$0$ to~$|Z|$ such that the sum of all entries is~$|Z|$.
Let~$\mathcal{V}$ be the set of all such vectors and let~$\chi$ be an arbitrary but fixed ordering of~$\mathcal{V}$.
Now, for fixed vertices~$Z$ and fixed ordering~$\pi$ the algorithm adds~$Z$ in order~$\pi$ into positions~$p_1,\ldots, p_t$ for all possibilities~$\mathcal{V}$ in order of~$\chi$.

Note that if~$\rmf_{P_i'}$ is empty, then there is exactly one position~$p_1$.
The resulting sequence~$P$ has length~$k$.
In fact,~$P$ is also a $k$-path in~$G$:~$Q_i'$ is a path in~$G$ and between each two clique vertices of~$Q_i'$ (and start/end if first/last vertex is from the clique), we add a consecutive sequence of clique vertices.
Thus, after each addition the resulting sequence is still a path.
We denote this algorithm as \texttt{SolLift}.

\paragraph{Check whether~$P_i'\in\PP_i'$ is suitable.}
Intuitively, the verification whether~$P_i'\in\PP_i'$ is suitable is done by modifying the two phases of \texttt{SolLift} as follows:
First, only kernel vertices from~$C'$ can be added instead of all clique vertices in~$C$.
Second, the output is a $k'$-path of~$\PP_i'$, and not all $k$-paths in~$\PP_i$.
For this, in both phases and each step of them only the vertex with smallest index is chosen.
This algorithm is referred to as \texttt{SolLiftTest} and the unique and well-defined output of \texttt{SolLiftTest} with input~$P_i'\in\PP_i'$ is a $k'$-path~$P_i^*\in\PP_i'$.
Next, we provide a formal description.

\begin{lemma}
\label{lemma:algo-suitable-check-clique}
Given a $k$-path $P_i \in \PP_i'$ of $G'$, there is an $\OO(n\cdot\cvd(G))$-time algorithm \texttt{SolLiftTest} that correctly decides whether $P_i$ is suitable or not.
\end{lemma}

\begin{proof}
Let~$\rmf_{P_i'}$ be the signature of~$P_i'$.
Recall that in Phase~1 of \texttt{SolLift} each gap-symbol is replaced by each possibility of a fitting non-rare clique vertex of~$C'$ (in this application only the set~$C'$ is used instead of~$C$).
Phase~$1$ is modified as follows: 
Let~$L_j\subseteq C'\setminus \RR$ be the set of vertices which can be used to replace the $j$th gap-symbol~$g_j$ (recall that they are ordered according to their position in~$\rmf_{P_i'}$).
We replace~$g_j$ by the vertex with smallest index in~$L_j\setminus(\bigcup_{j'<j}L_{j'})$, that is, we use the vertex with smallest index in~$L_j$ not already used to replace an earlier gap-symbol.
Let~$Q_i'$ be the resulting $|\rmf_{P_i'}|$-path in~$G'$.
%By~$\orderr(P_i')$ we denote the sequence of indices of~$Q_i'$.

We now argue that~$Q_i'=\pOne(\PP_i')$ via induction on the indices:
Assume that~$Q_i'$ and~$\pOne(\PP_i')$ coincide on the first $(j-1)$~vertices.
Now, consider the $j$th vertex~$v_j$ of~$Q_i'$.
If~$v_j\in X\cup \RR$, then since~$\rmf_{Q_i'}=\rmf_{\pOne(\PP_i')}$, $v_j$ is also the $j$th vertex of~$\pOne(\PP_i')$.
Otherwise, by the definition of \texttt{SolLiftTest}, vertex~$v_j$ is the vertex with minimal index in~$C'\setminus\RR$ which is not already used in a smaller index of~$Q_i'$.
Hence,~$v_j$ is also the $j$th vertex of~$\pOne(\PP_i')$.
Thus,~$Q_i'=\pOne(\PP_i')$.

Phase~2 is modified as follows:
Let~$L=C'\setminus(\RR\cup V(Q_i'))$ be the non-rare kernel vertices of the clique which are not already used in path~$Q_i'$.
Now, by definition,~$P_i^*$ is the lexicographically smallest $k'$-path which can be outputted by Phase~$2$ of \texttt{SolLift}.
More precisely, let~$p_1,\ldots, p_t$ be the positions in which vertices of~$L$ can be inserted into~$Q_i'$ to obtain an $k'$-path, according to Phase~2 of \texttt{SolLift}.
By~$a_j$ we denote the index of the vertex in~$Q_i'$ which is right of position~$p_j$ (if~$p_t$ is the last position in~$Q_i'$, that is, if no vertex to the right exists, then we set~$a_t=\infty$).
Let~$L_j=\{v\in L: \text{ index of } v \text{ is smaller than } a_j\}\setminus\{\bigcup_{j'<j}L_{j'}\}$ be the vertices which which have a smaller index than~$a_j$, the index of the right vertex of position~$p_j$ without the vertices which can be inserted in a smaller position.
Note that some~$L_{j_2}$ might be empty, if for example~$a_{j_1}<a_{j_2}$ for some~$j_1< j_2$.
Next, let~$\tau_j$ be the ordering of~$L_j$ returning the smallest sequence.
Now,~$P_i^*$ is the path obtained by adding all vertices of~$L_j$ in order~$\tau_j$ at position~$p_j$ to~$Q_i'$ starting with~$j=1$ and then increasing~$j$ by one in each step.
We stop this process if the resulting path has length~$k'$.
This means that for the last added~$L_j$ we might only add a subset of~$L_j$.
%By~$\orderrr(P_i')$ we denote the sequence of indices of~$P_i^*$.

\textbf{Running Time.}
Observe that for each gap-symbol the set of candidate vertices of~$C'$ can be computed in $\OO(|C'|)\in\OO(n)$~time.
Since there are at most $2\cdot |X|\in\OO(\cvd(G))$~gap-symbols in~$P_i'$, Phase~1 requires $\OO(n\cdot\cvd(G))$~time.
For Phase~2, the set~$L=C'\setminus(\RR\cup V(Q_i'))$ of vertices which can be added into positions~$p_1,\ldots,p_t$, where~$Q_i'$ is the result from Phase~1, can be computed in $\OO(k')\in\OO(n)$~time.
Observe that, by exploiting the ordering~$\prec$ on the vertices of~$G$, we can assign each vertex an unique index of~$[n]$ and thus set~$L$ can be sorted ascendingly according to the indices in linear time.
Furthermore, positions~$p_1,\ldots, p_t$ can be determined in $\OO(k')\in\OO(n)$~time.
Thus, the lexicographically smallest $k'$-path based on~$Q_i'$ can be computed in $\OO(n\cdot \cvd(G))$~time.
\end{proof}

It remains to show that~$P_i^*=\pTwo(\PP_i')$.
Since~$P_i',\pTwo(\PP_i')\in\PP_i'$, both $k'$-paths use the same vertices of~$X\cup \RR$ in the same order.
Let~$v_j$ be the $j$th of these vertices.
We show that~$P_i^*=\pTwo(\PP_i')$ via induction on~$j$.
For simplicity, we assume that both~$P_i^*$ and~$\pTwo(\PP_i')$ end with a vertex of~$X\cup\RR$; the other case can be handled similar.
Assume the statement is already shown for the $(j-1)$th such vertex and let~$L_{j-1}\subseteq C'\setminus \RR$ be the vertices already contained in both~$Q_i'$ and~$\pTwo(\PP_i')$.
In \texttt{SolLiftTest} all vertices in~$C'\setminus (\RR\cup L_{j-1})$ are inserted in increasing order of their indices until a vertex with higher index than the $j$th vertex of~$(X\cup\RR)\cap P_i'$ is reached.
Clearly, this leads to the lexicographically smallest sequence.
Thus,~$P_i^*=\pTwo(\PP_i')$.

%\paragraph{Equivalence Classes~$\PP$ and~$\PP'$ are identical.}

\begin{lemma}
\label{lemma:eqv-class-dist-to-clique-same}
Let $\PP$ and $\PP'$ denote the equivalence classes of $k$-paths of $G$ and $k'$-paths of $G'$ respectively.
Then, $|\PP| = |\PP'|$.
\end{lemma}

\begin{proof}
We show that there is a one-to-one correspondence between equivalence classes~$\PP_i$ and~$\PP_i'$.
First, consider a $k'$-path~$P_i'\in\PP_i'$.
By running algorithm \texttt{SolLift} on~$P_i'$ we obtain at least one $k$-path~$P_i$ with~$\rmf_{P_i'}=\rmf_{P_i}$.
Hence, the statement is verified.
Second, consider a $k$-path~$P_i\in\PP_i$.
We use \texttt{SolLiftTest} on input~$P_i$ to obtain a $k'$-path~$P_i'$ with~$\rmf_{P_i}=\rmf_{P_i}$:
For \texttt{SolLiftTest} it is not important whether its input is a $k$-path or a $k'$-path; in both cases the path is reduced to its signature~$\rmf(\cdot)$ and then a $k'$-path~$P_i'\in\PP_i'$ is returned.
Also note that each vertex in~$\rmf_{P_i'}$ is either rare or contained in the modulator.
Furthermore, by definition there are sufficiently many frequent vertices in the clique~$C'$ to replace each gap-symbol in~$\rmf_{P_i'}$ by a frequent vertex of~$C'$.
Finally, by definition,~$C'$ contains sufficiently many vertices such that in Phase~2 of \texttt{SolLiftTest} a $k'$-path~$P_i'$ can be computed.
Note that here we exploited that~$k'=4\cdot\ell$ if~$k>4\cdot\ell$ and~$k'=k$ if~$k\le 4\cdot\ell$, which implies that~$k'\ge |\rmf_{P'}|$ for each $k$-path~$P'$ in~$G$.

Thus, there is a one-to-one correspondence between the classes~$\PP_i$ and~$\PP_i'$.
\end{proof}

\paragraph{\texttt{SolLiftNS}, a variant of \texttt{SolLift} for a non suitable path~$P_i'\in\PP_i'$.}
The aim of \texttt{SolLiftNS} is to output exactly one $k$-path of~$\PP_i$, the equivalence class corresponding to~$\PP_i'$.
For this, the input of \texttt{SolLiftNS} is~$P_i'$ and not~$\rmf_{P_i'}$.
Hence, Phase~1 is not necessary; instead \texttt{SolLiftNS} only consists of an adaptation of Phase~2 of \texttt{SolLift}:
Let~$p^*$ be the last position of~$P_i'$ in which non-rare clique vertices can be added without changing the signature.
Furthermore, let~$Z= C\setminus (\RR\cup V(P_i'))$ be the set of non-rare clique vertices which can be added to~$P_i'$, let~$\pi$ be the ascending ordering of~$Z$, and let~$Z'\subseteq Z$ be the first $k-k'$~vertices of~$Z$ according to~$\pi$.
The vertices in~$Z'$ are added into~$p^*$ according to~$\pi$.
Note that the result is a $k$-path in~$G$ contained in the equivalence class~$\PP_i$.

By \texttt{SolLiftNS}$(P_i')$ we denote the unique $k$-path of~$G$ which is outputted when $k'$-path~$P_i'$ is given as an input for algorithm \texttt{SolLiftNS}.

\paragraph{\texttt{SolLiftS}, a variant of \texttt{SolLift} for a suitable path~$P_i'\in\PP_i'$.}
The aim of \texttt{SolLiftS} is to output all $k$-paths of~$\PP_i$ which are not outputted by \texttt{SolLiftNS} with some non-suitable path of~$\PP_i'$ as input.
\texttt{SolLift} is modified as follows:
Phase~1 is not changed, that is, \texttt{SolLiftS} branches into each possibility to create a $|\rmf_{P_i'}|$-path~$Q_i'$ by replacing each gap-symbol in~$\rmf_{P_i'}$ by a non-rare clique vertex.
Phase~2 is divided into two parts:
In Phase~2.1, all possibilities for creating a $k'$-path~$P_i^*$ from~$Q_i'$ are considered.
This is done similar as Phase~2 in \texttt{SolLift}; but now only $k'-|\rmf_{P_i'}|$~vertices are added instead of~$k-|\rmf_{P_i'}|$.
For Phase~2.2, let~$P^*=\texttt{SolLiftNS}(P_i^*)$ if~$P_i^*$ is a $k'$-path in~$G'$, and otherwise let~$P^*=\emptyset$.
Now, Phase~2.2 works similar to Phase~2.1 (and Phase~2 of \texttt{SolLift}): 
All possibilities of adding $k-k'$~vertices of~$C\setminus(\RR\cup V(P_i^*))$ in any order in each distribution to positions of~$P_i^*$ without changing the signature (as described in \texttt{SolLift}) are considered.
Let~$P_i$ be one such path.
If~$P_i=P^*$, then nothing is outputted, and if~$P_i\ne P^*$, then~$P_i$ is outputted.

\paragraph{Main Theorem.}
We now show that our solution lifting algorithm is correct and achieves polynomial delay.

%\begin{theorem}
%\label{thm-k-path-cvd}
%{\enumkpath} parameterized by $\cvd(G)$ admits a $\OO(n\cdot (\log(n)+\cvd(G)))$-delay enumeration kernelization with $\OO(\cvd(G)^3)$~vertices.
%\end{theorem}

{\thmThree*}

\begin{proof}
\textbf{Size bound.} The size bound follows from \Cref{lem-kernel-dist-clique}.

\textbf{Correctness.} 
Next, we show that each $k$-path~$P$ of~$G$ is enumerated exactly once.
In challenge~$1$ we verified that~$\PP_i\ne\emptyset$ if and only if~$\PP_i'\ne\emptyset$.
Thus, for each $k$-path in~$G$ there exists at least one $k'$-path in~$G'$ having the same signature.
Furthermore, algorithms \texttt{SolLiftS} and \texttt{SolLiftNS} on an input path~$P_i'\in\PP_i'$ only output $k$-paths of~$\PP_i$.
By \texttt{SolLiftNS}$(P_i')$ we denote the output of \texttt{SolLiftNS} for a non suitable path~$P_i'$.

For two different non suitable paths~$P_i',P_i''\in\PP_i'$, we have \texttt{SolLiftNS}$(P_i')\ne \texttt{SolLiftNS}(P_i'')$:
Assume towards a contradiction that \texttt{SolLiftNS}$(P_i')= \texttt{SolLiftNS}(P_i'')$.
Since~$\rmf_{P_i'}=\rmf_{P_i''}$, the positions~$p_1, \ldots, p_t$ in which non-rare and not already used clique vertices can be added is identical. 
Thus, also the last such position~$p^*=p_t$ is identical.
In this position \texttt{SolLiftNS} adds the $k-k'$~vertices with smallest index in~$C\setminus (\RR\cup V(P_i'))$ in ascending order.
Since also these vertices are identical and added in the same order, we obtain that~$P_i'=P_i''$, a contradiction.

Hence, the $k$-paths of~$\PP_i$ outputted by \texttt{SolLiftNS} with input of non suitable $k'$~paths of~$\PP_i'$ are pairwise distinct.
Furthermore, for the unique suitable $k'$-path of~$\PP_i'$, \texttt{SolLiftS} outputs all $k$-paths of~$\PP_i$ which are not already outputted by any call of \texttt{SolLiftNS}.
Thus, each $k$-path of~$G$ is enumerated exactly once.

\textbf{Delay.}
Let~$P_i'\in\PP_i'$ be some $k'$-path.
Before we compute $k$-paths having the same signature as~$P_i'$, we need to check with \texttt{SolLiftTest} whether~$P_i'$ is suitable or not which can be checked in $\OO(n\cdot\cvd(G))$~time, according to \Cref{lemma:algo-suitable-check-clique}.

Now, depending on whether~$P_i'$ is suitable or not we use \texttt{SolLiftS} or \texttt{SolLiftNS}, respectively.
First, we consider the case that~$P_i'$ is not suitable.
The last position~$p^*$ in which non-rare clique vertices can be added without changing the signature can be computed in $\OO(k')\in\OO(n)$~time.
Furthermore, the set~$Z=C\setminus(\RR\cup V(P_i'))$ of vertices we can add, can also be computed in $\OO(n)$~time and they can be sorted ascendingly according to their indices in linear time since we can use the ordering~$\prec$ to assign each vertex a unique index of~$[n]$.
Now, the $k-k'$~vertices with lowest indices of~$Z$ can be inserted into~$p^*$ in $\OO(n)$~time.
Thus, \texttt{SolLiftNS} outputs a $k$-path in $\OO(n)$~time.

Second, we consider the case that~$P_i'$ is suitable, that is, we need to bound the delay of \texttt{SolLiftS}.
Initially, we bound the time to compute the first/next/last $|\rmf_{P_i'}|$-path which is a valid output of Phase~1 of \texttt{SolLiftS}.
In $\OO(n)$~time all vertices of~$C$ are sorted ascendingly according to their indices.
Similar to \texttt{SolLiftTest}, for each gap-symbol in $\rmf_{P_i'}$ the set of vertices we can use to replace them can be computed in $\OO(n)$~time per gap-symbol, yielding a overall running time of~$\OO(n\cdot\cvd(G))$ to compute these sets initially.
Each such candidate set is sorted ascendingly according to the indices.
Subsequently, all $|\rmf_{P_i'}|$-paths having the same signature as~$P_i'$ can be enumerated with a delay of~$\OO(n\cdot\cvd(G))$:
Define a lexicographical ordering of such $|\rmf_{P_i'}|$-paths according to the orderings of the candidates sets, where the first criteria is the index of that vertex which replaced the first gap-symbol and so on.
Thus, all such $|\rmf_{P_i'}|$-paths can be enumerated with $\OO(n\cdot\cvd(G))$~delay.

It remains to bound the delay of Phase~2.
The analysis of Phases~2.1 and~2.2 is similar, so we only analyze the delay for Phase~2.1; the only difference is that the output of Phase~2.1 is a $k'$-path and that the output of Phase~2.2 is a $k$-path.
For this, let~$Q_i'$ be a $|\rmf_{P_i'}|$-path with~$\rmf_{Q_i'}=\rmf_{P_i'}$ outputted after Phase~1.
Furthermore, let~$L=C\setminus(\RR\cup V(Q_i'))$ be the set of non-rare clique vertices which can be added to~$Q_i'$.
Clearly,~$L$ is computable in $\OO(n)$~time and can be sorted ascendingly in linear time.
Also, positions~$p_1, \ldots, p_t$ in which the vertices of~$L$ can be added without changing the signature of~$Q_i'$ can be determined in $\OO(k')\in\OO(n)$~time.
Recall that initially, \texttt{SolLiftS} starts with the lexicographically smallest subset~$L'$ of~$L$ having $k'-|Q_i'|$~vertices and in each further iteration \texttt{SolLiftS} considers the (unique) successor subset~$L'$ according to this lexicographical order.
Clearly, the first/next such subset~$L'$ can be computed in $\OO(n)$~time.
Also, recall that \texttt{SolLiftS} starts with the lexicographical ordering~$\pi_1$ of~$L'$ and all orderings are sorted according to the lexicographical ordering of the resulting sequence.
Consequently, in each further iteration \texttt{SolLiftS} chooses the (unique) successor ordering~$\pi_j$.
Clearly, the first/next such ordering can be computed in $\OO(n)$~time.
Also recall that  all possibilities of adding the vertices~$L'$ in order~$\pi$ to positions~$p_1,\ldots, p_t$ are described by a vector~$\vec{v}$ of length~$t$ and~$\mathcal{V}$ is the set of all such vectors and~$\chi$ is an arbitrary but fixed ordering of~$\mathcal{V}$.
For each fixed~$L'$ and fixed~$\pi$, \texttt{SolLiftS} considers all vectors in~$\mathcal{V}$ according to~$\chi$.
Clearly, the next such distribution of the vertices of~$L'$ according to~$\chi$ can be computed in $\OO(n)$~time.
Thus, in overall $\OO(n\cdot \log(n))$~time the first/next/last $k'$-path~$P_i^*$ is computed.
Hence, the delay of Phase~2.1 (and also Phase~2.2) is $\OO(n)$.

Recall that if~$P_i^*$ is a path in~$G'$, not each $k$-path enumerated by Phase~2.2 is outputted: the unique $k$-path enumerated by \texttt{SolLiftNS} with input~$P_i^*$ is not outputted.
Thus, if~$P_i^*$ is a path in~$G'$, we check for each enumerated $k$-path~$P_i$ of \texttt{SolLiftS} whether~$P_i=\texttt{SolLiftNS}(P_i^*)$ and only output~$P_i$ if~$P_i\ne \texttt{SolLiftNS}(P_i^*)$.
We now show that if~$P_i$ is not outputted (since~$P_i$ is the output of one call of \texttt{SolLiftNS}), then the next $k$-path~$P_i'$ computed by \texttt{SolLiftS} is not the output of any call of \texttt{SolLiftNS} and is thus outputted.
If this property is verified, we can safely conclude that we only get an additional factor of~2 for the delay and thus we have shown that the overall delay of our algorithm is~$\OO(n\cdot \cvd(G))$.

It remains to show this property.
%Note that we can safely assume that the clique~$C$ is much larger the the kernel bound; since otherwise the instance itself is the desired kernel.
Assume that~$P_i$ is not outputted since~$P_i$ is the output of one call of \texttt{SolLiftNS}.
Recall that in \texttt{SolLiftNS} all non-rare clique vertices are added in the last possible position in ascending order.
Hence, in \texttt{SolLiftS} the last possible vector of~$\mathcal{V}$ according to~$\chi$ was used.
Consequently, the next ordering~$\pi$ of the vertices~$L'$ is considered.
If there is another ordering, the subsequent ordering is not ascending and thus a $k$-path which is not the output of any call of \texttt{SolLiftNS} is produced.
So assume that there is no next ordering~$\pi$.
This implies that~$|L'|=1$ since the ascending order is the first order,
Hence, the algorithm now considers the next vertex in~$C\setminus(\RR\cup V(Q_i'))$.
If such a vertex exists, then the subsequent $k$-path cannot be the output of a call of \texttt{SolLiftNS} and otherwise, if no such vertex exists, the algorithm \texttt{SolLiftS} terminates.
Thus, this property is verified.
\end{proof}

\subsection{Extension to other Path and Cycle Variants}

We now extend our positive result of {\enumkpath} to {\enumkcycle} and the variants where all paths/cycles of length \emph{at least}~$k$ need to be outputted.
We start with {\enumkpathAll}.

\begin{proposition}
\label{prop-clique-k-path-at-least}
{\enumkpathAll} parameterized by $\cvd(G)$ admits an $\OO(n\cdot \cvd(G))$-delay enumeration kernelization with $\OO(\cvd(G)^3)$~vertices.
\end{proposition}
\begin{proof}
The marking scheme, the kernelization, the definition of signature, the definition of the equivalence classes, and algorithm \texttt{SolLift} is unchanged.
Now, the argumentation that an equivalence class~$\PP_i$ of the original instance is non-empty if and only if the corresponding equivalence class~$\PP_i'$ of the kernel is non-empty is slightly different:
By algorithm \texttt{SolLift} for each non-empty class~$\PP_i'$ one obtains at least one $k$-path of~$\PP_i$.
Now, let~$P_i\in\PP_i$ of length at least~$k$.
Algorithm \texttt{SolLiftTest} computes a $k'$-path~$P_i'$ with the same signature.
Hence, if~$\PP_i\ne\emptyset$, then there exists at least one $k'$-path in~$\PP_i'$.
Thus, there is a one-to-one-correspondence between the classes~$\PP_i$ and~$\PP_i'$.
Since for each non-empty class~$\PP_i'$ contains at least one $k'$-path, we can define the suitable path similar:
For~$P_1', P_2'\in\PP_i'$ with~$|P_1'|<|P_2'|$, we let~$P_1'\triangleleft P_2'$ and if~$|P_1'|=|P_2'|$, then we use the existing definition described in \Cref{sec-clique-solve-challenges}.

Next we adapt algorithms \texttt{SolLiftNS} and \texttt{SolLiftS}.
For each $(k'+z)$-path~$P_i'$ in~$G$, where~$z\ge 0$, \texttt{SolLiftNS} adds exactly $k-k'$~clique vertices to~$P_i'$ as described in \Cref{sec-clique-solve-challenges}.
Note that for distinct paths~$P_1'$ and~$P_2'$ of length at least~$k'$ we have \texttt{SolLiftNS}$(P_1')\ne$ \texttt{SolLiftNS}$(P_2')$: if~$P_1'$ and~$P_2'$ have different length, then clearly also the outputs have different length, and if~$P_1'$ and~$P_2'$ have the same length, then the outputs are different as shown in \Cref{thm-k-path-cvd}.
Let~$P_i'$ be the suitable $k'$-path of class~$\PP_i'$.

Next, we adapt \texttt{SolLiftS}.
Phase~1 is not changed, that is, a $|\rmf_{P_i'}|$-path~$Q_i'$ is computed where~$P_i'$ is the suitable path of class~$\PP_i'$.
For Phase~2, recall that $L= C\setminus (\RR\cup V(Q_i'))$ is the set of candidate vertices which can be added in Phase~2.
In \Cref{sec-clique-solve-challenges}, \texttt{SolLiftS} for each~$Z\subseteq L$ of size exactly~$k-|Q_i'|$, enumerated all $k$-paths where exactly the vertices of~$Z$ where added to~$Q_i'$.
We enhance Phase~2 as follows:
Initially, all $k$-paths of~$\PP_i$ are enumerated by using each~$Z$ of size~$k-|Q_i'|$.
Again, in Phase~2.1 we first compute a $k'$-path~$P_i''$ and if~$P_i''$ is a path in~$G'$ we do not enumerate the output of \texttt{SolLiftNS}$(P_i'')$ again. 
Then, we do the same for each~$Z$ of size~$k-|Q_i'|+1$ and continue until~$Z=L$.
Again, in Phase~2.1 we first compute a $k'+1$-path~$P_i''$ (and so on) and if~$P_i''$ is a path in~$G'$ we do not enumerate the output of \texttt{SolLiftNS}$(P_i'')$ again.
Thus, in Phase~2.2 always $k-k'$~vertices are added. 
In other words, initially all $k$-paths of~$\PP_i$ are enumerated, then all $(k+1)$-paths of~$\PP_i$ until the longest paths of~$\PP_i$ are enumerated and each time we skip the outputs of \texttt{SolLiftNS}.

The correctness is similar:
\texttt{SolLiftS} and \texttt{SolLiftNS} only output paths of~$G$ from the same equivalence class as the input.
Furthermore, as argued above, the outputs for distinct inputs of \texttt{SolLiftNS} are distinct.
Also, \texttt{SolLiftS} enumerated all paths of length at least~$k$ in~$G$ from the equivalence class of the input which are not already outputted by any call of \texttt{SolLiftNS}.
Thus, each solution is enumerated exactly once.

The analysis of the delay is similar.
The only difference is that \texttt{SolLiftS} now has several iterations to output paths of length at least~$k$; this does not change the delay.
\end{proof}

Next, we adapt our result for $k$-cycles.

\begin{proposition}
\label{prop-clique-k-cycle}
{\enumkcycle} parameterized by $\cvd(G)$ admits an $\OO(n\cdot \cvd(G))$-delay enumeration kernelization with $\OO(\cvd(G)^3)$~vertices.
\end{proposition}
\begin{proof}
The marking scheme and the kernelization work analog.
The definition of the signature is slightly adjusted since a cycle has no designated starting vertex:
Let~$C$ be a $k$-cycle and let~$v^*\in C\cap(X\cup\RR)$ be the vertex of the modulator or rare with minimal index.
Clearly, $v^*$ has 2~neighbors~$v_1$ and~$v_2$ in~$C$.
Let~$v'\in\{v_1,v_2\}$ be the vertex with minimal index.
Now,~$C$ is transformed into a $k$-path~$P_C$ which starts with~$v^*$, has~$v'$ as its second vertex and uses the same vertices of~$C$ in the same order.
Note that if~$C\cap(X\cup\RR)=\emptyset$ then~$P_C$ is an empty path.
Afterwards, the definition of suitable and of the equivalence classes is analog.

In the solution lifting algorithm \texttt{SolLift} and its variants we need to treat the case that~$\rmf_{C}$ ends with a gap-symbol~$g_C$ or~$g_N$ differently:
Since~$C$ is a cycle, this last gap-symbol indicates a neighbor of the first vertex~$v^*$.
Hence, when replacing that gap-symbol this has to be considered.
Afterwards, the proof is analog.
\end{proof}

Now, the result for {\enumkcycleAll} follows by using the adaptations for paths of length at least~$k$ (\Cref{prop-clique-k-path-at-least}) and $k$-cycles (\Cref{prop-clique-k-cycle}).

\begin{corollary}
\label{cor-clique-k-cycle-at-least}
{\enumkcycleAll} parameterized by $\cvd(G)$ admits a $\OO(n\cdot \cvd(G))$-delay enumeration kernelization with $\OO(\cvd(G)^3)$~vertices.
\end{corollary}
\else
\begin{comment}

Our techniques can also be used to show the following.

\begin{proposition}[$\star$]
\label{prop-clique-k-path-at-least-short}
{\enumkpathAll}, {\enumkcycle}, and {\enumkcycleAll} parameterized by $\cvd(G)$ admits a $\OO(n\cdot \cvd(G))$-delay enumeration kernelization with $\OO(\cvd(G)^3)$~vertices.
\end{proposition}

\end{comment}

\fi

\section{Conclusion}
\label{sec:conclusion}

We initiated a systematic study of enumeration kernels for the problems of finding paths (and cycles) with (at least) $k$~vertices, leading to polynomial-delay enumeration kernels of polynomial size (\pdpsk{}s) for the parameters ${\vc}$, ${\diss}$, $r$-${\coc}$, and ${\cvd}$.
It is open whether our positive results can be extended to a \pdpsk{} for the parameters distance to cluster, feedback vertex set, or neighborhood diversity (for which even a polynomial-size kernel for the decision problem seems to be open).
%  Recall that even if we consider the size of a largest connected component (${\mcp}$) as the parameter, we cannot expect to get a polynomial-size (enumeration) kernel.
% This parameter is larger than the vertex integrity, automatically ruling out the existence of polynomial-size (enumeration) kernelization for treedepth and thus also treewidth.
More generally, it is open to develop \pdpsk{}s for other well-studied problems.

%\newpage 
%
%
%
%
% ---- Bibliography ----
%
% BibTeX users should specify bibliography style 'splncs04'.
% References will then be sorted and formatted in the correct style.
%
 \bibliographystyle{splncs04}

\begin{thebibliography}{10}
\providecommand{\url}[1]{\texttt{#1}}
\providecommand{\urlprefix}{URL }
\providecommand{\doi}[1]{https://doi.org/#1}

\bibitem{AdamsonGM24}
Adamson, D., Gawrychowski, P., Manea, F.: Enumerating $m$-length walks in
  directed graphs with constant delay. In: Proceedings of the 16th Latin
  American Theoretical INformatics Symposium ({LATIN}~'24). Lecture Notes in
  Computer Science, vol. 14578, pp. 35--50. Springer (2024)

\bibitem{BirmeleFGMPRS13}
Birmel{\'{e}}, E., Ferreira, R.A., Grossi, R., Marino, A., Pisanti, N., Rizzi,
  R., Sacomoto, G.: Optimal listing of cycles and $st$-paths in undirected
  graphs. In: Khanna, S. (ed.) Proceedings of the Twenty-Fourth Annual
  {ACM-SIAM} Symposium on Discrete Algorithms ({SODA}~'13). pp. 1884--1896.
  {SIAM} (2013)

\bibitem{BjorklundHKK17}
Bj{\"{o}}rklund, A., Husfeldt, T., Kaski, P., Koivisto, M.: Narrow sieves for
  parameterized paths and packings. J. Comput. Syst. Sci.  \textbf{87},
  119--139 (2017)

\bibitem{BlazejCKSSV22}
Blazej, V., Choudhary, P., Knop, D., Schierreich, S., Such{\'{y}}, O., Valla,
  T.: On polynomial kernels for traveling salesperson problem and its
  generalizations. In: Proceedings of the 30th Annual European Symposium on
  Algorithms ({ESA}~'22). LIPIcs, vol.~244, pp. 22:1--22:16. Schloss Dagstuhl -
  Leibniz-Zentrum f{\"{u}}r Informatik (2022)

\bibitem{BodlaenderDFH09}
Bodlaender, H.L., Downey, R.G., Fellows, M.R., Hermelin, D.: On problems
  without polynomial kernels. J. Comput. Syst. Sci.  \textbf{75}(8),  423--434
  (2009)

\bibitem{BodlaenderJK13}
Bodlaender, H.L., Jansen, B.M.P., Kratsch, S.: Kernel bounds for path and cycle
  problems. Theor. Comput. Sci.  \textbf{511},  117--136 (2013)

\bibitem{BohmovaHMPSS18}
B{\"{o}}hmov{\'{a}}, K., H{\"{a}}fliger, L., Mihal{\'{a}}k, M., Pr{\"{o}}ger,
  T., Sacomoto, G., Sagot, M.: Computing and listing $st$-paths in public
  transportation networks. Theory Comput. Syst.  \textbf{62}(3),  600--621
  (2018)

\bibitem{CreignouMMSV17}
Creignou, N., Meier, A., M{\"{u}}ller, J., Schmidt, J., Vollmer, H.: Paradigms
  for parameterized enumeration. Theory Comput. Syst.  \textbf{60}(4),
  737--758 (2017)

\bibitem{CyganFKLMPPS15}
Cygan, M., Fomin, F.V., Kowalik, L., Lokshtanov, D., Marx, D., Pilipczuk, M.,
  Pilipczuk, M., Saurabh, S.: Parameterized Algorithms. Springer (2015)

\bibitem{Damaschke06}
Damaschke, P.: Parameterized enumeration, transversals, and imperfect phylogeny
  reconstruction. Theor. Comput. Sci.  \textbf{351}(3),  337--350 (2006)

\bibitem{Diestel-Book}
Diestel, R.: Graph Theory, 4th Edition, Graduate texts in mathematics,
  vol.~173. Springer (2012)

\bibitem{DowneyF13}
Downey, R.G., Fellows, M.R.: Fundamentals of Parameterized Complexity. Texts in
  Computer Science, Springer (2013)

\bibitem{Eppstein98}
Eppstein, D.: Finding the $k$ shortest paths. {SIAM} J. Comput.
  \textbf{28}(2),  652--673 (1998)

\bibitem{FellowsLMMRS09}
Fellows, M.R., Lokshtanov, D., Misra, N., Mnich, M., Rosamond, F.A., Saurabh,
  S.: The complexity ecology of parameters: An illustration using bounded max
  leaf number. Theory Comput. Syst.  \textbf{45}(4),  822--848 (2009)

\bibitem{FerreiraGRSS14}
Ferreira, R.A., Grossi, R., Rizzi, R., Sacomoto, G., Sagot, M.: Amortized
  \~{O}({\(\vert\)}{V}{\(\vert\)})-delay algorithm for listing chordless cycles
  in undirected graphs. In: Proceedings of the 22nd European Symposium on
  Algorithms ({ESA}~'14). Lecture Notes in Computer Science, vol.~8737, pp.
  418--429. Springer (2014)

\bibitem{FominLLSTZ19}
Fomin, F.V., Le, T., Lokshtanov, D., Saurabh, S., Thomass{\'{e}}, S., Zehavi,
  M.: Subquadratic kernels for implicit 3-hitting set and 3-set packing
  problems. {ACM} Trans. Algorithms  \textbf{15}(1),  13:1--13:44 (2019)

\bibitem{FominLPS16}
Fomin, F.V., Lokshtanov, D., Panolan, F., Saurabh, S.: Efficient computation of
  representative families with applications in parameterized and exact
  algorithms. J. {ACM}  \textbf{63}(4),  29:1--29:60 (2016)

\bibitem{FominLPS17}
Fomin, F.V., Lokshtanov, D., Panolan, F., Saurabh, S.: Representative families
  of product families. {ACM} Trans. Algorithms  \textbf{13}(3),  36:1--36:29
  (2017)

\bibitem{FominLSZ19}
Fomin, F.V., Lokshtanov, D., Saurabh, S., Zehavi, M.: {Kernelization: Theory of
  Parameterized Preprocessing}. Cambridge University Press (2019)

\bibitem{GolovachKKL22}
Golovach, P.A., Komusiewicz, C., Kratsch, D., Le, V.B.: Refined notions of
  parameterized enumeration kernels with applications to matching cut
  enumeration. J. Comput. Syst. Sci.  \textbf{123},  76--102 (2022)

\bibitem{gomes2024matchingmulticutalgorithmscomplexity}
Gomes, G.C.M., Juliano, E., Martins, G., dos Santos, V.F.: Matching (multi)cut:
  Algorithms, complexity, and enumeration. In: Proceedings of the 19th
  International Symposium on Parameterized and Exact Computation ({IPEC}~'24).
  LIPIcs, vol.~321, pp. 25:1--25:15. Schloss Dagstuhl - Leibniz-Zentrum
  f{\"{u}}r Informatik (2024)

\bibitem{GoyalMPZ15}
Goyal, P., Misra, N., Panolan, F., Zehavi, M.: Deterministic algorithms for
  matching and packing problems based on representative sets. {SIAM} J.
  Discret. Math.  \textbf{29}(4),  1815--1836 (2015)

\bibitem{HKSS13}
Ho{\`{a}}ng, C.T., Kaminski, M., Sawada, J., Sritharan, R.: Finding and listing
  induced paths and cycles. Discrete Applied Mathematics  \textbf{161}(4-5),
  633--641 (2013)

\bibitem{Jansen17}
Jansen, B.M.P.: Turing kernelization for finding long paths and cycles in
  restricted graph classes. J. Comput. Syst. Sci.  \textbf{85},  18--37 (2017)

\bibitem{JansenPW19}
Jansen, B.M.P., Pilipczuk, M., Wrochna, M.: Turing kernelization for finding
  long paths in graph classes excluding a topological minor. Algorithmica
  \textbf{81}(10),  3936--3967 (2019)

\bibitem{JansenS23}
Jansen, B.M.P., van~der Steenhoven, B.: Kernelization for counting problems on
  graphs: Preserving the number of minimum solutions. In: Proceedings of the
  18th International Symposium on Parameterized and Exact Computation
  ({IPEC}~'23). LIPIcs, vol.~285, pp. 27:1--27:15. Schloss Dagstuhl -
  Leibniz-Zentrum f{\"{u}}r Informatik (2023)

\bibitem{KobayashiKW21}
Kobayashi, Y., Kurita, K., Wasa, K.: Polynomial-delay enumeration of large
  maximal matchings. CoRR  \textbf{abs/2105.04146} (2021),
  \url{https://arxiv.org/abs/2105.04146}

\bibitem{KobayashiKW22}
Kobayashi, Y., Kurita, K., Wasa, K.: Polynomial-delay and polynomial-space
  enumeration of large maximal matchings. In: Proceedings of the 48th
  International Workshop on Graph-Theoretic Concepts in Computer Science
  ({WG}~'22). Lecture Notes in Computer Science, vol. 13453, pp. 342--355.
  Springer (2022)

\bibitem{LokshtanovMSbhZ24}
Lokshtanov, D., Misra, P., Saurabh, S., Zehavi, M.: Kernelization of counting
  problems. In: Proceedings of the 15th Innovations in Theoretical Computer
  Science Conference ({ITCS}~'24). LIPIcs, vol.~287, pp. 77:1--77:23. Schloss
  Dagstuhl - Leibniz-Zentrum f{\"{u}}r Informatik (2024)

\bibitem{LPRS17}
Lokshtanov, D., Panolan, F., Ramanujan, M.S., Saurabh, S.: Lossy kernelization.
  In: Proceedings of the 49th Annual {ACM} {SIGACT} Symposium on Theory of
  Computing ({STOC}~'17). pp. 224--237. {ACM} (2017)

\bibitem{MajumdarR23}
Majumdar, D.: Enumeration kernels of polynomial size for cuts of bounded
  degree. CoRR  \textbf{abs/2308.01286} (2023),
  \url{https://doi.org/10.48550/arXiv.2308.01286}

\bibitem{Meier20}
Meier, A.: Parametrised enumeration (2020). \doi{10.15488/9427},
  \url{https://doi.org/10.15488/9427}, {H}abilitation thesis, Leibniz
  Universit{\"a}t Hannover

\bibitem{NRT05}
Nishimura, N., Ragde, P., Thilikos, D.M.: Parameterized counting algorithms for
  general graph covering problems. In: Proceedings of the 9th International
  Workshop on Algorithms and Data Structures ({WADS}~'05). Lecture Notes in
  Computer Science, vol.~3608, pp. 99--109. Springer (2005)

\bibitem{PapadimitriouY82}
Papadimitriou, C.H., Yannakakis, M.: The complexity of restricted spanning tree
  problems. J. {ACM}  \textbf{29}(2),  285--309 (1982)

\bibitem{RizziSS14}
Rizzi, R., Sacomoto, G., Sagot, M.: Efficiently listing bounded length
  $st$-paths. In: Proceedings of the 25th International Workshop on
  Combinatorial Algorithms ({IWOCA}~'14). Lecture Notes in Computer Science,
  vol.~8986, pp. 318--329. Springer (2014)

\bibitem{Savage82}
Savage, C.D.: Depth-first search and the vertex cover problem. Inf. Process.
  Lett.  \textbf{14}(5),  233--237 (1982)

\bibitem{Thomasse10}
Thomass{\'{e}}, S.: A $4 k^2$ kernel for feedback vertex set. {ACM} Trans.
  Algorithms  \textbf{6}(2),  32:1--32:8 (2010)

\bibitem{Thur07}
Thurley, M.: Kernelizations for parameterized counting problems. In:
  Proceedings of the 4th International Conference on Theory and Applications of
  Models of Computation ({TAMC}~'07). Lecture Notes in Computer Science,
  vol.~4484, pp. 703--714. Springer (2007)

\bibitem{Wasa16arXiv}
Wasa, K.: Enumeration of enumeration algorithms. CoRR  \textbf{abs/1605.05102}
  (2016), \url{http://arxiv.org/abs/1605.05102}

\end{thebibliography}

\end{document}